\documentclass{article}

\usepackage{amssymb}
\usepackage{amsmath, amsthm, mathtools, xcolor,bbm,booktabs,multicol,enumitem}
\usepackage[authoryear, round]{natbib}
\usepackage[margin=1in]{geometry}
\usepackage{graphicx} % Required for inserting images

\usepackage{subcaption}
\usepackage[
	breaklinks,
	colorlinks=true,
	pagebackref=false,
	pdfauthor={},%
	pdftitle={},% 
	citecolor=blue,
	linkcolor=blue,
	urlcolor=blue, hypertexnames=false
	]{hyperref}

\usepackage[disable]{todonotes}

\usepackage{placeins}

\usepackage{multicol}
\usepackage{multirow}
\usepackage{enumitem}
\usepackage{centernot}
\usepackage[bottom]{footmisc}

\allowdisplaybreaks

\usepackage{multibib}
\newcites{supp}{Supplementary References}
\newcommand{\beginsupplement}{
	\setcounter{page}{1}
	\renewcommand{\thepage}{S\arabic{page}}
	\setcounter{section}{0}
	\renewcommand{\thesection}{S\arabic{section}}
	\setcounter{equation}{0}
	\renewcommand{\theequation}{S\arabic{equation}}
	\setcounter{table}{0}
	\renewcommand{\thetable}{S\arabic{table}}
	\setcounter{figure}{0}
	\renewcommand{\thefigure}{S\arabic{figure}}}
\newcommand{\CRPS}{\operatorname{CRPS}}

\newcommand{\dd}{\,\mathrm{d}}

\newcommand{\one}{\mathbbm{1}}
\newcommand{\sgn}{\operatorname{sgn}}

\newcommand{\D}{\operatorname{D}}
\newcommand{\CD}{\operatorname{CD}}
\newcommand{\AV}{\operatorname{AVM}}
\newcommand{\unif}{\mathcal{U}}
\newcommand{\norm}{\mathcal{N}}
\newcommand{\dirac}[1]{\delta_{#1}}
\newcommand{\R}{\mathbb{R}}

\newcommand{\shift}{\operatorname{Shift}}
\newcommand{\disp}{\operatorname{Disp}}

\newcommand{\WD}{\operatorname{WD}}
\newcommand{\WDp}{\operatorname{WD}_p}

\newcommand{\mudiff}{\widetilde{\mu}}
\newcommand{\sigmadiff}{\widetilde{\sigma}}
\newcommand{\sigmaavg}{\widetilde{\rho}}

\usepackage{lipsum}

\newtheorem{theorem}{Theorem}[section]
\newtheorem{prop}[theorem]{Proposition}

\newtheorem{remark}[theorem]{Remark}
\newtheorem{lemma}[theorem]{Lemma}
\newtheorem{example}[theorem]{Example}
\newtheorem{conj}[theorem]{Conjecture}

\setlist[enumerate,1]{label={(\alph*)}} % enumerate as (a), (b), ... by default

\newcommand\defaultfigwidth{0.48}

\newcommand\blfootnote[1]{%
	\begingroup
	\renewcommand\thefootnote{}\footnote{#1}%
	\addtocounter{footnote}{-1}%
	\endgroup
}

\title{Shift-Dispersion Decompositions of \\ Wasserstein and Cram\'er Distances}

\author{
	Johannes Resin\thanks{Faculty of Economics and Business, Goethe University Frankfurt, Theodor-W.-Adorno-Platz 4, D--60323 Frankfurt am Main, Germany, and Heidelberg Institute for Theoretical Studies (HITS), \href{mailto:resin@econ.uni-frankfurt.de}{resin@econ.uni-frankfurt.de} and \href{mailto:dimitriadis@econ.uni-frankfurt.de}{dimitriadis@econ.uni-frankfurt.de}.}
\and 
	Daniel Wolffram\thanks{Institute of Statistics, Karlsruhe Institute of Technology (KIT), Bl\"ucherstr. 17, D--76185 Karlsruhe, Germany, and Heidelberg Institute for Theoretical Studies (HITS), \href{daniel.wolffram@kit.edu}{daniel.wolffram@kit.edu} and \href{johannes.bracher@kit.edu}{johannes.bracher@kit.edu}.}
\and
	Johannes Bracher$^{\dagger\S}$
\and 
	Timo Dimitriadis$^{*\S}$
}

\begin{document}

\maketitle

\begin{abstract}
	\noindent
	Divergence functions are measures of distance or dissimilarity between probability distributions that serve various purposes in statistics and applications. 
	We propose decompositions of Wasserstein and Cram\'er distances---which compare two distributions by integrating over their differences in  distribution or quantile functions---into directed shift and dispersion components. 
	These components are obtained by dividing the differences between the quantile functions into contributions arising from shift and dispersion, respectively.    
	Our decompositions add information on \emph{how} the distributions differ in a condensed form and consequently enhance the interpretability of the underlying divergences. 
	We show that our decompositions satisfy a number of natural properties and are unique in doing so in location-scale families.
	The decompositions allow us to derive sensitivities of the divergence measures to changes in location and dispersion, and they give rise to weak stochastic order relations that are linked to the usual stochastic and the dispersive order.
	Our theoretical developments are illustrated in two applications, where we focus on forecast evaluation of temperature extremes and on the design of probabilistic surveys in economics.
\end{abstract}

\section{Introduction}

The\blfootnote{\hspace{-0.215cm}$^{\S}$The last two authors contributed equally.} 
task of comparing pairs of probability distributions arises in numerous contexts and has given rise to a wealth of \textit{divergence functions} or \textit{statistical distances} \citep{DD2013}. 
A particularly well-studied divergence is the Wasserstein distance with numerous theoretical developments in machine learning (e.g., \citealt{Frogner2015, arjovsky2017wasserstein}), dependence modelling \citep{Wiesel2022}, distributional regression \citep{chen2023wasserstein}, and model diagnostics \citep{munk1998nonparametric}.
Applications exist in diverse fields including image processing \citep{ni2009local}, biostatistics \citep{SFG21}, and economics \citep{gm1915di, rachev2011probability};
see \citet{Panaretos2019} for a recent review.
The Cram\'er distance is an alternative which shares many attractive properties of the Wasserstein distance, including its symmetry property and the rewarding of closeness \citep{RS16}. It has become a popular metric in forecast evaluation, with applications in the atmospheric sciences \citep{TGG13, RCP2020}, hydrology \citep{Barna2023} and electricity markets \citep{JS19}. The Cram\'er distance has also found use in machine learning as an alternative to the Wasserstein distance in generative adversarial networks \citep{Bellemare2017}.

In this paper, we are concerned with decompositions of statistical distances of Wasserstein and Cram\'er type.
We focus on the real-valued case, where $F$ and $G$ are probability distributions on the real line $\R$, and we identify both distributions with their respective cumulative distribution functions (CDFs). 
In a nutshell, a divergence is a function $\D$ such that $\D(F,G)$ is non-negative for all $F, G$, and equals zero if and only if $F = G$.
The purpose of divergence functions is to reduce the difference between two probability distributions, i.e., two infinite-dimensional objects, to a single number.
Of course, this reduction implies a severe loss of information: 
A divergence function only quantifies the \emph{magnitude of dissimilarity} between two distributions $F$ and $G$, but it hides the \textit{specific nature of the differences}, e.g., whether the main difference is in location or dispersion. 
To shed light on these aspects, we propose novel decompositions of divergence functions $\D$ into four non-negative and interpretable components
\begin{equation}
	\label{Eq:GeneralDecomp}
	\begin{aligned}
		\D(F, G) = \shift_+^{\D} + \shift_-^{\D}
		&+ \disp_+^{\D} + \disp_-^{\D}.
	\end{aligned}
\end{equation}
Here, the \emph{shift components} $\shift_{\pm}^{\D}$ with $\pm \in \{+,-\}$ quantify differences in \emph{location}, while the \emph{dispersion components} $\disp_\pm^{\D}$ measure differences in \emph{variability} between $F$ and $G$.
The signed components (with subscript `$+$' and `$-$') attribute parts of the distance to upwards and downwards shifts, and more or less dispersion of $F$ relative to $G$.
Of course, the components in \eqref{Eq:GeneralDecomp} are functions of the pair of distributions $(F,G)$. For the sake of brevity, however, we will sometimes omit this dependence and use the shorthands $\shift_\pm^{\D} = \shift_\pm^{\D}(F,G)$ and $\disp_\pm^{\D} = \disp_\pm^{\D}(F,G)$. These refer to the components of $\D(F,G)$ between for two generic distributions $F$ and $G$ or a pair of distributions which becomes clear from the context.

Our decompositions in \eqref{Eq:GeneralDecomp} apply to arbitrary (possibly discontinuous) distributions and to divergence measures that allow for certain representations through quantile functions.
We address the aforementioned Wasserstein distance, more specifically, the $p$-th power of the $p$-\emph{Wasserstein distance},
\begin{equation}
	\label{Eq:pWasserstein}
	\WD_p(F, G) = \int_0^1 \big| F^{-1}(\tau) - G^{-1}(\tau) \big|^p \dd \tau
\end{equation}
for $p \in [1,\infty)$.
Here, $F^{-1}$ and $G^{-1}$ denote the quantile functions, given by $F^{-1}(\tau) = \inf\{x \in\R \mid \tau \leq F(x)\}$, $\tau \in [0,1]$ and accordingly for $G^{-1}$, i.e., the left-inverses of the CDFs. The \emph{Cram\'er distance} arises as the square of the special case $p=2$ within the class of $l_p$ distances,
\begin{equation}
	\label{Eq:Lp}
	l_p(F, G) = \left(\int_{-\infty}^\infty |F(x) - G(x)|^p \dd x \right)^{1/p}.
\end{equation}
While its classic definition (as a special case of \eqref{Eq:Lp}) is in terms of CDFs, we provide an alternative representation via quantile functions, which forms the basis for its decomposition. 
Akin to the distances themselves, our proposed decompositions integrate over suitably assigned differences of the quantile functions and as such account for \emph{any} distributional difference between $F$ and $G$.

Section~\ref{sec:Decompositions} introduces the novel decompositions together with extensive intuitive explanations.
A particularly straightforward graphical illustration is available for $p = 1$, in which case the expressions in \eqref{Eq:pWasserstein} and \eqref{Eq:Lp} coincide and are referred to as the \textit{area validation metric} (AVM). The shift and dispersion terms in \eqref{Eq:GeneralDecomp} then arise from simple comparisons of \emph{central intervals} at \emph{coverage levels} $\alpha \in [0,1]$, i.e., intervals spanned by the $(1 \pm \alpha)/2$ 
quantiles of $F$ and $G$.
Roughly speaking, central intervals of differing lengths indicate differences in dispersion, while shifted intervals point to differing locations.
The components in \eqref{Eq:GeneralDecomp} are obtained by integrating over all coverage levels $\alpha \in [0,1]$.

In Sections \ref{Sec:Properties} and \ref{sec:Orders} we provide theoretical arguments that support the particular specifications of our decompositions.

Firstly, the decompositions behave naturally in settings where the distributions $F$ and $G$ are linked through additive shifts, symmetry relations, or a location-scale property. 
We also provide closed-form expressions for the components in the Gaussian case.
Crucially, we prove that the proposed decompositions are \emph{unique} in simultaneously satisfying a number of natural properties for (symmetric) location-scale families.
This uniqueness is especially remarkable given that our decompositions operate directly on the quantile functions of the distributions.

Secondly, we show that the decompositions for the considered divergence measures mostly agree on which components are non-zero up to a subtle difference in the shift components.
We further derive sensitivities of the divergences to differences in shift and dispersion. For symmetric distributions and with increasing power $p$, the $p$-Wasserstein distance exhibits an increasing sensitivity towards differences in dispersion.
Furthermore, for Gaussian distributions, we show analytically that the Cram\'er distance weighs differences in dispersion even lower than the $\AV$ (i.e., the Wasserstein distance $\WD_1$ with the smallest power, $p = 1$).

Lastly, we derive comprehensive relations between our decompositions and suitable order relations of probability distributions. For each divergence function, there exist weak stochastic and dispersive orders such that the directed shift and dispersion components are (non-)zero if and only if the two distributions $F$ and $G$ are ordered accordingly.
These properties further strengthen the theoretical backbone of our decompositions.

While extensive work has been done on decomposing proper scoring rules (\citealt{Hersbach2000, Broecker2009, DGJ_2021, Bracher2021}, among many others), the literature on decompositions of divergence functions into interpretable components is sparse.
An exception is the exact decomposition of the 2-Wasserstein distance into the squared differences of the distributions' means and standard deviations, together with an analytically known remainder term capturing differences in shape \citep{BCMR1999, IV2015}, which has recently been used in applications by \citet{SFG21} and \citet{LA2022}.
Remark~\ref{rem:paramDecomp} illustrates that in contrast to this moment-based approach, which compares summary statistics that arise for each of the distributions separately, 
our decompositions are fully \emph{nonparametric} in the sense of aggregating (integrating over) all distributional differences of $F$ and $G$. 
Furthermore, our approach does not require a remainder term and is applicable to a range of divergence functions. 

We note that our decompositions do not apply to other well-known divergence functions such as the Kullback–Leibler divergence or the Hellinger distance, as they lack a suitable representation in terms of quantile functions. Broadly speaking, these divergences are based on point-wise comparisons of probability density functions, without a notion of distance between the elements of $\mathbb{R}$ (the support of $F$ and $G$), whereas the Wasserstein and Cram\'er distances consider \textit{closeness} of $F$ and $G$ \citep{Bellemare2017}, i.e., the concentration of probability mass in nearby regions. This way of quantifying the distance between two distributions connects naturally to the notions of shifts and dispersion put forward in this work.

In Section \ref{sec:applications}, we illustrate the decompositions in two applications from the fields of climate science and economic survey design. 
Firstly, we take a closer look at the evaluation of climate models by \cite{Thorarinsdottir2020}, who employ the Cram\'er distance to assess predictions of temperature extremes. For most models, our decompositions reveal systematic biases (upward for some, downward for others). Moreover, most models are found to produce overdispersed predictive distributions. 
In a second application, we build upon work by \cite{BDE23} who study how different histogram binning schemes  
impact responses in macroeconomic probabilistic surveys.
Here, our decompositions serve to demonstrate that the submitted forecast distributions indeed change in ways which are coherent with the authors' pre-registered hypotheses.

The supplementary text contains proofs and derivations together with details on how the decompositions can be computed or approximated in practice, additional illustrations, counterexamples and further details. Replication code is available at \url{https://github.com/resinj/replication_SD-Decomp}.

\section{Quantile-based decompositions of divergence functions}
\label{sec:Decompositions}

We start by illustrating our decomposition in detail for the area validation metric (AVM) in Section \ref{subsec:AVM}. The more complex cases of the general $p$-Wasserstein distance ($\WDp$) and the Cram\'er distance (CD) will be addressed in Sections \ref{subsec:Wasserstein} and \ref{subsec:CD}, respectively.

\subsection{The area validation metric}
\label{subsec:AVM}

For $p = 1$, the Wasserstein and $l_p$ distances coincide and are referred to as the \emph{area validation metric} ($\AV$),
\begin{equation}
	\AV(F, G) 
	= \int_{-\infty}^\infty \big| F(x) - G(x) \big| \dd x
	= \int_0^1 \big| F^{-1}(\tau) - G^{-1}(\tau) \big| \dd \tau.
\label{eq:defavm}
\end{equation}
In the following we rely on the latter quantile-based representation, which we rewrite as
\begin{align} \label{Eq:AVMFolded}
\AV(F, G) 
&= \frac{1}{2} \int_0^{1} \AV_\alpha(F, G) \dd \alpha,\quad\text{where} \\
\AV_\alpha(F, G) &= \left|F^{-1}\left(\tfrac{1-\alpha}{2}\right) - G^{-1}\left(\tfrac{1-\alpha}{2}\right)\right|
+ \left|F^{-1}\left(\tfrac{1+\alpha}{2}\right) - G^{-1}\left(\tfrac{1+\alpha}{2}\right)\right|. \notag
\end{align}
Here, the integrand $\AV_\alpha(F, G)$ simply compares the quantiles at levels $\frac{1-\alpha}{2}$ and $\frac{1+\alpha}{2}$, which span the \emph{central intervals} with coverage probability $\alpha$ of the two distributions.
The AVM results from integrating over all \emph{coverage levels} $\alpha \in [0, 1]$.

The quantile-based and interval-based representations of the AVM are illustrated graphically in the top row of Figure \ref{fig:Illustration_FoldedQuantiles}. The left plot visualizes expression \eqref{eq:defavm}, while the right panel illustrates \eqref{Eq:AVMFolded} in what we call a \textit{quantile spread plot}. In a nutshell, the latter plot displays the (at the median) folded and (to coverage levels) rescaled quantile functions, i.e., the $(1 \pm \alpha)/2$ quantiles of $F$ and $G$, which characterize the central intervals, as a function of the coverage $\alpha \in [0, 1]$. Twice the AVM then appears as the gray area between the two $\prec$-shaped curves if the dark gray area, which arises at coverages with disjoint central intervals, is counted twice.

\begin{figure*}[tb!]
\begin{center}
	\subcaptionbox{Quantile functions}{
		\includegraphics[width=\defaultfigwidth\textwidth,page = 2]{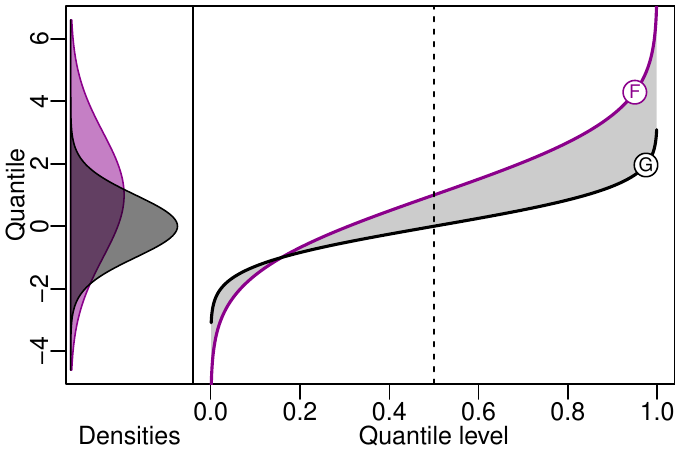}}
	\subcaptionbox{Quantile spread plot}{\includegraphics[width=\defaultfigwidth\textwidth,page = 5]{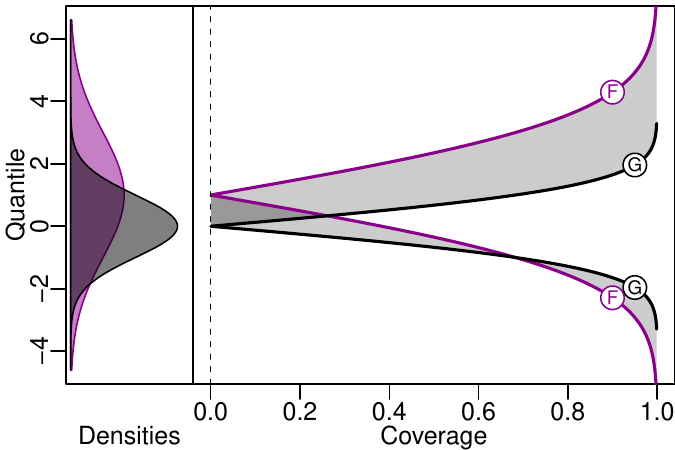}}\bigskip\\
	\subcaptionbox{AVM interval comparisons}{\includegraphics[width=\defaultfigwidth\textwidth,page = 7]{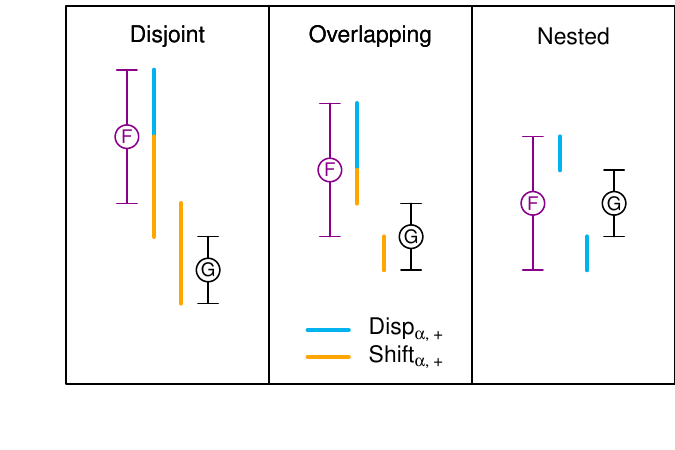}}
	\subcaptionbox{AVM decomposition}{\includegraphics[width=\defaultfigwidth\textwidth,page = 6]{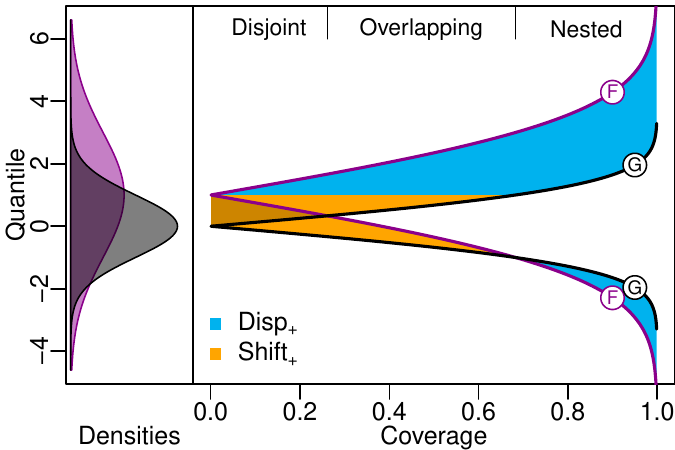}} 
\end{center}
\caption{Illustration of the AVM decomposition for a pair of normal distributions, $F=\norm(1,2)$ and $G=\norm(0,1)$ together with their densities. 
	The top row illustrates the AVM (gray shaded areas) in terms of the quantile-based formulation in \eqref{eq:defavm} in panel (a) and the formulation in terms of central intervals in \eqref{Eq:AVMFolded} in panel (b), where the shaded area equals $2 \times \AV$. 
	Panel (c) illustrates the $\AV_\alpha$ decomposition in \eqref{Eq:AV-decomp} for three distinct configurations of the central intervals of two generic distributions $F$ and $G$, which is discussed in the main text.
	Panel (d) shows the resulting decomposition of (twice) the $\AV$ in the quantile spread plot across all coverage levels $\alpha \in [0, 1]$. 
	The labels at the top of panel (d) indicate which configuration (as illustrated in panel (c)) occurs at each coverage level.}
\label{fig:Illustration_FoldedQuantiles}
\end{figure*}

In order to obtain four components as in \eqref{Eq:GeneralDecomp}, we use an $\alpha$-wise decomposition of the integrand in \eqref{Eq:AVMFolded},
\begin{align}
\AV_\alpha(F, G) 
=\, \underbrace{\shift_{\alpha,+}^{\AV}(F, G)}_{\text{``}F \text{ shifted up''}} 
\,+\, \underbrace{\shift_{\alpha,-}^{\AV}(F, G)}_{\text{``}F \text{ shifted down''}}
\,+\, \underbrace{\disp_{\alpha,+}^{\AV}(F, G)}_{\text{``}F \text{ more dispersed''}}
\,+\, \underbrace{\disp_{\alpha,-}^{\AV}(F, G)}_{\text{``}F \text{ less dispersed''}}.
\label{Eq:AV-decomp}
\end{align}
Using $[z]_+ := \max(z,0)$ for the positive part of a real number $z \in \mathbb{R}$, we define the $\alpha$-wise components
\begin{align}
\shift_{\alpha,+}^{\AV}(F,G)
&:=2 \Big[ \min\big\{F^{-1}\left(\tfrac{1+\alpha}{2}\right) - G^{-1}\left(\tfrac{1+\alpha}{2}\right),
F^{-1}\left(\tfrac{1-\alpha}{2}\right) - G^{-1}\left(\tfrac{1-\alpha}{2}\right) \big\} \Big]_+\quad\text{and}
\label{Eq:AVM_Shift} \\
\disp_{\alpha,+}^{\AV}(F,G) 
&:= 
\Big[\Big(F^{-1}\left(\tfrac{1+\alpha}{2}\right) - G^{-1}\left(\tfrac{1+\alpha}{2}\right)\Big)
- \Big(F^{-1}\left(\tfrac{1-\alpha}{2}\right) - G^{-1}\left(\tfrac{1-\alpha}{2}\right)\Big)\Big]_+ \label{Eq:AVM_Disp}
\end{align}
explained below.
The remaining two $\alpha$-wise components are symmetrically defined as 
\begin{equation}
	\shift_{\alpha,-}^{\AV}(F,G) := \shift_{\alpha,+}^{\AV}(G,F)\quad\text{and}\quad
	\disp_{\alpha,-}^{\AV}(F,G) := \disp_{\alpha,+}^{\AV}(G,F).
\label{Eq:AVM_Minus}
\end{equation}

We henceforth refer to the two components with subscript `$+$' as the \emph{plus} components as they quantify how $F$ is, relative to $G$, shifted \emph{upwards} and has an \emph{increased} dispersion, respectively.
Similarly, we refer to the terms with subscript `$-$' as the \emph{minus} components. 

The decomposition terms in \eqref{Eq:AVM_Shift}--\eqref{Eq:AVM_Disp} attribute the overall difference between the central $\alpha$-interval endpoints of $F$ and $G$---which is captured by the integrand $\AV_{\alpha}(F,G)$ in \eqref{Eq:AVMFolded}---to the components by more intricate interval comparisons.
In a nutshell, the shift components capture twice the distance that one of the $\alpha$-intervals needs to be moved to lie within the other, while the dispersion components measure by how much one needs to (de)compress the $\alpha$-interval of $G$ to have the same length as the $\alpha$-interval of $F$.

The interval comparisons are illustrated in detail in the bottom left panel in Figure \ref{fig:Illustration_FoldedQuantiles} for three distinct configurations of the central intervals of $F$ and $G$ (in purple and black, respectively). The two 
bars plotted between the intervals capture the two summands in \eqref{Eq:AVMFolded}. In the illustrations, the $F$-intervals are larger than the $G$-intervals by the blue parts, which are attributed to the dispersion component in \eqref{Eq:AVM_Disp}.
Note that \eqref{Eq:AVM_Disp} can be rewritten as the positive part of the difference between the lengths of the $F$- and $G$-intervals.

In the case of \emph{nested} central intervals, the entire $\alpha$-wise AVM in \eqref{Eq:AVMFolded} is attributed to the dispersion component in \eqref{Eq:AVM_Disp}.
In the illustrated cases of \emph{overlapping} or \emph{disjoint} central intervals, the $F$-intervals lie higher than the $G$-intervals (in that the upper and lower endpoints are ordered in the same way). 
In this case, the shift component in \eqref{Eq:AVM_Shift} captures twice the minimum difference between the endpoints, which accounts for the remaining orange part of the colored bars. 

If, on the other hand, the $G$-interval is wider than the $F$-interval, the difference in interval length is attributed to the minus dispersion component. Analogously, if the intervals are ordered differently, their difference in location is attributed to the minus shift component.
On a technical level, these separations into plus and minus terms are achieved by the $[\cdot]_+$ operator in \eqref{Eq:AVM_Shift} and \eqref{Eq:AVM_Disp}, respectively.

Notably, the $\alpha$-wise dispersion components generalize upon the difference of the interquartile ranges (that arises for $\alpha = 0.5$ in \eqref{Eq:AVM_Disp}), and the shift components nest a comparison of the distributions' medians (for $\alpha = 0$ in \eqref{Eq:AVM_Shift}).

The bottom right panel of Figure~\ref{fig:Illustration_FoldedQuantiles} illustrates how the overall AVM decomposition arises through  \emph{integration} over all coverage levels $\alpha \in [0,1]$.
As in the top right panel, the entire colored area corresponds to twice the AVM.
Notice that the dark yellow area is counted twice, because the two bars overlap in disjoint interval configurations, as illustrated in panel (c). 

Consequently, the components of the final decomposition as in \eqref{Eq:GeneralDecomp} are defined as
\begin{equation}
	\shift_\pm^{\AV}(F, G) := \frac{1}{2}  \int_0^1 \shift_{\alpha,\pm}^{\AV}(F, G) \dd \alpha \quad\text{and}\quad
	\disp_\pm^{\AV}(F, G) := \frac{1}{2} \int_0^1 \disp_{\alpha,\pm}^{\AV}(F, G) \dd \alpha,
\label{eq:avm_disp_plus}
\end{equation}
which yields the following result.

\begin{prop}
\label{prop:AVMExact}
The $\AV$ decomposition,
\begin{align*}
	\AV(F, G) = \shift_+^{\AV}(F, G) \,+\, \shift_-^{\AV}(F, G)\,
	+\, \disp_+^{\AV}(F, G) \,+\, \disp_-^{\AV}(F, G),
\end{align*}
whose components are given in \eqref{Eq:AVM_Shift}--\eqref{eq:avm_disp_plus}, is exact.
\end{prop}

Our decomposition is \emph{nonparametric} by construction in the sense that it disaggregates the integrand $\AV_\alpha(F, G)$ in \eqref{Eq:AVMFolded} at the fundamental quantile (or central interval) level and aggregates the resulting distributional differences through integration.
Hence, the decomposition terms inherently capture all distributional discrepancies between $F$ and $G$, as opposed to e.g., a moment-based decomposition into differences in means and variances as in \citet{BCMR1999} and \citet{IV2015}.

At the $\alpha$-wise level, at most one shift and one dispersion term in \eqref{Eq:AVM_Shift}--\eqref{Eq:AVM_Minus} can be positive.
In contrast, all four components can be positive in the aggregated (integrated over all $\alpha$ levels) decomposition in Proposition \ref{prop:AVMExact}, which we consider to be a natural feature of a nonparametric decomposition, as illustrated in the following examples.

\begin{figure*}[tb]
\begin{center}
	\subcaptionbox{AVM decomposition with nonzero dispersion terms}{
		\includegraphics[width = \defaultfigwidth\textwidth,page = 1]{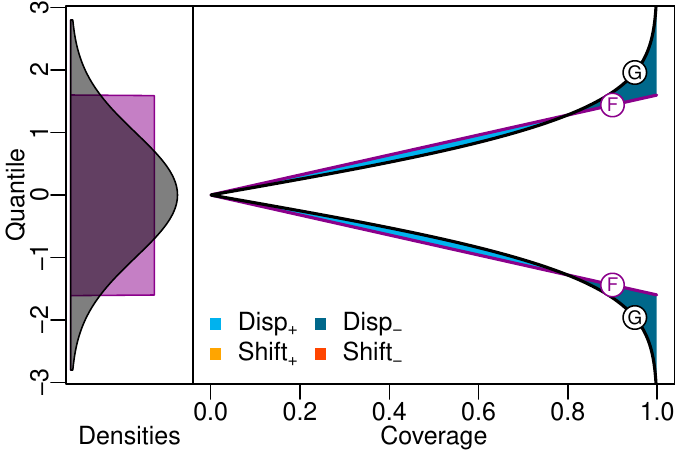}}
	\subcaptionbox{AVM decomposition with nonzero shift terms}{
		\includegraphics[width = \defaultfigwidth\textwidth,page = 2]{figures/Fig2.pdf}}
	\caption{Graphical illustrations (as in Figure~\ref{fig:Illustration_FoldedQuantiles}, panel (d)) of the $\AV$ decompositions for the two distributional comparisons from Example~\ref{exmpl:BothComponentsPositive} with densities on the left.
		Comparison (a) illustrated on the left leads to nonzero plus and minus dispersion components, while comparison (b) on the right leads to nonzero plus and minus dispersion components.}
\label{fig:examples_avm}
\end{center}
\end{figure*}

\begin{example}
\label{exmpl:BothComponentsPositive}
Here, we present two simple examples that lead to an $\AV$ decomposition where both the plus and the minus shift or dispersion components are nonzero.

\begin{enumerate}
\item 
\sloppy
The left panel of Figure~\ref{fig:examples_avm} compares a uniform distribution $F = \unif (-1.6, 1.6)$ with a standard normal distribution $G = \norm(0, 1)$.
This example illustrates two distributions whose differences in central interval width vary for different coverage levels $\alpha$, which yields nonzero plus and minus dispersion components, 
$\AV(F,G) = \disp_+^{\AV}(F, G) +  \disp_-^{\AV}(F, G) = 0.065 + 0.063$.
As neither distribution exhibits a smaller dispersion at all coverage levels, having two positive dispersion components is a natural feature of our decomposition.

\item 
The right panel of Figure~\ref{fig:examples_avm} compares two mixtures of uniform distributions with a mirrored asymmetry, $F = 0.5\times \unif(0,4) + 0.5\times \unif(4,6)$, and $ G = 0.5\times \unif(1,3) + 0.5\times \unif(3,7)$. 
While the width of all central intervals coincides (resulting in zero dispersion components), the locations of the central intervals are shifted in different directions for different coverage levels. This results in nonzero plus and minus shift components,
$\AV(F,G) = \shift_+^{\AV}(F, G) + \shift_-^{\AV}(F, G) = 0.25 + 0.25$.     
\end{enumerate}
\end{example}

Further examples with (up to) four nonzero components are readily constructed. 
While such examples may occasionally be encountered in practice, the decomposition uncovers mostly clear differences in location and dispersion in our applications.

\subsection{The $p$-Wasserstein distance}
\label{subsec:Wasserstein}

We next present a generalized decomposition for the $p$-th power of the $p$-Wasserstein distance in \eqref{Eq:pWasserstein} for $p \in [1,\infty)$.
Here, we directly state the decomposition in its integral-form and dispense with a disaggregated $\alpha$-wise treatment as in \eqref{Eq:AV-decomp}, which can be obtained by simply omitting the integrals in the following formulas.

Using the notation $z^{[p]} := \sgn(z)\cdot\vert z\vert^p$ as a shorthand for the \emph{signed $p$-th power} of a real number $z \in \mathbb{R}$, we generalize the components given by \eqref{Eq:AVM_Shift}--\eqref{eq:avm_disp_plus} to

\begin{align}    
\shift_+^{\WD_p}(F,G) &:= 
\frac12 \int_0^1 2 \bigg[\min \bigg\{\Big( F^{-1}\left( \tfrac{1+\alpha}{2} \right) - G^{-1}\left( \tfrac{1+\alpha}{2} \right) \Big)^{[p]},
\Big(F^{-1}\left( \tfrac{1-\alpha}{2} \right) - G^{-1}\left( \tfrac{1-\alpha}{2} \right) \Big)^{[p]}  \bigg\}\bigg]_+  \dd \alpha,
\label{Eq:WD_Shift} \\
\disp_+^{\WD_p}(F,G) &:=  
\frac12 \int_0^1 \bigg[\Big( F^{-1} \left( \tfrac{1+\alpha}{2} \right) - G^{-1}\left( \tfrac{1+\alpha}{2} \right)\Big)^{[p]}
-\Big(F^{-1}\left( \tfrac{1-\alpha}{2} \right) - G^{-1}\left( \tfrac{1-\alpha}{2} \right) \Big)^{[p]} \bigg]_+ \dd\alpha. \label{Eq:WD_Disp}
\end{align}
We again define $\shift_-^{\WD_p}(F,G) := \shift_+^{\WD_p}(G,F)$ and $\disp_-^{\WD_p}(F,G) := \disp_+^{\WD_p}(G,F)$ through symmetry and obtain the following result.

\begin{prop}
\label{prop:WDExact}
The $\WDp$ decomposition,
\begin{align*}
\WD_p(F, G) = \shift_+^{\WD_p}(F, G) \,+\, \shift_-^{\WD_p}(F, G) \,+\, \disp_+^{\WD_p}(F, G) \,+\, \disp_-^{\WD_p}(F, G),
\end{align*}
whose components are given in \eqref{Eq:WD_Shift}--\eqref{Eq:WD_Disp}, is exact.
\end{prop}

The interpretations of the respective components match the ones from Section \ref{subsec:AVM} and Figure~\ref{fig:Illustration_FoldedQuantiles} when simply taking the signed $p$-th power of the distances between the respective interval ends considered in \eqref{Eq:WD_Shift}--\eqref{Eq:WD_Disp}.
For $p > 1$, taking the signed $p$-th power of interval end differences (opposed to $p=1$ for the $\AV$) tends to result in a larger proportion of the Wasserstein distance being explained by a difference in dispersion than for the AVM. 
We study this phenomenon in mathematical detail in Section \ref{subsec:ComparisonDecompositions}.

\begin{remark}
	\label{rem:paramDecomp}
	\citet{IV2015} discuss an alternative decomposition of the squared 2-Wasserstein distance into a location, scale and a remaining shape component given by
	\begin{equation}
		\label{eq:paramDecomp}
		\WD_2(F,G) = \left(\int_0^1 |F^{-1}(\alpha) - G^{-1}(\alpha)|^2 \text{d}\alpha \right)
		=\underbrace{(\mu_F - \mu_G)^2}_\text{location} + \underbrace{(\sigma_F - \sigma_G)^2}_\text{scale} + \underbrace{2\sigma_F\sigma_G(1 - \rho_{F, G}),}_\text{shape}
	\end{equation}
	where $\mu_F, \mu_G$ and $\sigma_F,\sigma_G$ denote the respective means and standard deviations, while $\rho_{F, G}$ is the Pearson correlation of the points in a quantile-quantile plot of $F$ and $G$.
	In contrast to our quantile-based decomposition, the decomposition in \eqref{eq:paramDecomp} draws on the respective moments of the distributions. 
	It could further be refined to \textit{directed} location and scale components as in our decomposition by considering the sign of the differences in means and standard deviations, respectively.
	
	The decomposition in \eqref{eq:paramDecomp} shares some of the properties of our decomposition for the CD and AVM, such as being shift invariant (i.e., simple shifts only affect the location component; c.f.\ Proposition \ref{prop:DispInvarianceforShifts}) and compatibility with location-scale families (c.f.\ Proposition \ref{Prop:CompLocScale}, which however does not apply to our $\WDp$-decompositions for $p>1$). 
	Furthermore, \eqref{eq:paramDecomp} is compatible with the usual stochastic and dispersive orders and is clearly linked to weak stochastic and dispersive orders based on ordering by means and standard deviations (cf.\ Section \ref{sec:Orders}).
	The moment-based decomposition, however, is limited to the squared 2-WD and does not apply to the AVM and CD.
\end{remark}

\begin{remark}
	While the $\WDp$ decomposition from Proposition~\ref{prop:WDExact} holds for any finite $p \in [1,\infty)$, it is not obvious how this can be generalized to $p = \infty$, where
	\[
	\lim_{p\rightarrow\infty} \WDp^{1/p}(F,G) = \sup_{\alpha\in(0,1)} \vert F^{-1}(\alpha) - G^{-1}(\alpha)\vert.
	\]
	
	In fact, our $\WDp$ decomposition operates on the $p$-th power of the $p$-Wasserstein distance, and there is no canonical way of transposing it to the ordinary Wasserstein distance $\WDp^{1/p}$.
	One possibility that achieves a decomposition of $\WDp^{1/p}$ for fixed $p < \infty$ is to simply rescale the $\WDp$ decomposition terms with $\frac{\WDp^{1/p}}{\WDp}$, i.e., for any $\operatorname{Comp} \in \{\shift_\pm,\disp_\pm\}$  we set $\operatorname{Comp}^{\WDp^{1/p}} = \frac{\WDp^{1/p}}{\WDp} \operatorname{Comp}^{\WDp}$, which clearly yields rescaled components that sum to the $p$-th root $\WDp^{1/p}$.
	
	Then, taking the limit as $p \to \infty$ yields a decomposition for $\WD_\infty$. 
	We suspect that these proportions will converge under fairly general conditions if the $\infty$-Wasserstein distance is finite. 
	For example, Theorem \ref{Thm:DispInequality_pWD} below (together with Proposition \ref{prop:ShiftSymDist}) implies that this is the case for the shift components
	when comparing symmetric distributions.
\end{remark}

\subsection{The Cram\'er distance}
\label{subsec:CD}

The \textit{Cram\'er distance} ($\CD$) or \textit{integrated quadratic distance} corresponds to the squared $l_2$ metric,
\begin{align}
\label{Eq:CramerDistance}
\CD(F, G) = \int_{-\infty}^\infty \big| F(x) - G(x) \big|^2 \dd x.
\end{align}
We start by providing a novel representation of the Cram\'er distance in terms of quantile functions, which is necessary to apply the ideas behind our quantile-based decompositions to the Cram\'er distance.

\begin{prop}
\label{prop:CDQuantileRep}
The Cram\'er distance in \eqref{Eq:CramerDistance} can be expressed as
\begin{align}
\label{Eq:CD_quantile}
\CD(F, G) &= 2 \int_0^1 \int_0^1 \chi(\tau,\xi) \; \left| F^{-1}(\tau) - G^{-1}(\xi) \right| \dd \tau \dd \xi,\quad\text{where} \\
\chi(\tau,\xi) &:= \one \left\{ \sgn(\tau - \xi) \neq \sgn \big( F^{-1}(\tau) - G^{-1}(\xi) \big) \right\}.\notag
\end{align}
\end{prop}

The indicator $\chi(\tau,\xi)$ used in \eqref{Eq:CD_quantile} serves as an \emph{incompatibility check} of the $\tau$-quantile of $F$ with the $\xi$-quantile of $G$: 
Whenever the order of the quantiles and quantile levels is at odds, i.e., $F^{-1}(\tau) > G^{-1}(\xi)$ despite $\tau < \xi$ or vice versa, the indicator function $\chi$ returns one and hence the pair of quantiles
contributes to the Cram\'er distance.

Starting from the representation \eqref{Eq:CD_quantile}, a decomposition similar to the ones in Sections \ref{subsec:AVM}--\ref{subsec:Wasserstein} arises that compares central intervals at differing coverage levels $\alpha, \beta \in [0,1]$ by setting
\begin{align}
\label{Eq:CD_Shift}
&\begin{aligned}
\shift_+^{\CD}(F, G) :=
\frac{1}{2} \int_{0}^{1}\int_{0}^1 &\Big[\min\Big\{ F^{-1}\left(\tfrac{1+\alpha}{2}\right) - G^{-1}\left(\tfrac{1+\beta}{2} \right),
F^{-1}\left(\tfrac{1-\alpha}{2}\right) - G^{-1}\left(\tfrac{1-\beta}{2}\right) \Big\} \Big]_+ \\ &+ \left[F^{-1}\left(\tfrac{1-\alpha}{2}\right) - G^{-1}\left(\tfrac{1+\beta}{2}\right)\right]_+ \dd \alpha \dd \beta,    
\end{aligned} \\
\label{Eq:CD_Disp}
&\disp_+^{\CD}(F, G) :=
\frac{1}{2} \int_0^1\int_0^\beta \Big[ \Big( F^{-1}\left(\tfrac{1+\alpha}{2}\right) - G^{-1}\left(\tfrac{1+\beta}{2}\right) \Big)
-\Big( F^{-1}\left(\tfrac{1-\alpha}{2} \right) - G^{-1}\left(\tfrac{1-\beta}{2}\right) \Big) \Big]_+ \dd\alpha \dd\beta.
\end{align}
We define the minus counterparts by $\shift_-^{\CD}(F,G) := \shift_+^{\CD}(G,F)$ and $\disp_-^{\CD}(F,G) := \disp_+^{\CD}(G,F)$, as before.

\begin{prop}
\label{prop:CDExact}
The $\CD$ decomposition,
\begin{align*}
\CD(F, G) = \shift_+^{\CD}(F, G) \,+\, \shift_-^{\CD}(F, G)
\,+\, \disp_+^{\CD}(F, G) \,+\, \disp_-^{\CD}(F, G),
\end{align*}
whose components are given in \eqref{Eq:CD_Shift}--\eqref{Eq:CD_Disp}, is exact.
\end{prop}

\begin{figure*}[tb]
	\begin{center}
		\subcaptionbox{CD interval comparisons}{\includegraphics[width=\defaultfigwidth\textwidth,page = 1]{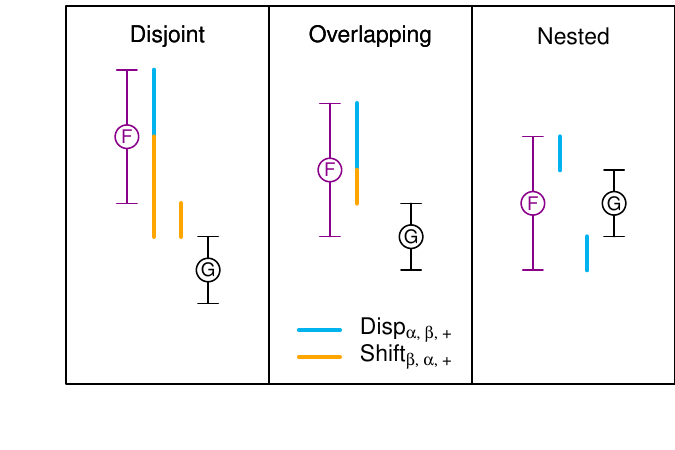}}
		\subcaptionbox{CD decomposition slice}{
			\includegraphics[width = \defaultfigwidth\textwidth,page = 5]{figures/Fig3.pdf}}
		\caption{Illustration of the CD decomposition similar to Figure \ref{fig:Illustration_FoldedQuantiles}, panels (c) and (d).
			Panel (a) illustrates how the CD decomposition arises from individual comparisons of the central intervals of $F$ and $G$ for three distinct configurations.
			Panel (b) shows a quantile spread plot for the two normal distributions $F=\norm(1,2)$ and $G=\norm(0,1)$ together with their densities. The central interval of $G$ at the fixed coverage level $\beta = 0.4$ is emphasized by the black horizontal lines. The corresponding \emph{slice} of the Cram\'er distance is obtained by comparing the fixed central interval of $G$ to all central intervals spanned by $F$. The overall CD and its decomposition are obtained by integrating all slices (i.e., integration across $\beta \in [0,1]$).
			The labels at the top of panel (d) indicate which configuration (as illustrated in panel (a)) occurs at a given coverage level $\alpha$ (for the fixed level $\beta = 0.4$).}
		\label{fig:illustration_cd}
	\end{center}
\end{figure*}

In analogy to the AVM decomposition in \eqref{Eq:AVM_Shift}--\eqref{eq:avm_disp_plus}, the CD decomposition in \eqref{Eq:CD_Shift}--\eqref{Eq:CD_Disp} aggregates suitable comparisons of the central intervals via integration. Hence, we focus our discussion on the differences between the two decompositions.
Most strikingly and in contrast to the previous decompositions, the components in \eqref{Eq:CD_Shift}--\eqref{Eq:CD_Disp} compare central intervals at differing coverage levels $\alpha$ and $\beta$ through the double integrals.
For individual pairs of central intervals, differences to the AVM decompositions can be observed in the left-hand plot of Figure~\ref{fig:illustration_cd} that graphically illustrates the CD decomposition akin to the bottom-left plot in Figure~\ref{fig:Illustration_FoldedQuantiles}.

While the integrand in the shift components in \eqref{Eq:CD_Shift} resembles \eqref{Eq:AVM_Shift} in capturing the distance that one of the intervals needs to be moved to lie within the other, the factor of two is missing in \eqref{Eq:CD_Shift}.
Moreover, in the case of disjoint intervals, the term in the last line of \eqref{Eq:CD_Shift} yields an additional contribution capturing the distance that one of the intervals needs to be moved to overlap with the other.

The integrand in the dispersion components in \eqref{Eq:CD_Disp} captures by how much
one needs to (de)compress the $\beta$-interval of $G$ to have the
same length as the $\alpha$-interval of $F$ \emph{if the coverage levels and interval lengths are at odds}. For example, if the $\alpha$-interval of $F$ is larger than the $\beta$-interval of $G$, different interval lengths indicate a difference in distributions only if $\alpha \leq \beta$, which is reflected by the integration boundary of the inner integral in \eqref{Eq:CD_Disp}. 
Otherwise, an increase in coverage naturally leads to an increased interval length even for identical distributions. In contrast, a shift between central intervals always hints at a distributional difference, regardless of coverage, as the central intervals of identical distributions are always nested.

The final components in \eqref{Eq:CD_Shift}--\eqref{Eq:CD_Disp} arise by integrating over the coverage levels $\alpha$ and $\beta$.
As a joint graphical illustration of the double integral over $\alpha$ and $\beta$ is challenging, we exemplarily fix the level of $\beta = 0.4$ in the right-hand plot of Figure~\ref{fig:illustration_cd} and illustrate the contributions in the integral over $\alpha$ for two normal distributions. In the supplement, Figure~\ref{fig:illustration_cd_grid} 
shows equivalent plots for other values of $\beta$.

In the figure, we illustrate the comparison of the central intervals of $F$ at all coverage levels $\alpha$ to the fixed central interval of $G$ with coverage $\beta$, which is emphasized by the horizontal black lines.
Contributions to the plus dispersion component that quantify by how much the central intervals of $F$ are wider than the fixed central interval of $G$ only arise for coverages $\alpha \leq \beta = 0.4$.
No contributions to the plus shift component arise for coverages $\alpha \gtrsim 0.554$, % 0.5540747
as the $0.4$-interval of $G$ is strictly nested in the central intervals of $F$ and hence no shifts between intervals arise.
For coverages $0.188 \lesssim \alpha \lesssim 0.554$, the respective intervals are overlapping and the height of the orange area corresponds to the shift distance.
We plot the area between the central intervals in such a way that it is conveniently bounded by lower or upper interval ends, which leads to the break in the orange area at coverage $\alpha = 0.4$.
Finally, for $\alpha \lesssim 0.188$, % 0.1879462
the intervals are disjoint, and the additional contribution in the last line of \eqref{Eq:CD_Shift} results in the dark yellow area contributing twice. 

As an aside, we note that the CD decomposition gives rise to a decomposition of the continuous ranked probability score (CRPS), a popular scoring rule used to evaluate probabilistic forecasts \citep{GR07}.
\begin{remark}
If $G = \dirac{y}$ is a Dirac distribution at $y$, the Cram\'er distance reduces to the CRPS,
\[
\CD(F, G) = \CRPS(F, y) = \int_{-\infty}^\infty |F(x) - \one(x \geq y)|^2 \textnormal{d}x.
\]
As $G$ has no variance in this case, one of the dispersion components vanishes, namely, $\disp_-^{\CD}(F, G) = 0$. Denoting by $m_F$ any median of $F$, the remaining components simplify to 
\begin{align*}
\disp_+^{\CD}(F, G) &= \disp^{\CRPS}(F,y)
= \CRPS(F, m_F),\\
\shift_+^{\CD}(F, G) &= \shift_+^{\CRPS}(F, y)
= \one(m_F > y) \times [\CRPS(F, y) - \CRPS(F, m_F)], \\
\shift_-^{\CD}(F, G) &= \shift_-^{\CRPS}(F, y) = \one(m_F < y) \times [\CRPS(F, y) - \CRPS(F, m_F)].
\end{align*}
This decomposition is equivalent to the decomposition of the weighted interval score (an interval-based approximation of the CRPS) mentioned in \cite{Bracher2021}.
\end{remark}

\section{Theoretical properties of the decompositions}
\label{Sec:Properties}

This section provides an in-depth analysis of the theoretical properties of the proposed decompositions.
Section \ref{subsec:BasicProperties} establishes their natural behavior for distribution classes that are shifted, symmetric, of location-scale type and Gaussian.
In Section \ref{subsec:ComparisonDecompositions}, we contrast the sensitivity of the different divergence measures to distributional changes in dispersion and shift.

\subsection{Basic properties}
\label{subsec:BasicProperties}

By construction, the components satisfy a simple symmetry. 
\begin{prop}[Symmetry]
\label{prop:Symmetry}
\sloppy
For any $p \in [1,\infty)$ and $\D \in \{\AV, \WDp,\CD\}$, we have
\begin{align}
\label{Eq:Symmetry}
\shift_+^{\D}(F,G) = \shift_-^{\D}(G,F), 
\quad \text{and}\quad
\disp_+^{\D}(F,G) = \disp_-^{\D}(G,F).
\end{align}
\end{prop}
Hence, the properties that are outlined for the plus components (with subscript `$+$') also apply to the minus counterparts (with subscript `$-$') by symmetry.

Furthermore, it is easy to see that the dispersion components of the $\AV$ and the $\CD$ are invariant to simple changes in location without imposing any distributional restrictions on $F$ and $G$.

\begin{prop}[Dispersion invariant to shifts]
\label{prop:DispInvarianceforShifts}
For any distribution $F$ and $s\in \R$, the shifted distribution $F_s$ is given by $F_s(z) := F(z - s)$ for any $z \in \mathbb{R}$.
Then, for all $s \in \mathbb{R}$ and $\D \in \{\AV, \CD\}$, it holds that
\begin{equation}
\label{Eq:DispInvarianceforShifts}
\disp_\pm^{\D}(F_s,G) = \disp_\pm^{\D}(F,G). 
\end{equation}
\end{prop}

This invariance property of the dispersion components to simple location shifts is very natural.
Unfortunately, it is not shared by the higher order Wasserstein distance decompositions with $p > 1$ as the simple shift by $s$ does not cancel out in the dispersion components when subtracting signed $p$-th powers of the upper and lower interval end differences in \eqref{Eq:WD_Disp}.\footnote{This also becomes obvious in the closed-form expressions for normal distributions  in Supplement \ref{sec:WDpNormals} 
when $F$ and $G$ have the same mean but different variances and a shifted version $F_s$ is compared to $G$.}

In order to analyze when one (or both) shift components vanish, we restrict attention to symmetric distributions.
We call a distribution $F$ symmetric if there exists a value $m \in \mathbb{R}$ such that $F^{-1}(\gamma) = 2m - F^{-1}(1-\gamma)$ for almost all $\gamma \in (0,1)$.\footnote{Note that admitting at most countably many points where the symmetry condition for the quantile function may be violated accounts for discontinuities in the quantile function. Such points form a null set in $[0,1]$ and can therefore be excluded from the integration domain without changing the value of an integral.}
If there exists a unique median, then $m = F^{-1}(0.5)$, otherwise $m$ is the midpoint of the interval of medians, which we henceforth call the \emph{central median}.
For symmetric distributions, the central median also coincides with the mean.

The following result shows that for two symmetric (but otherwise entirely flexible) distributions, at most one shift component is nonzero, and the direction of the shift agrees with the order of the medians.

\begin{prop}[Shift for symmetric distributions]
\label{prop:ShiftSymDist}
Let $F$ and $G$ be symmetric distributions with central medians $m_F$ and $m_G$ and $\D \in \{\AV,\WDp,\CD\}$, $p \in [1,\infty)$.
\begin{enumerate}
\item    
If $m_F \leq m_G$, then $\shift_+^{\D}(F,G) = 0$. 

\item     
If the medians of $F$ and $G$ are unique, then
\begin{equation}
\label{Eq:ShiftSymDist}
\shift_+^{\D}(F,G) > 0 \quad\Longleftrightarrow\quad m_F > m_G.
\end{equation}
\end{enumerate}
\end{prop}

The result in part (b) of Proposition \ref{prop:ShiftSymDist} can be generalized to distributions with non-unique medians if the respective median intervals are non-nested. In this case, it still suffices to compare the central medians.

We continue to illustrate that our decompositions work as expected for distributional comparisons within location-scale families in that they reflect changes in location (scale) through a corresponding shift (dispersion) component.
Notice that in the following, the location parameters $\ell_F$ and $\ell_G$, and the scale parameters $s_F$ and $s_G$ are not necessarily means, medians or standard deviations of $F$ and $G$, respectively.\footnote{This is the case only if $H$ is standardized to have mean or median zero, respectively, and variance equal to one.}

\begin{prop}
\label{Prop:CompLocScale}
Let $F$ and $G$ be distributions from the same location-scale family, i.e., there exists a non-degenerate distribution $H$ such that the quantile functions satisfy the relations $F^{-1} = s_F H^{-1} + \ell_F$ and $G^{-1} = s_G H^{-1} + \ell_G$ for some $\ell_F, \ell_G \in \mathbb{R}$ and $s_F, s_G > 0$.
\begin{enumerate}
\item 
For any $\D \in \{\AV,\WDp,\CD\}$ with $p \in [1,\infty)$, it holds that
\begin{equation}
\label{Eq:DispLocScale}
\disp_+^{\D}(F,G) > 0 \quad \Longleftrightarrow \quad s_F > s_G.
\end{equation}

\item 
The central medians of $F$ and $G$ are $m_F = s_F m_H + \ell_F$ and $m_G =  s_G m_H + \ell_G$, respectively, where $m_H$ denotes the central median of $H$. 
For any $p \in [1,\infty)$, if $m_F \leq m_G$, then $\shift_+^{\WDp}(F,G) = 0$. Moreover, if the median of $H$ is unique,
\begin{equation}
\label{Eq:ShiftLocScale}
\shift_+^{\WDp}(F,G) > 0 \quad \Longleftrightarrow \quad m_F > m_G.
\end{equation}
\end{enumerate}
\end{prop}

Proposition~\ref{Prop:CompLocScale} obviously also applies to pure location or pure scale families.
As above, part (b) of Proposition \ref{Prop:CompLocScale} can be generalized to distributions with non-unique medians if the median intervals of $F$ and $G$ are non-nested by comparing the central medians in \eqref{Eq:ShiftLocScale}.

The following example shows that the second claim of Proposition \ref{Prop:CompLocScale} does not generalize to the CD decomposition. Nonetheless, an analogous equivalence holds for symmetric distributions by Proposition~\ref{prop:ShiftSymDist}. 

\begin{figure*}
\centering
\subcaptionbox{AVM decomposition}{
\includegraphics[width = \defaultfigwidth\textwidth,page = 1]{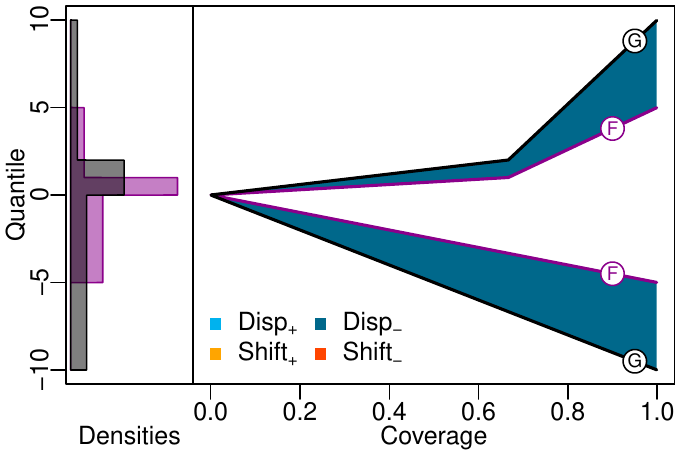}}
\subcaptionbox{CD decomposition slice}{
\includegraphics[width = \defaultfigwidth\textwidth,page = 6]{figures/Fig4.pdf}}
\caption{Illustration of the $\AV$ and $\CD$ (for $\beta=0.5$) decompositions for distributions $F$ and $G$ given in Example~\ref{Ex:CD-shift_loc-scale} from an asymmetric location-scale family where the distributions only differ in scale.
See Figures~\ref{fig:Illustration_FoldedQuantiles} and \ref{fig:illustration_cd} for detailed descriptions of the plots.}
\label{Fig:CD-shift_loc-scale}
\end{figure*}

\begin{example}
\label{Ex:CD-shift_loc-scale}
Consider the asymmetric distributions
$F=0.5 \times \mathcal{U}[-5,0] + \tfrac13 \times \mathcal{U}[0,1] + \tfrac16 \times \mathcal{U}[1,5]$ and $G = 0.5 \times \mathcal{U}[-10,0] + \tfrac13 \times \mathcal{U}[0,2] + \tfrac16 \times \mathcal{U}[2,10]$, whose densities are shown in Figure~\ref{Fig:CD-shift_loc-scale}.
The distributions $F$ and $G$ stem from the same location-scale family as $H=F$ yields $F^{-1} = H^{-1}$ and $G^{-1} = 2 F^{-1}$ with location parameters $m_F = m_G$.
Hence, the right-hand side of the equivalence statement in  Proposition~\ref{Prop:CompLocScale} (b) does not hold.
However, as illustrated in the right-hand plot of Figure~\ref{Fig:CD-shift_loc-scale} (here, for fixed $\beta = 0.5$), the CD decomposition exhibits a nonzero shift component.
In particular, it holds that $\CD(F,G) = \shift_+^{\CD}(F,G) + \disp_+^{\CD}(F,G) \approx 0.033 + 0.232$, while $\AV(F,G) = \disp_+^{\AV}(F,G) = 23/12 \approx 1.917$.
\end{example}

\begin{remark}
	The properties established in this section refer to the (non-)positivity of certain components, and it generally does not seem possible to make statements of a more \textit{quantitative} nature. As an exception, for the dispersion components of the AVM, it holds that
	$$
	\disp_+^{\text{AVM}}(F, G) - \disp_-^{\text{AVM}}(F, G) = D_F - D_G,
	$$
	where $D_F$ and $D_G$ are the mean absolute deviations around the medians of $F$ and $G$.
\end{remark}

The following theorem derives \emph{uniqueness} of our decompositions if certain natural properties are warranted.

\begin{theorem}
\label{thm:Uniqueness}
Let $F$ and $G$ be distributions from the same location-scale family with unique medians.
\begin{enumerate}
\item The shift-dispersion decomposition of $\AV(F,G)$ given in Proposition \ref{prop:AVMExact} is uniquely determined by the conditions \eqref{Eq:Symmetry}, \eqref{Eq:DispInvarianceforShifts}, \eqref{Eq:DispLocScale} and \eqref{Eq:ShiftLocScale}.

\item If in addition $F$ and $G$ are symmetric, the shift-dispersion decomposition of $\CD(F,G)$ given in Proposition \ref{prop:CDExact} is uniquely determined by the conditions \eqref{Eq:Symmetry}, \eqref{Eq:DispInvarianceforShifts}, \eqref{Eq:DispLocScale} and \eqref{Eq:ShiftSymDist}.
\end{enumerate}
\end{theorem}

Theorem~\ref{thm:Uniqueness} shows that our AVM and CD decompositions are unique amongst all possible decompositions in jointly satisfying the symmetry in \eqref{Eq:Symmetry}, that dispersion terms are invariant to translations in \eqref{Eq:DispInvarianceforShifts}, and that for (symmetric) location-scale families, non-negativity of the shift/dispersion terms agrees with the ordering of the location/scale terms in \eqref{Eq:ShiftSymDist}--\eqref{Eq:ShiftLocScale}.
As these properties are very natural to stipulate, Theorem~\ref{thm:Uniqueness} strongly supports the particular form of our decompositions.
The uniqueness only applies to their integrated form, whereas the $\alpha$-wise decomposition can of course be changed on (Lebesgue) null sets without affecting the resulting integral.
A similar uniqueness condition cannot be established for the $\WD_p$ decomposition for $p > 1$ as the shift invariance of Proposition \ref{prop:DispInvarianceforShifts} cannot be invoked for $p > 1$.

We finally consider the case of normal distributions, which, arguably, form the most prominent location-scale family.
Let $F = \norm(\mu_F,\sigma_F^2)$ and $G = \norm(\mu_G,\sigma_G^2)$ be normal distributions with means $\mu_F$ and $\mu_G$ and variances $\sigma_F^2$ and $\sigma_G^2$, respectively. 
In what follows, we use the shorthand notations $\mudiff = |\mu_F - \mu_G|$, $\sigmadiff = |\sigma_F - \sigma_G|$ and $\sigmaavg = \sqrt{\sigma_F^2 + \sigma_G^2}$. 
We further denote the  standard normal density and distribution function by $\phi(\cdot)$ and $\Phi(\cdot)$, respectively.
Here, we provide closed-form expressions for the decompositions of the $\AV$ and the $\CD$. More involved formulas for the $\WD_p$ with $p \in \mathbb{N}$ are given in Supplement~\ref{sec:WDpNormals}.

\sloppy
First, in the special case where $\sigmadiff = 0$, we get $\AV(F,G) = \mudiff$, which is entirely attributed to a single shift component.
If $\sigmadiff \not= 0$, we get the closed-form expression
\begin{align*}
\AV(F,G) &= \mudiff \big( 2\Phi(\mudiff/\sigmadiff) - 1 \big)  + 2\sigmadiff \phi(\mudiff/\sigmadiff).
\end{align*}
For its decomposition terms, if $\sigma_F > \sigma_G$, then 
\begin{align*}
\disp_-^{\AV}(F,G) = 0
\quad \text{and} \quad
\disp_+^{\AV}(F,G) = 2\sigmadiff\phi(0).
\end{align*}
Furthermore (irrespective of whether $\sigma_F > \sigma_G$ holds), if $\mu_F > \mu_G$, then  $\shift_-^{\AV}(F,G)  = 0$ and 
\begin{align*}    
\shift_+^{\AV}(F,G) = \mudiff \big(2\Phi(\mudiff/\sigmadiff) - 1\big) + 2\sigmadiff \big(\phi(\mudiff/\sigmadiff) - \phi(0) \big).
\end{align*}
Hence, the shift and dispersion components agree, as expected, with the signs of differences in means and standard deviations, respectively.

For the $\CD$ of two normal distributions, we get the expression
\begin{equation*}
\CD(F,G) = 2\sigmaavg\phi(\mudiff/\sigmaavg) + \mudiff \big( 2\Phi(\mudiff/\sigmaavg) - 1 \big)
- \sqrt{2}\phi(0)\big(\sigma_F + \sigma_G\big).
\end{equation*}
For its components, we obtain that if $\sigma_F \geq \sigma_G$, then $\disp_-^{\CD}(F,G) = 0$ and 
\begin{align*}
\disp_+^{\CD}(F,G) = 
2\sigmaavg\phi(0) - \sqrt{2}\phi(0)(\sigma_F + \sigma_G).
\end{align*}
Furthermore, if $\mu_F \geq \mu_G$, then $\shift_-^{\CD}(F,G) = 0$ and 
\begin{align*}
\shift_+^{\CD}(F,G) = 
\mudiff \big( 1-2\Phi(-\mudiff/\sigmaavg) \big) - 2\tau \big(\phi(0) - \phi(\mudiff/\sigmaavg) \big).
\end{align*}
Thus, the components of the $\CD$ also behave as expected for normal distributions.

\subsection{Agreement and differences across distance measures}
\label{subsec:ComparisonDecompositions}

We now analyze how the decompositions of the considered divergence measures are related.
We first show that they mostly agree on which components are nonzero.
Subsequently, we use our decompositions to assess the relative importance of differences in location and dispersion of the analyzed distributions on the respective distance measures.
Our results in the latter context shed light on fundamental differences between the distance measures under consideration.

Our first result shows that the decompositions of all considered divergences agree on which \emph{dispersion} components are nonzero, without imposing any assumptions.

\begin{prop}[Positive dispersion]
\label{prop:PositiveDispersion}
For any two distributions $F$ and $G$, and $p \in [1,\infty)$, it holds that
\begin{align*}
\disp_\pm^{\WDp}(F,G) > 0 
\quad\Longleftrightarrow\quad 
\disp_\pm^{\AV}(F,G) > 0 
\quad\Longleftrightarrow\quad 
\disp_\pm^{\CD}(F,G) > 0. 
\end{align*}
\end{prop}

Proposition \ref{prop:PositiveDispersion} also implies that $\disp_\pm^{\WD_p}(F,G) > 0  \; \Longleftrightarrow \; \disp_\pm^{\WD_q}(F,G) > 0$ for any $p\not=q$ by simply invoking the first equivalence for differing $p$ and $q$.

A similar concordance property also holds for the shift components, however, without the equivalence with the CD components.

\begin{prop}[Positive shift]
\label{prop:PositiveShift}
For any two distributions $F$ and $G$, and $p \in [1,\infty)$, it holds that
\begin{align*}
\shift_\pm^{\WDp}(F,G) > 0 
\quad\Longleftrightarrow\quad 
\shift_\pm^{\AV}(F,G) > 0
\quad\Longrightarrow\quad 
\shift_\pm^{\CD}(F,G) > 0. 
\end{align*}
\end{prop}

In the following example, a positive shift component arises in the $\CD$ decomposition while the $\WD_p$ decomposition(s) include zero shift, which illustrates the missing equivalence between $\CD$ and $\AV$ in Proposition \ref{prop:PositiveShift}.

\begin{figure*}
	\centering
	\subcaptionbox{AVM decomposition}{
		\includegraphics[width = \defaultfigwidth\textwidth,page = 1]{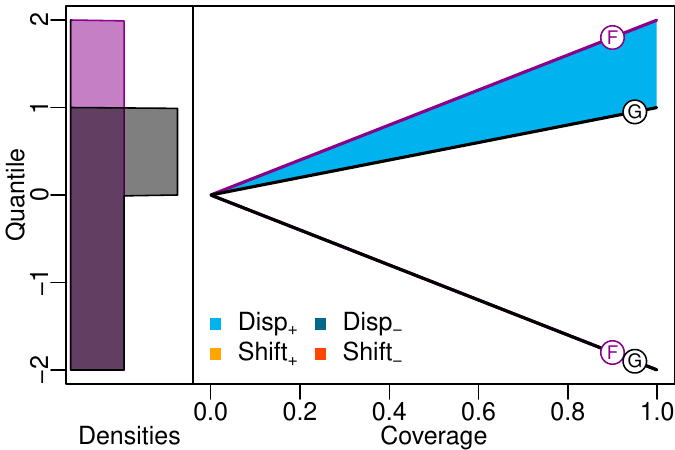}}
	\subcaptionbox{CD decomposition slice}{
		\includegraphics[width = \defaultfigwidth\textwidth,page = 7]{figures/Fig5.pdf}}
	\caption{Illustration of the $\AV$ and $\CD$ (for $\beta=0.6$) decompositions for distributions $F$ and $G$ given in Example~\ref{Ex:DiffShift} for which only the $\CD$ exhibits a nonzero shift component.
		See Figures~\ref{fig:Illustration_FoldedQuantiles} and \ref{fig:illustration_cd} for detailed descriptions of the plots.}
	\label{Fig:DiffShift}
\end{figure*}

\begin{example}
\label{Ex:DiffShift}
Consider the distributions
$F= \mathcal{U}[-2,2]$ and $G = 0.5 \times \mathcal{U}[-2,0] + 0.5 \times \mathcal{U}[0,1]$, whose densities and folded quantile functions are shown in Figure~\ref{Fig:DiffShift}.
While the $\AV = \disp_+^{\AV}(F,G) = 0.25$ is attributed entirely to an increase in dispersion of $F$ relative to $G$, the $\CD(F,G) = \shift_+^{\CD}(F,G) + \disp_+^{\CD}(F,G) = 0.02083 + 0.02083$  is attributed to both shift and dispersion terms. 
(Notice for the later use of this example in Section~\ref{subsec:WDOrder} that the distribution $F$ is strictly larger than $G$ in the usual stochastic order.)
\end{example}

While the fact that the decompositions (almost) agree on which components are nonzero is conceptually reassuring, the nonzero components differ in (relative) magnitude.
The following results shed light on the sensitivity of the considered divergence measures towards changes in shift and dispersion.

\begin{theorem}
\label{Thm:DispInequality_pWD}
Let $F$ and $G$ be symmetric distributions, $F\neq G$, and $p,q \in [1,\infty)$.
If $q > p$, it holds that
\begin{align}
\label{Eq:DispInequality_pWD}
\frac{\disp_+^{\WD_p} + \disp_-^{\WD_p}}{\WD_p(F,G)} 
\; \leq \; 
\frac{\disp_+^{\WD_q}+ \disp_-^{\WD_q}}{\WD_q(F,G)}.
\end{align}
\end{theorem}

Theorem \ref{Thm:DispInequality_pWD} shows that for symmetric distributions, the relative weight that the $p$-Wasserstein distance assigns to the dispersion components increases in its order $p$.
Hence, Wasserstein distances of higher orders emphasize differences in dispersion as opposed to differences in location.

The symmetry condition in Theorem \ref{Thm:DispInequality_pWD} is necessary to obtain a rigorous connection, as illustrated by
Example \ref{Ex:pWD_loc-scale} 
in Supplement~\ref{sec:ExamplesRelativeComparison}, 
which shows that inequality \eqref{Eq:DispInequality_pWD} is not guaranteed to hold for asymmetric distributions, even when restricting attention to a location-scale family. 
Furthermore, Example \ref{Ex:pWD_unimodal} 
shows that the inequality also does not generally hold for unimodal distributions. 
Such examples are however rarely encountered in practice.

We now turn to a relative comparison of the $\AV$ and the $\CD$, which we formally establish for normal distributions using the closed-form expressions given at the end of Section \ref{subsec:BasicProperties}.

\begin{theorem}
\label{Thm:DispInequalities}
Let $F = \norm(\mu_F,\sigma_F^2)$ and $G = \norm(\mu_G,\sigma_G^2)$ be normal distributions with $F\neq G$. 
Then, 
\begin{align}
\label{Eq:DispInequality_CD-WD_norm} 
\frac{\disp_+^{\CD} + \disp_-^{\CD}}{\CD(F,G)} 
\; \leq \; 
\frac{\disp_+^{\AV} + \disp_-^{\AV}}{\AV(F,G)}.
\end{align}
\end{theorem}

Theorem \ref{Thm:DispInequalities} establishes that for normal distributions, the $\AV$  puts a higher emphasis on differences in dispersion in contrast to the $\CD$.
The relation in \eqref{Eq:DispInequality_CD-WD_norm} is typically found in practice for other distributions as well, and counterexamples appear to be rare. Nonetheless, Example \ref{Ex:CD-WD_sym} 
shows that the inequality does not hold for arbitrary symmetric distributions (that are sufficient for the related statement in Theorem~\ref{Thm:DispInequality_pWD}).
We hypothesize that the following generalization of Theorem \ref{Thm:DispInequalities} to (possibly asymmetric) unimodal distributions may hold.

\begin{conj}
\label{conj:AVCD_ineq_unimodal}
Let $F$ and $G$ be unimodal distributions in the sense that the quantile functions are differentiable almost everywhere and the derivatives $(F^{-1})'$ and $(G^{-1})'$ are decreasing functions for $\alpha < \frac12$ and increasing ones for $\alpha > \frac12$.
Then, it may hold that
\begin{align*}
\frac{\disp_+^{\CD} + \disp_-^{\CD}}{\CD(F,G)} 
\; \leq \; 
\frac{\disp_+^{\AV} + \disp_-^{\AV}}{\AV(F,G)}.
\end{align*}
\end{conj}

Such a result might appear counter-intuitive at first sight, as the shift components of the CD feature a factor of one-half that cancels out in the shift components of the AVM. 
The double integral and the capping of the CD dispersion component integrals at $\beta$ apparently more than compensate for this halving.

\section{Compatibility with stochastic order relations}
\label{sec:Orders}

\begin{table*}[tb]
\scriptsize
\begin{tabular}{llll}
\toprule 
Order name & Type & Symbol & Definition \\
\midrule 
\addlinespace \addlinespace
\multicolumn{4}{l}{\textbf{\footnotesize{Dispersive orders in Section~\ref{Sec:DispOrder}:}}} \\ 
\midrule 
Dispersive  & preorder & $F \geq_\mathrm{D} G$ & $F^{-1}(\tau) - F^{-1}(\xi) \; \geq \;  G^{-1}(\tau) - G^{-1}(\xi) \quad \forall \, 0<\xi<\tau<1$ \\
% Dispersive order (strict) &  & $F >_\mathrm{D} G$ &  \color{red} ??? $F \ge_\mathrm{D} G$ and $F \not\leq_\mathrm{D} G$ \\
\midrule 
Weak dispersive & preorder &  $F \geq_\mathrm{wD} G$ & $ F^{-1}(\tau) - F^{-1}(1-\tau) \; \geq \; G^{-1}(\tau) - G^{-1}(1-\tau) \quad \forall \, 0.5 < \tau < 1$ \\
% Weak dispersive order (strict) &  &   \\
\midrule 
\addlinespace \addlinespace
\multicolumn{4}{l}{\textbf{\footnotesize{Stochastic orders in Section~\ref{Sec:StochOrder}:}}} \\
\midrule 
Usual stochastic & partial order & $F \geq_\mathrm{S} G$ &  $\min \big\{ F^{-1}(\tau) - G^{-1}(\tau), \, F^{-1}(1-\tau) - G^{-1}(1-\tau) \big\} \; \geq \; 0 \quad \forall \, 0.5 < \tau < 1$  \\
%$F^{-1}(\alpha) \geq G^{-1}(\alpha) \quad \forall 0 < \alpha < 1$ \\
% Usual stochastic order (strict) & & $F >_\mathrm{S} G$ &  \\
\midrule 
Weak stochastic & preorder (cond.) & $F \geq_\mathrm{wS} G$ & $\max\{F^{-1}(\tau) - G^{-1}(\xi), \, F^{-1}(1-\tau) - G^{-1}(1-\xi)\} \; \geq \; 0 \quad \forall \, 0.5 < \tau,\xi < 1$ \\
% Weak stochastic order (strict) &  & $F >_\mathrm{wS} G$ & \\
\midrule
Relaxed stochastic & preorder (cond.)  & $F \geq_\mathrm{rS} G$ & $\max\{F^{-1}(\tau) - G^{-1}(\tau), \, F^{-1}(1-\tau) - G^{-1}(1-\tau)\} \; \geq \; 0 \quad \forall \, 0.5 < \tau  < 1$ \\
% Relaxed stochastic order (strict) & & $F >_\mathrm{rS} G$  \\
\midrule 
Strong stochastic & strict partial order & $F >_\mathrm{sS} G$ & $F \geq_\mathrm{S} G$ and $\exists \tau \in[0,1]: \;\; \min\{F(\tau) - G(\tau), \, F(1-\tau) - G(1-\tau)\} \; > \; 0$ \\
\bottomrule
\end{tabular}
\caption{Overview of the order relations used in this paper. All \emph{strict} relations $F >_\bullet G$ are defined as $F \geq_\bullet G$ and $F \not\leq_\bullet G$. The unconventional (but equivalent; see \eqref{Eq:StochOrder}) definition of the usual stochastic order is used to highlight the similarity to the other stochastic order relations.
	The bracket ``(cond.)'' stands for conditions that are further discussed in Propositions \ref{Thm:WeakStochPreorder} and \ref{Thm:RelaxedSOPreorder}. 
} 
\label{tab:OrderRelations}
\end{table*}

In this section, we establish connections between our decomposition terms and some well-known stochastic orders \citep[e.g.,][]{Muller2002, Shaked2007}. 
We first show that our dispersion components align well with the dispersive order in Section~\ref{Sec:DispOrder}. Subsequently, we relate our shift components to the usual stochastic order in Section~\ref{Sec:StochOrder}. 
Table \ref{tab:OrderRelations} provides an overview of all order relations used in this section. Here, we refer to orders that are related to the usual stochastic order as stochastic and use the term dispersive orders for stochastic variability orders. 
Throughout the section, we formulate the order conditions in terms of quantile levels $\tau$ and $\xi$ for ease of exposition. Formulations in terms of coverage levels $\alpha$ and $\beta$ that align more closely with the decomposition terms are obtained by replacing $\tau$ with $\frac{1+\alpha}{2}$ and $\xi$ with $\frac{1+\beta}{2}$.

In what follows, we limit our analysis to probability distributions $F$ and $G$ with \emph{continuous} quantile functions $F^{-1}$ and $G^{-1}$, respectively. As alluded to at the end of Section \ref{Sec:DispOrder}, the results can be generalized to arbitrary distributions by slightly adapting the order relations to account for discontinuities in the quantile functions.

\subsection{Connections to the dispersive order}
\label{Sec:DispOrder}

Here, we establish logical connections of (non-)zero dispersion components to order relations capturing differences in dispersion from the literature. The findings of this section are summarized in Figure \ref{Fig:OrdersDisp}.

\begin{figure}[tb]
\newcommand\SEarrow{\mathrel{\rotatebox[origin=c]{-45}{$\Longrightarrow$}}}
\newcommand\SWarrow{\mathrel{\rotatebox[origin=c]{45}{$\Longleftarrow$}}}
\newcommand\Longdownarrow{\mathrel{\rotatebox[origin=c]{90}{$\Longleftarrow$}}}
\begin{center}
\begin{tabular}{ccc}
$F >_\text{D} G$ & $\Longrightarrow$ &  \fbox{\begin{tabular}{@{}c@{}}
		$\; F >_\text{wD} G \; \Longleftrightarrow \;
		\footnotesize 
		\begin{Bmatrix}
			\disp_-^{\D}(F,G) = 0 \\[0.5mm]
			\disp_+^{\D}(F,G) > 0 
		\end{Bmatrix}$ \\
\end{tabular}}
\\[5mm]
$\Longdownarrow$ & & $\Longdownarrow$ 
\\[3mm]
$F \geq_\text{D} G$ & $\Longrightarrow$ &          
\fbox{\begin{tabular}{@{}c@{}}
		$\; F \geq_\text{wD} G \;  \Longleftrightarrow \; 
		\disp_-^{\D}(F,G) = 0 \;$ \\
\end{tabular}}
\end{tabular}
\end{center}
\caption{Overview of the logical implications between the studied dispersive orders and the dispersion components of our decompositions.}
\label{Fig:OrdersDisp}
\end{figure}

The distribution $F$ is said to be larger than $G$ in \emph{dispersive order}, $F \geq_\text{D} G$, if for all $0 < \xi < \tau < 1$,
\begin{align*}
F^{-1}(\tau) - F^{-1}(\xi) \; \geq \; G^{-1}(\tau) - G^{-1}(\xi).
\end{align*}
Note that the dispersive order is a preorder (which is reflexive and transitive) and not a partial order (which is reflexive, transitive and antisymmetric).
As common for order relations, we define the \emph{strict} relation $F >_\mathrm{D} G$ through $F \geq_\mathrm{D} G$ and $F \not\leq_\mathrm{D} G$, and equivalently for all other strict relations considered in this section.

The dispersive order turns out to be too strong for a logical equivalence with zero dispersion components to hold. 
To establish such an equivalence, we consider a weaker order relation.
The distribution $F$ is said to be larger than $G$ in \emph{weak dispersive order}, $F \geq_\mathrm{wD} G$, if for all $0.5 <  \tau < 1$,
\begin{equation}
\label{Eq:wD}
F^{-1}(\tau) - F^{-1}(1-\tau) \; \geq \; G^{-1}(\tau) - G^{-1}(1-\tau).
\end{equation}
This order is also known under the name \emph{quantile spread order} \citep{TC05}. It is easy to see from its definition that the weak dispersive order is a preorder.
It gives rise to the following characterization result.

\begin{theorem}
\label{Thm:DispOrderJoint}
Let $F$ and $G$ be probability distributions with continuous quantile functions $F^{-1}$ and $G^{-1}$, respectively. 
Then, for any $\D \in \{\AV, \WD_p, \CD\}$, $p\in[1,\infty)$, we have that
\begin{enumerate}[label=(\alph*)]
\item 
\label{item:wDequivalence}
$F \geq_\mathrm{wD} G \quad \Longleftrightarrow \quad \disp_-^{\D}(F,G) = 0$;

\item  
\label{item:StrictwDequivalence}
$F >_\mathrm{wD} G \quad \Longleftrightarrow\quad \big\{\disp_-^{\D}(F,G) = 0  \; \text{ and } \; \disp_+^{\D}(F,G) > 0 \big\}$.
\end{enumerate}
\end{theorem}

Subject to the regularity condition of continuous quantile functions, which is further discussed below, part (a) of the theorem establishes an equivalence between a weak dispersive ordering of the distributions and a zero dispersion component in our decompositions.
As shown in part (b), the equivalence extends naturally to an equivalence between a corresponding strict ordering and a unique nonzero dispersion component. Notably, the equivalences and implications in Theorem \ref{Thm:DispOrderJoint} hold for all considered divergence measures $\D \in \{\AV, \WD_p, \CD\}$, $p\in[1,\infty)$, by the equivalences from Proposition \ref{prop:PositiveDispersion}.

We continue to analyze the properties of the weak dispersive order.
The following proposition shows that it is implied by the dispersive order.
\begin{prop}
\label{Prop:DispOrder}
Let $F$ and $G$ be probability distributions with continuous quantile functions $F^{-1}$ and $G^{-1}$, respectively. Then,
\begin{enumerate}
\item $F \geq_{\textnormal{D}} G\quad\Longrightarrow\quad F \geq_{\textnormal{wD}} G$;
\item $F >_{\textnormal{D}} G\quad\Longrightarrow\quad F >_{\textnormal{wD}} G$.
\end{enumerate}
\end{prop}
Notice that the implication in (b) is not trivial as if one order implies another, this does not necessarily mean that the same is true for the respective \emph{strict orders}. 

As summarized in Figure \ref{Fig:OrdersDisp}, the previous two results jointly show that a strict dispersive ordering implies a unique nonzero dispersion component in our decompositions.
However, having a single nonzero dispersion term does not imply a dispersive ordering, even in its non-strict form.
For example, the distributions $F$ and $G$ in Example \ref{Ex:IntransitiveCD} 
are not dispersively ordered despite the unique nonzero dispersion component. 

Theorem~\ref{Thm:DispOrderJoint} requires
continuous quantile functions for $F$ and $G$.
The necessity of this condition is illustrated in the following example.

\begin{figure}
\centering
\includegraphics[width = \defaultfigwidth\textwidth]{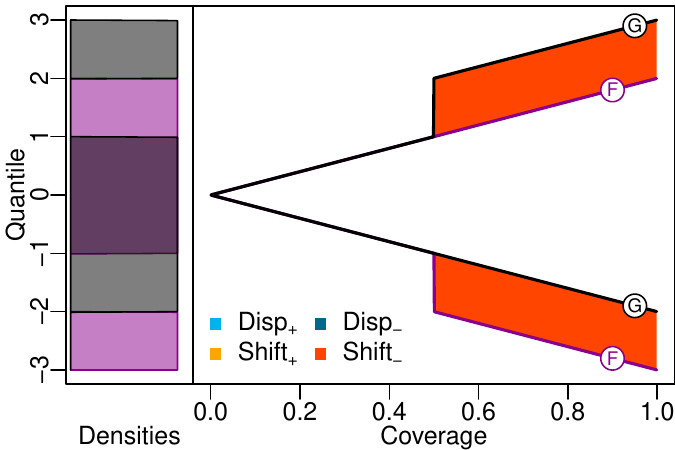}
\caption{Illustration of the AVM decomposition for distributions $F$ and $G$ given in Example~\ref{Ex:StrictDispOrder} with discontinuous quantile functions such that $F$ is strictly larger than $G$ in weak dispersive order. See Figure \ref{fig:Illustration_FoldedQuantiles} for a detailed description of the plot.}
\label{Fig:StrictDispOrder}
\end{figure}

\begin{example}
Consider two mixtures of uniform distributions, $F = \frac14 \times \unif[-3,-2] + \frac34 \times \unif[-1,2]$ and $G = \frac34 \times \unif[-2,1] + \frac34 \times \unif[2,3]$, with discontinuous quantile functions illustrated in Figure \ref{Fig:StrictDispOrder}. The distribution $F$ is strictly larger than $G$ in weak dispersive order as the defining inequality is strict for $\tau = 0.75$ by left-continuity of the quantile functions, which results in a sort of asymmetric right-continuity in the lower part of the quantile spread plot. Despite the ordering, the decomposition does not produce a nonzero dispersion term, in clear contrast to part \ref{item:StrictwDequivalence} of Theorem \ref{Thm:DispOrderJoint}. 
\label{Ex:StrictDispOrder}
\end{example}

The issue encountered in Example~\ref{Ex:StrictDispOrder} could be avoided by slightly adapting the definition of the weak dispersive order in \eqref{Eq:wD} to require for any $0.5 \leq \tau <1$ that
\begin{align*}
\sup \{x\mid F(x) \leq \tau \} - \inf\{x\mid F(x) \leq 1-\tau\}
\geq \sup\{x\mid G(x) \leq \tau \} - \inf\{x\mid G(x) \leq 1-\tau\}.
\end{align*}
This ensures that the lower and upper quantiles (as a function of the coverage) in the quantile spread plot are both right-continuous resulting in a symmetric directional continuity. 
We refrained from doing so in our analysis for ease of exposition and to align with the prevalent notion of the weak dispersive (or quantile spread) order.

\subsection{Connections to the usual stochastic order}
\label{Sec:StochOrder}

\begin{figure*}[tb]
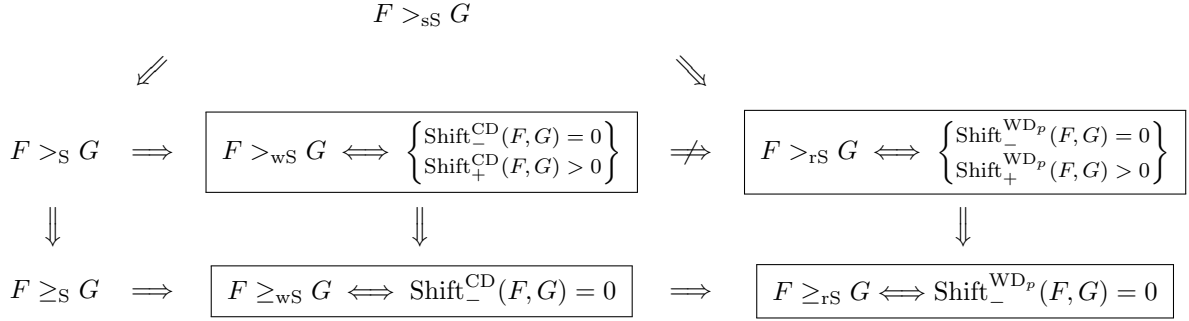

\newcommand\SEarrow{\mathrel{\rotatebox[origin=c]{-45}{$\Longrightarrow$}}}
\newcommand\SWarrow{\mathrel{\rotatebox[origin=c]{45}{$\Longleftarrow$}}}
\newcommand\Longdownarrow{\mathrel{\rotatebox[origin=c]{90}{$\Longleftarrow$}}}
\centering
\begin{center}
\begin{tabular}{ccccc}
&& $F >_\text{sS} G$ && 
\\[3mm]
& $\SWarrow$ & 
& $\SEarrow$ & 
\\[3mm]
$F >_\text{S} G$ & $\Longrightarrow$ & 
\fbox{\begin{tabular}{@{}c@{}}
		$\; F >_\text{wS} G \; \Longleftrightarrow \; 
		\footnotesize 
		\begin{Bmatrix}
			\shift_-^{\CD}(F,G) = 0 \\[0.5mm]
			\shift_+^{\CD}(F,G) > 0
		\end{Bmatrix}\; $ \\
\end{tabular}}
& $\centernot\Longrightarrow$ & 
\fbox{\begin{tabular}{@{}c@{}}
		$\; F >_\text{rS} G \; \Longleftrightarrow \;
		\footnotesize 
		\begin{Bmatrix}
			\shift_-^{\WDp}(F,G) = 0  \\[0.5mm]
			\shift_+^{\WDp}(F,G) > 0
		\end{Bmatrix} \;$ \\
\end{tabular}}
\\[5mm]
$\Longdownarrow$ && $\Longdownarrow$ && $\Longdownarrow$ 
\\[3mm]
$F \geq_\text{S} G$ & $\Longrightarrow$ & 
\fbox{\begin{tabular}{@{}c@{}}
		$\; F \geq_\text{wS} G \; \Longleftrightarrow \;
		\footnotesize 
		\shift_-^{\CD}(F,G) = 0 \;$ \\
\end{tabular}}
& $\Longrightarrow$ &
\fbox{\begin{tabular}{@{}c@{}}
		$\; F \geq_\text{rS} G \Longleftrightarrow 
		\footnotesize 
		\shift_-^{\WDp}(F,G) = 0 \;$ \\
\end{tabular}}
\end{tabular}
\end{center}
\caption{Connections between the stochastic order relations introduced in Section~\ref{Sec:StochOrder} and the shift components of the $\CD$ and $\WDp$ divergence measures (for any $p \in [1,\infty)$).}
\label{Fig:Orders}
\end{figure*}

Establishing connections of (non)zero shift components to the usual stochastic order is more complicated than the connections to the dispersive orderings outlined in the previous Section \ref{Sec:DispOrder}.
In a nutshell, zero shift components in the Cram\'er and Wasserstein decompositions give rise to a \emph{weak} and a \emph{relaxed} form of the stochastic order relation, respectively.
Hence, we treat the Cram\'er and Wasserstein distances separately in the following Sections~\ref{subsec:CDOrder} and \ref{subsec:WDOrder}.
The results of this section are summarized in Figure \ref{Fig:Orders}.

The distribution $F$ is said to be larger than $G$ in the \emph{usual stochastic order}, $F \geq_\mathrm{S} G$, if $F(x) \leq G(x)$ for all $x \in \mathbb{R}$, or, equivalently, if $F^{-1}(\tau) \geq G^{-1}(\tau)$ for all $0 < \tau < 1$.
The usual stochastic order is a partial order, and its definition can be reformulated as
\begin{equation}
\label{Eq:StochOrder}
\min \big\{ F^{-1}(\tau) - G^{-1}(\tau), \, F^{-1}(1-\tau) - G^{-1}(1-\tau) \big\} \; \geq \; 0 
\end{equation}
for all $0.5 <  \tau < 1$.
The term used for this unconventional characterization arises in the shift components (see \eqref{Eq:AVM_Shift} with $\tau = \frac{1+\alpha}{2}$), thereby serving as a natural starting point to investigate the connection. 

While $F \geq_\mathrm{S} G$ implies $\shift_-^{\D}(F,G) = 0$ for all the considered distances $\D \in \{\AV, \WD_p, \CD\}$, the usual stochastic order is too strong to establish an \emph{equivalence} with zero shift components.
Hence, the following two subsections present two relaxations of the usual stochastic order pertaining to the Cram\'er and Wasserstein distances, respectively.

\subsubsection{A weak stochastic order based on the Cram\'er distance}
\label{subsec:CDOrder}

To establish an equivalence relation with the shift components of the $\CD$, we define the \emph{weak stochastic order} in which the distribution $F$ is larger than $G$, $F \geq_\mathrm{wS} G$, if 
\begin{equation}
\label{Eq:wS}
\max \big\{F^{-1}(\tau) - G^{-1}(\xi), \, F^{-1}(1-\tau) - G^{-1}(1-\xi) \big\} \; \geq \; 0
\end{equation}
for all $0.5 < \tau,\xi < 1$.
To the best of our knowledge, the weak stochastic order is a new order relation, and we discuss its properties after establishing its equivalence with the shift components of the $\CD$.

\begin{theorem}
\label{Thm:StochOrderCDJoint}
Let $F$ and $G$ be probability distributions with continuous quantile functions $F^{-1}$ and $G^{-1}$, respectively. Then,
\begin{enumerate}
\item $F \geq_\mathrm{wS} G \quad \Longleftrightarrow \quad \shift_-^{\CD}(F,G) = 0$;
\item $F >_\mathrm{wS} G \quad \Longleftrightarrow \quad\big\{ \shift_-^{\CD}(F,G) = 0 \; \text{ and } \;  \shift_+^{\CD}(F,G) > 0 \big\}$.
\end{enumerate}
\end{theorem}

The theorem establishes the equivalence of the (strict) weak stochastic order and (non)zero shift components of the $\CD$, akin to Theorem~\ref{Thm:DispOrderJoint}.

The weak stochastic order is implied by the usual stochastic order, as detailed by the following proposition.

\newpage
\begin{prop}
\label{Prop:StochOrderCD}
Let $F$ and $G$ be probability distributions with continuous quantile functions $F^{-1}$ and $G^{-1}$, respectively. Then,
\begin{enumerate}
\item $F \geq_\mathrm{S} G \quad \Longrightarrow \quad  F \geq_\mathrm{wS} G$;
\item $F >_\mathrm{S} G  \quad \Longrightarrow \quad 
F >_\mathrm{wS} G$.
\end{enumerate}
\end{prop}

As summarized in Figure \ref{Fig:Orders}, the previous results show that a strict stochastic ordering implies a unique nonzero shift component in the $\CD$ decomposition.
On the other hand, having a
single nonzero shift term in the $\CD$ does not imply a usual stochastic ordering, even in its non-strict form. For example,
the distributions $F$ and $G$ in Example \ref{Ex:IntransitiveCD} 
are not stochastically ordered despite a unique nonzero shift
component.

The weak stochastic order is a preorder for relatively broad classes of distributions.

\begin{prop}
\label{Thm:WeakStochPreorder}
The weak stochastic order is a preorder on sets of distributions with common support and continuous quantile functions.
\end{prop}

Example \ref{Ex:IntransitiveCD} 
illustrates the necessity of the common support assumption in Proposition~\ref{Thm:WeakStochPreorder}.

\subsubsection{A relaxed stochastic order based on Wasserstein distances}
\label{subsec:WDOrder}

Similar to the treatment of the Cram\'er distance, we require a weaker form of a stochastic order relation to establish equivalence with nonzero shift components of the $\WDp$ decompositions.
To this end, we call $F$ larger than $G$ in \emph{relaxed stochastic order}, $F \geq_\mathrm{rS} G$, if
\begin{align}
\label{Eq:rS}
\max \big\{ F^{-1}(\tau) - G^{-1}(\tau), \, F^{-1}(1-\tau) - G^{-1}(1-\tau) \big\} \geq 0 
\end{align}
for all $0.5 < \tau < 1$.
To the best of our knowledge, the relaxed stochastic order relation is new to the literature, and we briefly discuss its properties at the end of this section.

\begin{theorem}
\label{Thm:StochOrderAVMJoint}
Let $F$ and $G$ be probability distributions with continuous quantile functions $F^{-1}$ and $G^{-1}$, respectively. Then, for all $p\in[1,\infty)$,
\begin{enumerate}
\item $F \geq_\mathrm{rS} G \quad \Longleftrightarrow \quad \shift_-^{\WDp}(F,G) = 0$;
\item $F >_\mathrm{rS} G \quad \Longleftrightarrow \quad\big\{\shift_-^{\WDp}(F,G) = 0, \quad \shift_+^{\WDp}(F,G) > 0  \big\}$.
\end{enumerate}
\end{theorem}

This theorem establishes an equivalence between the (strict) relaxed stochastic order and (non)zero shift components of the $\WDp$ shift components. Notably, the equivalences and implications in Theorem \ref{Thm:StochOrderAVMJoint} hold for the Wasserstein distance of any order $p$ by the equivalence from Proposition \ref{prop:PositiveShift}.

For the Wasserstein distances, a further complication arises because a \emph{strict} stochastic ordering does not imply a nonzero shift component, as illustrated by Example \ref{Ex:DiffShift}. To establish a result akin to part (b) of Proposition~\ref{Prop:StochOrderCD}, the strict stochastic order needs to be further strengthened. To this end, we define the \emph{strong stochastic order}, $F >_\mathrm{sS} G$, as
\begin{gather}
\label{Eq:sS}
F \geq_\mathrm{S} G
\qquad \text{ and } \qquad 
\exists \tau \in (0,1):
\min \{ F^{-1}(\tau) - G^{-1}(\tau), F^{-1}(1-\tau) - G^{-1}(1-\tau)\} > 0.
\end{gather}
This relation is a \emph{strict} partial order (irreflexive, antisymmetric and transitive) and implies the \emph{strict} version of the usual stochastic order, which arises when replacing the minimum in \eqref{Eq:sS} by a maximum. 

\begin{prop}
\label{Prop:StochOrderAVM}
Let $F$ and $G$ be probability distributions with continuous quantile functions $F^{-1}$ and $G^{-1}$, respectively. Then,
\begin{enumerate}
\item $F \geq_\mathrm{wS} G \quad~\Longrightarrow\quad F \geq_\mathrm{rS} G$;
\item $F >_\mathrm{sS} G  \quad\Longrightarrow \quad 
F >_\mathrm{rS} G$.
\end{enumerate}
\end{prop}

Part (a) of this proposition shows that the relaxed stochastic order is implied by the weak stochastic order, and hence, by invoking Proposition~\ref{Prop:StochOrderCD}, also by the usual stochastic order.
However, a corresponding implication fails to hold for their strict versions, as the distributions in Example \ref{Ex:DiffShift} illustrate.

As summarized in Figure \ref{Fig:Orders}, part (b) of  Proposition~\ref{Prop:StochOrderAVM} together with Theorem~\ref{Thm:StochOrderAVMJoint} shows that a strong stochastic ordering implies a unique nonzero shift component in the $\WDp$ decompositions.

We now establish that the relaxed stochastic order is a preorder under symmetry.

\begin{prop}
\label{Thm:RelaxedSOPreorder}
The relaxed stochastic order is a preorder on sets of symmetric distributions with continuous quantile functions.
\end{prop}

Example \ref{Ex:Intransitive} 
shows that the common support assumption (that was imposed in Proposition~\ref{Thm:WeakStochPreorder}) is not sufficient to establish transitivity for the \emph{relaxed} stochastic order.

\section{Applications}
\label{sec:applications}

Here, we illustrate the decompositions in two applications from the fields of climate science and economic survey design.

\subsection{Prediction of seasonal temperature extrema}

Our first application is from the field of climate modelling and illustrates how our decompositions can render forecast evaluations more interpretable. We revisit an evaluation of historical climate simulations from the Coupled Model Intercomparison Project (CMIP, \citealt{Taylor2012}) performed by \cite{Thorarinsdottir2020}. Our focus is on monthly maximum temperatures (TXx) over Europe during the Boreal summer months June, July, and August.

\begin{figure*}[tb!]
\centering
\subcaptionbox{AVM decompositions}{
\includegraphics[width=0.48\textwidth]{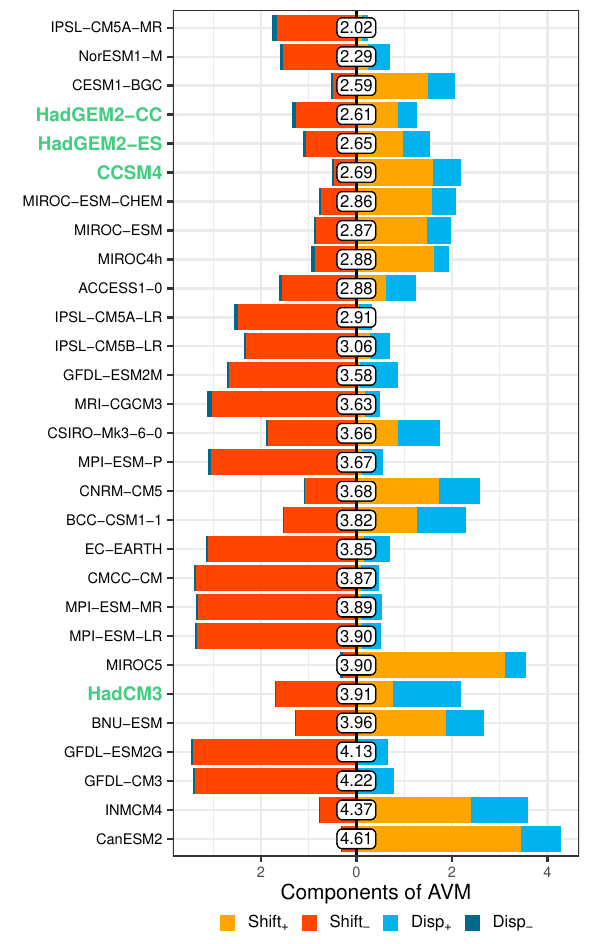}} \quad\subcaptionbox{CD decompositions}{\includegraphics[width=0.48\textwidth]{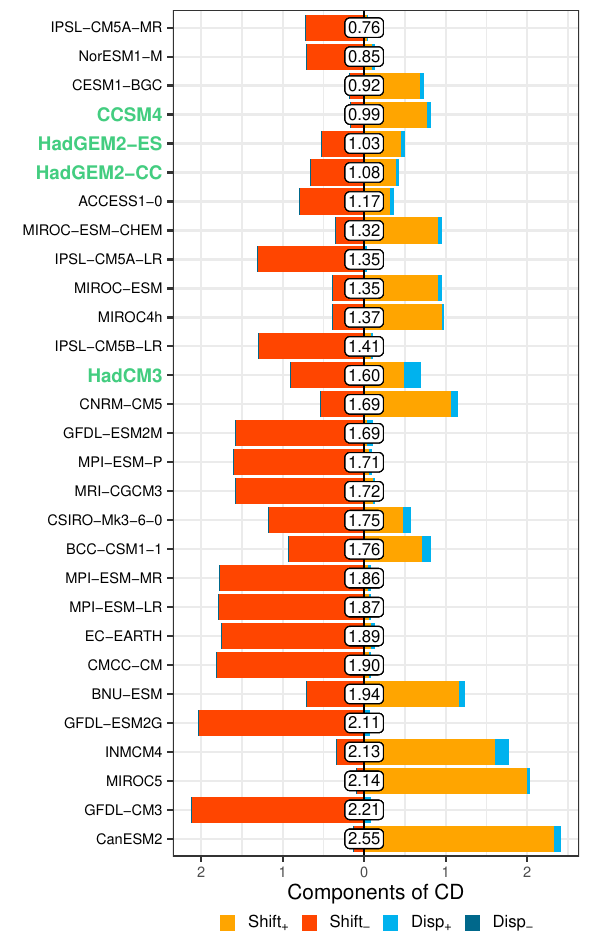}}
\caption{Evaluation of forecasts for monthly temperature maxima (TXx) in summer months from 29 climate models (given in the rows), averaged over 136 grid cells in Europe, sorted by their average (across grid cells) $\AV$ (left) and $\CD$ (right).
We decompose these distances between the empirical distributions (HadEX2) and the model forecasts into our shift and dispersion components, whose magnitude is shown by the colored bars.
To facilitate visual distinction, opposing shift and dispersion components are drawn in different directions. 
The models that are discussed in the text are highlighted in green.
}
\label{fig:climate}
\end{figure*}

\cite{Thorarinsdottir2020} compare the empirical distributions of these temperature extremes over the years 1979--2005 according to various data sources to corresponding forecast distributions from a variety of models. We replicate a comparison between the primary data source used in the paper (HadEX2, \citealt{Donat2013}) and 29 different forecasting models from the CMIP5 project \citep{Sillmann2013}.
For 136 grid cells of size $3.75^{\circ}$ (longitude) $\times$ $2.5^{\circ}$ (latitude) covering European land masses, the cell-wise Cram\'er distances between the empirical and model-based distributions are computed and subsequently averaged (Figure 4 in \citealt{Thorarinsdottir2020}). This results in a ranking of the different models in terms of their capacity to predict the distribution of temperature extremes.

\begin{figure}[p]
\centering
\subcaptionbox{AVM decomposition (one grid cell)}{
\includegraphics[width = \defaultfigwidth\textwidth,page = 1]{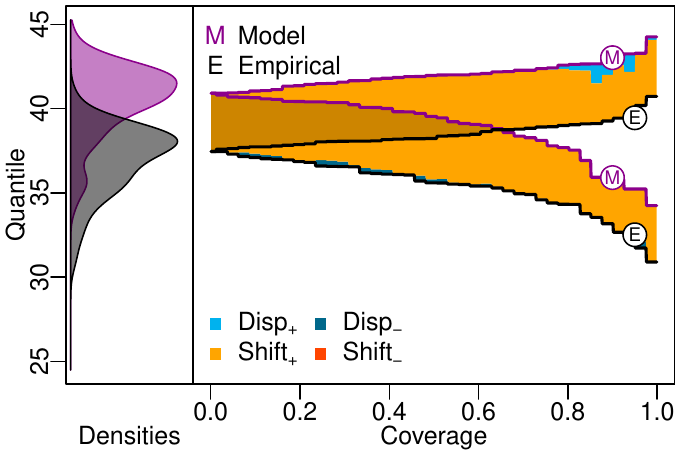}}
\subcaptionbox{AVM decomposition (another grid cell)}{\includegraphics[width = \defaultfigwidth\textwidth,page = 2]{figures/Fig10.pdf}}
\caption{Two examples of comparisons on specific grid cells between the empirical TXx distributions (HadEX2, black) and the predictive distributions from the \texttt{HadGEM2-ES} forecasting model (purple). In each panel, we show kernel density estimates of the two distributions, as well as their AVM decompositions (based on empirical CDFs).
In each of the grid cells, exactly one shift component is non-zero: The model overpredicts in the left, while underpredicting in the right plot. See Figure~\ref{fig:Illustration_FoldedQuantiles} for details on the graphical display of the AVM decomposition.}
\label{fig:climate2}
\bigskip\bigskip

	\centering
	\subcaptionbox{AVM decomposition}{
	\includegraphics[width = \defaultfigwidth\textwidth,page = 1]{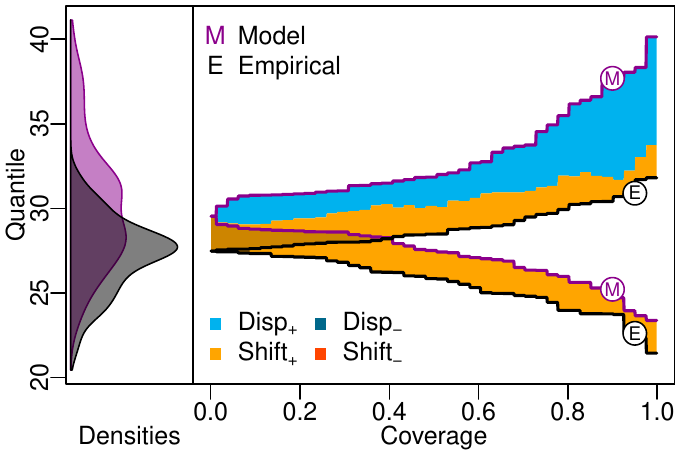}}
	\subcaptionbox{CD decomposition slice}{
		\includegraphics[width = \defaultfigwidth\textwidth,page = 2]{figures/Fig11.pdf}}
	\caption{Comparison of the decompositions of the AVM and CD for a selected grid cell (HadEX2 data in black against the \texttt{HadGEM2-ES} model in purple). For the selected grid cell, the two distributions visibly differ in dispersion, which is picked up more strongly by the AVM (second panel) than the CD (third panel). Note that the CD decomposition is only shown for the level $\beta = 0.5$, but as can be seen from the fourth panel, the overall decomposition gives similarly little weight to $\disp_+$. 
		See Figures \ref{fig:Illustration_FoldedQuantiles} and \ref{fig:illustration_cd} for details on the graphical displays of the decompositions.}
	\label{fig:climate3}
\end{figure}

We rerun these computations\footnote{
The results for the average Cram\'er distance for the models MIROC5 and MIROC-ESM\_CHEM reported here differ from the ones in \citet[][Figure 4, left]{Thorarinsdottir2020} for unknown reasons.}, 
apply both the Cram\'er distance and the area validation metric and compute the respective decompositions. The results are displayed in Figure \ref{fig:climate}. Several relevant patterns emerge, which are not discernible based on the average divergences from \citet{Thorarinsdottir2020}.
Firstly, there is a clear dominance of one of the shift components for most models. 
This indicates that the differences between empirical and predicted distributions tend to be systematic across grid cells.
An interesting exception is the model \texttt{HadGEM2-ES}, where both shift components are of similar size. As shown in Figure \ref{fig:climate2}, this mixed pattern results from different shifts in different grid cells rather than simultaneously positive shift components within the same grid cell. 
Secondly, concerning the dispersion components, it can be seen that the model forecasts are consistently more dispersed than the empirical distributions.

The application moreover illustrates that the AVM tends to give more weight to dispersion components than the CD, as discussed in Section \ref{subsec:ComparisonDecompositions}.
While for some models, the dispersion component represents a substantial part of the overall AVM, the dispersion components are largely negligible for the CD. In Figure \ref{fig:climate3} we illustrate this behavior for a selected grid cell. The different relative importance of shift and dispersion components also explains some of the differences in model rankings across the two divergences. Most notably, the \texttt{HadCM3} model receives a large dispersion component under the AVM, leading to a considerably worse ranking than under the CD. Similarly, the \texttt{HadGEM2CC} and \texttt{CCSM4} models (ranks 4 and 6, respectively, under AVM) change places under the CD. The reason for these changes is that differences in dispersion, which play a relevant role in the AVM, become negligible in the CD.

\subsection{Elicitation of inflation predictions}
\label{subsec:inflation}

\begin{figure*}[tb]
\centering
\includegraphics[width=\linewidth]{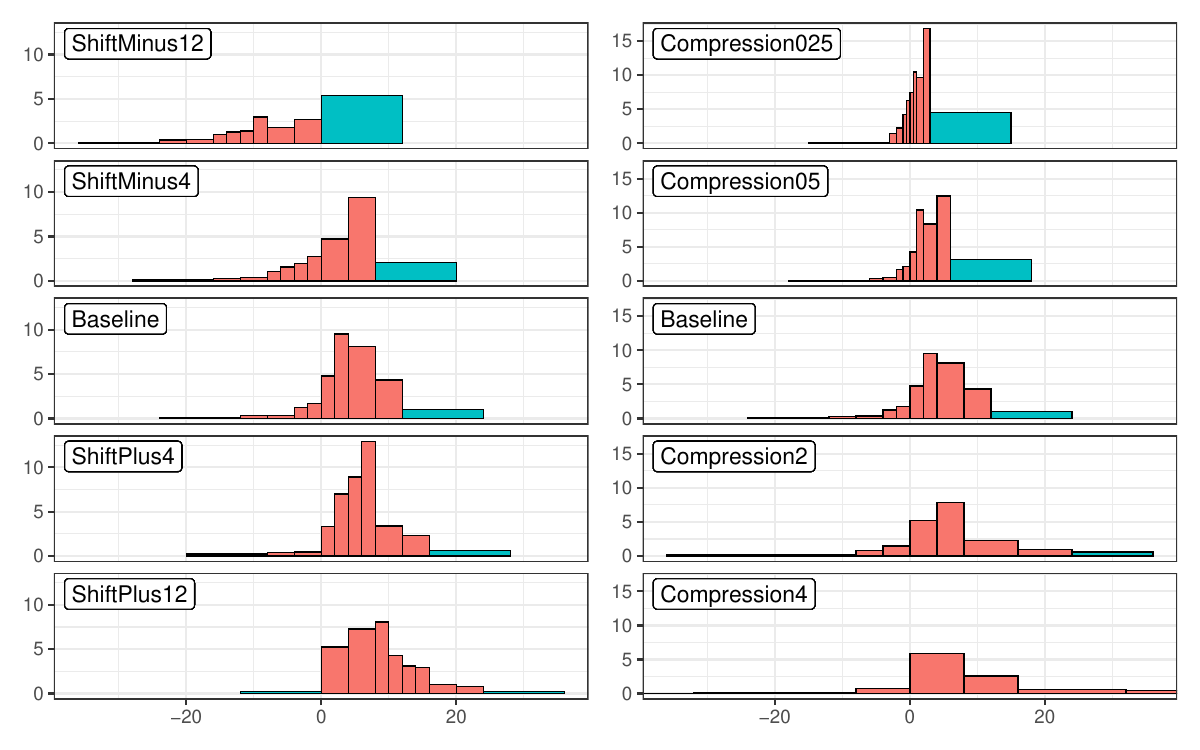}
\caption{Density (histogram) forecasts averaged over all survey respondents for the shift and compression treatments of \citet{BDE23}. 
All treatments include an open interval on the left and right in blue, which we truncate such that it has a width of 12 percentage points purely for illustrative purposes.
}
\label{fig:Hist_FCs}
\end{figure*}

Our second application considers economic surveys on households' expectations about future inflation, which have become an important tool for central banks over the previous decades \citep{Coibion2022}.
Inspired, among others, by the call of \citet{Manski2004} for probabilistic forecasts, recent surveys such as the Survey of Consumer Expectations (SCE) of the Federal Reserve Bank of New York are often in a probabilistic form, where the respondents assign probabilities to pre-specified bins, resulting in histogram-like forecasts.
The issued probabilities together with the right endpoints of the bins elicit precise information on the forecasted CDFs.
However, this survey format has come under question as the respondents' elicited distributions are found to be sensitive to the bin specification \citep{Schwarz1985, BDE23}.
This effect is especially disconcerting in times of varying inflation rates that necessitate adjustments in the bin specifications.

\citet{BDE23} focus on the undesirable effect that a given binning specification has on the responses in a preregistered experimental study on the platform Prolific (\url{www.prolific.co}), where randomly selected respondents obtain varying bin specifications as a treatment when issuing their probabilistic expectations.
The baseline specification follows the traditional binning used in the SCE.
The authors consider (among others) ``shift'' and ``compression'' treatments, where the bins are shifted to the right or left, and compressed or decompressed, as illustrated in Figure~\ref{fig:Hist_FCs}.

\begin{table*}[tb]
\centering
\begin{tabular}{l c rrrrr c rrrrr}
\toprule
&& \multicolumn{5}{c}{Cram\'er Distance (CD)} &&  \multicolumn{5}{c}{Area Validation Metric (AVM)} \\
\cmidrule{3-7}  	\cmidrule{9-13}
Treatment ($F$)  && $\CD$ & $\shift_+$ & $\shift_-$ & $\disp_+$ & $\disp_-$ &&  $\AV$ & $\shift_+$ & $\shift_-$ & $\disp_+$ & $\disp_-$ \\
\midrule
ShiftMinus12 && 0.820 & 0.000 & 0.534 & 0.286 & 0.000 && 4.556 & 0.000 & 1.823 & 2.733 & 0.000 \\ 
ShiftMinus4 && 0.081 & 0.000 & 0.070 & 0.011 & 0.000 && 1.222 & 0.000 & 0.743 & 0.443 & 0.037 \\ 
ShiftPlus4 && 0.068 & 0.060 & 0.000 & 0.001 & 0.007 && 0.996 & 0.649 & 0.000 & 0.118 & 0.228 \\ 
ShiftPlus12 && 0.206 & 0.198 & 0.000 & 0.008 & 0.001 && 2.025 & 1.626 & 0.000 & 0.303 & 0.096 \\
\midrule
Compression025 && 0.123 & 0.000 & 0.108 & 0.001 & 0.015 && 1.344 & 0.000 & 0.842 & 0.051 & 0.451 \\ 
Compression05 && 0.050 & 0.000 & 0.017 & 0.000 & 0.033 && 0.859 & 0.000 & 0.079 & 0.000 & 0.780 \\ 
Compression2 && 0.149 & 0.007 & 0.000 & 0.142 & 0.000 && 2.442 & 0.097 & 0.000 & 2.345 & 0.000 \\ 
Compression4 && 0.488 & 0.120 & 0.000 & 0.368 & 0.000 && 5.060 & 0.358 & 0.000 & 4.702 & 0.000 \\ 
\bottomrule
\end{tabular}
\caption{Approximate Cram\'er distances and area validation metrics together with their (approximated) decompositions
	of the average histogram forecast under the various treatments ($F$) described in the first column in comparison to the average forecast under the Baseline treatment ($G$).
}
\label{Tab:Inflation}
\end{table*}

While it would be desirable that responses are unaffected by the given binning, \citet{BDE23} find significant shift and compression effects in the average responses that are closely related to the implemented changes in the binning.
However, their methodology is limited to analyzing the means and standard deviations that are obtained from fitting a parametric beta distribution to the individual responses \citep{Engelberg2009}, and to the probability of the binarized event of deflation.

However, given that the survey is probabilistic in nature, so should be its evaluation. 
Hence, we refine the results of \citet{BDE23} by considering our decompositions of the $\AV$ and $\CD$ of the aggregated response distributions.
Our decompositions in shift and dispersion components are particularly suitable for the shift and compression treatments of \citet{BDE23}.
Table \ref{Tab:Inflation} shows the $\AV$ and $\CD$ together with the four components of the decompositions comparing the different treatments ($F$) to the baseline distribution ($G$).
As the distributions shown in Figure~\ref{fig:Hist_FCs} identify the respective CDFs only at the values separating the histogram bins, we use the approximation of the decomposition terms described in Supplement~\ref{subsec:inflation_approx} 
and shown in Figure~\ref{Fig:Inflation} for two exemplary treatments.
The solid points in the quantile spread plots show the known quantiles and a linear interpolation is used in between.
In the tails, a conservative extrapolation is used that we describe in detail in Supplement~\ref{subsec:inflation_approx}.

\begin{figure*}[tb]
\centering
\subcaptionbox{AVM decomposition for ShiftPlus12}{
\includegraphics[width=\defaultfigwidth\textwidth,page = 4]{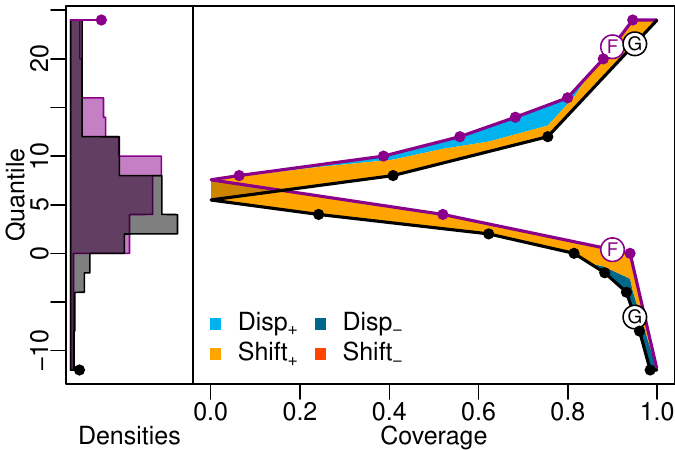}}
\subcaptionbox{AVM decomposition for Compression05}{
\includegraphics[width=\defaultfigwidth\textwidth,page = 6]{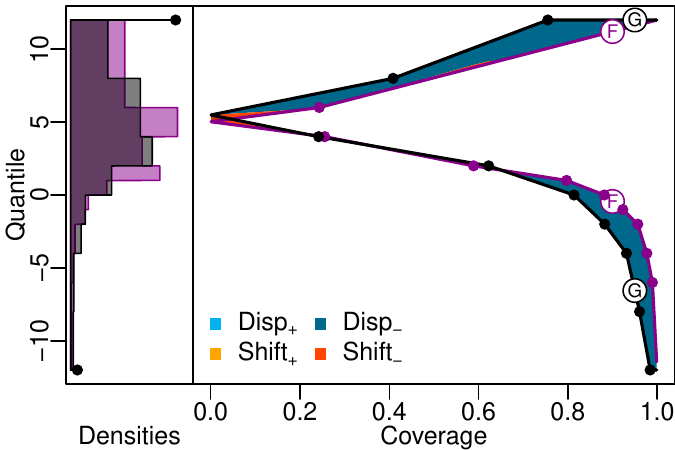}}
\caption{Illustrations of the approximate AVM decompositions for the ShiftPlus12 (left) and Compression05 (right) treatments ($F$) against the baseline treatment ($G$).
Point masses are represented by solid dots in the density plots. Known quantiles are marked with solid dots in the quantile spread plots. See Figure \ref{fig:Illustration_FoldedQuantiles} for a detailed description of the plots.}
\label{Fig:Inflation}
\end{figure*}

For all four shift treatments, Table~\ref{Tab:Inflation} shows a large shift component in the anticipated direction for both, the AVM and the CD.
As already noted by \citet{BDE23}, the ``artificial truncation'' of the bins (especially in ShiftMinus12) that can be observed in Figure~\ref{fig:Hist_FCs} together with the recent surge in inflation rates results in an increase in dispersion.
As expected from the theoretical results of Section~\ref{subsec:ComparisonDecompositions}, the relative magnitude (among the entire divergence) of the shift components is smaller for the AVM than for the CD.

In the compression treatment, we find equivalent effects in the dispersion components in the expected direction.
The compression treatments with factors 0.25 and 0.5 entail---perhaps surprisingly---large shifts, which can be traced back to the artificial truncation combined with the high inflation rates at the time of the survey in December 2021.

The results in this section supplement the parametric analysis of \citet{BDE23} by adding a nonparametric verification that takes into account the full distributional differences between the baseline and the treated probability forecasts.

\section{Discussion}
\label{sec:discussion}

In the present paper, we introduce decompositions into dispersion and shift components for various statistical distances. 
The decompositions meaningfully attribute the overall distance to differences in location and variability. They behave as expected in clear-cut special cases such as pairs of symmetric distributions and distributions from the same location-scale family, while being applicable to arbitrary pairs of distributions.
Furthermore, the decompositions shed light on the sensitivities of the studied divergences towards differences in location and dispersion. 
The decompositions are compatible with the usual stochastic and dispersive order relations, and we establish correspondences to suitably constructed weakened order relationships.
Finally, we demonstrate the practical use of the decompositions in two case studies.

Both our theoretical results and case studies indicate that the Cram\'er distance is rather insensitive to differences in dispersion, and puts a strong focus on differences in location. The area validation metric (and higher-order Wasserstein distances) on the other hand emphasize differences in dispersion.
Besides these natural differences, we note that the Cram\'er distance has an important advantage when used for forecast evaluation as in the application to climate predictions. Unlike the area validation metric (and higher-order Wasserstein distances), it is a \textit{proper divergence metric} \citep{TGG13}, which rewards truthful predictions.

As a consequence of our interval-based approach, the proposed decompositions attribute the distance between two distributions $F$ and $G$ entirely to differences in location and dispersion. 
Of course, distributions are often characterized by additional (higher-order) properties such as skewness or kurtosis attributed to their shape. In our decompositions, differences in higher-order properties sometimes lead to both shift or dispersion components being nonzero simultaneously (e.g., in Example \ref{exmpl:BothComponentsPositive}), but there are no clear-cut connections. 
The \emph{nonparametric} nature of our decompositions that aggregate fundamental comparisons of central intervals does not attest to differences in shape. In contrast, decompositions that compare summary statistics as the one in \eqref{eq:paramDecomp}, which may do so, provide a rather superficial comparison and 
are not available for most of the studied distances.
Accommodating additional shape components likely requires considerably revised techniques and is reserved for future work.

The decompositions are thus far descriptive tools that appear to be promising both from a theoretical and applied perspective. A next step will be to develop inference techniques that assess the statistical significance of the components in standard one- and two-sample settings as well as more involved settings such as aggregated inflation expectations.
While resampling techniques such as permutation tests or the bootstrap are well suited to assess the statistical significance of the overall distance, assessing the significance of individual components is challenging. 
As the null hypotheses of no difference in shift or location are complex composite hypotheses, simulating them does not seem feasible. A careful look at the sum of the dispersion components of the area validation metric reveals that it can be rewritten as a sum of conditional expectations, which may well facilitate inference for this particular component. We view such developments as a natural avenue for future research.

It remains an open question how the discussed distances can be decomposed for multivariate distributions.
For the specific case of the 2-Wasserstein distance, the location component of the moment-based decomposition \eqref{eq:paramDecomp} discussed in Remark \ref{rem:paramDecomp} can be generalized to higher dimensions as follows.
If $P$ and $Q$ are probability distributions on $\R^d$ with first moments $\mu_P,\mu_Q$, respectively, and finite second moments, 
the squared (multivariate) 2-Wasserstein distance \citep[Eq.\ 1.]{Panaretos2019} can be written as
\begin{equation*}
	\inf_{X\sim P, Y\sim Q} \mathbb{E} \Vert X - Y\Vert^2 = \Vert \mu_P - \mu_Q\Vert^2 + \inf_{X\sim P, Y\sim Q} \mathbb{E} \Vert (X - \mu_P) - (Y - \mu_Q)\Vert^2,
\end{equation*}
which may provide a promising starting point for further refinements.

On the other hand, direct generalizations of our nonparametric approach are hindered by the fact that representations of the distances in terms of quantiles are limited to the univariate setting. 
Our decomposition may however be leveraged in the multivariate setting via a slicing approach; see e.g., \citet{KSN+19}.
Here, a sliced distance is computed as the average distance of linear projections (in all directions) of the distributions to the real line.
Accordingly, a sliced decomposition can be obtained by averaging the decomposition components of all linear projections. 
However, the positive and negative shift components will become equal here, as for each projection, the projection in the opposite direction will lead to an equal contribution to the opposing shift component.
Due to the averaging step over all slicing directions, the final shift components will always coincide.
One possible improvement would be to fix a hyperplane relative to which the directed components could be defined.
Further refinements that incorporate multiple directions in higher dimensions may be feasible. We leave such considerations to future research.

\section*{Acknowledgements}

We thank Thomas Eife and Thordis Thorarinsdottir for sharing the data sets for the two applications and Christoph Becker, Andreas Eberl,  Thomas Eife, Tilmann Gneiting and Melanie Schienle, as well as an anonymous referee, for their insightful comments.
T.~Dimitriadis and J.~Resin (project 502572912) as well as J.~Bracher (project 512483310) gratefully acknowledge support of the Deutsche Forschungsgemeinschaft (DFG, German Research Foundation). 
D.~Wolffram was supported by the Helmholtz Association under the joint research school HIDSS4Health -- Helmholtz Information and Data Science School for Health.
All authors acknowledge the Heidelberg Institute for Theoretical Studies (HITS) as secondary affiliation.

We further acknowledge the World Climate research Programme's Working Group on Coupled Modelling, which is responsible for CMIP, and we thank the climate modeling groups (listed in Table \ref{tab:ListCMIP5models} 
in the supplement) for producing and making available their model output. 
For CMIP the U.S.~Department of Energy's Program for Climate Model Diagnosis and Intercomparison provides coordinating support and led development of software infrastructure in partnership with the Global Organization for Earth System Science Portals.

\bibliographystyle{apalike}
\bibliography{manuscript}

\newpage

\section*{Supplement}
\beginsupplement

\noindent 
The supplement contains proofs and derivations together with additional illustrations, counterexamples, and details on approximations.
Section~\ref{sec:illustration_cd_grid} graphically illustrates the CD decomposition for a range of $\beta$ values.
We provide
closed-form expressions for the $\WD_p$ decomposition in Section~\ref{sec:NormalDistribution}.
Section~\ref{sec:Counterexamples} gives additional counterexamples.
Section~\ref{sec:Approximation} provides additional details on the computation and approximation of the decompositions. In particular, we provide simple summation formulas for discrete distributions used in Section \ref{subsec:inflation} as well as two approximation strategies for settings with limited knowledge of the distributions at hand along with worst case bounds that can be leveraged to derive approximation error rates.
Finally, Sections~\ref{sec:ProofsDecomposition}--\ref{sec:ConnectionOrders} contain all proofs of the results presented in Sections \ref{sec:Decompositions}--\ref{sec:Orders}, respectively.

\section{Graphical illustration of the CD decomposition}
\label{sec:illustration_cd_grid}

Figure \ref{fig:illustration_cd_grid} contains graphical illustrations as in Figure~\ref{fig:illustration_cd}, panel (b), at various levels $\beta \in \{0.1,0.2,\dots,0.9\}$.

\begin{figure}[!ht]
	\begin{center}
		\includegraphics[width = \textwidth]{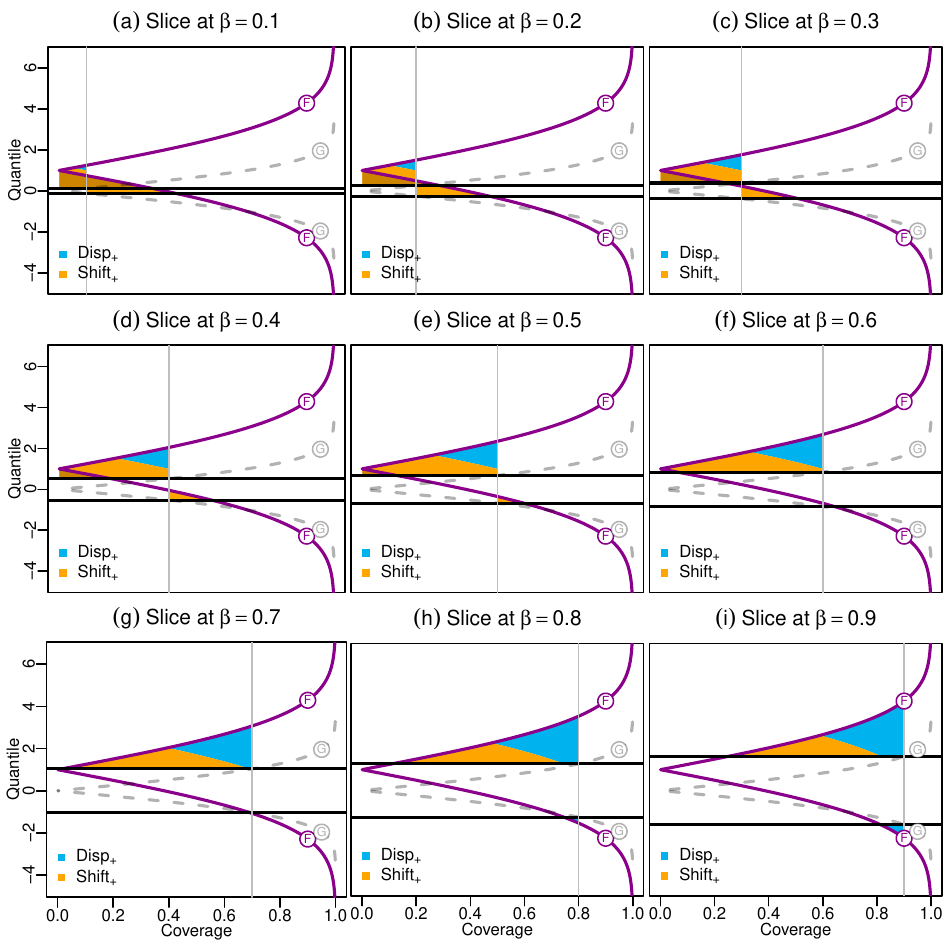}
		\caption{Graphical illustrations as in Figure~\ref{fig:illustration_cd} at various $\beta \in \{0.1,0.2,\dots,0.9\}$ levels.}
		\label{fig:illustration_cd_grid}
	\end{center}
\end{figure}

\section{Closed-form expressions for normal distributions}
\label{sec:NormalDistribution}

We supplement the closed-form expressions given in the main text for the area validation metric and the Cram\'er distance between two normal distributions with formulas for the Wasserstein distance of order $p$.

\subsection{The $\WD_p$ decomposition for normal distributions}
\label{sec:WDpNormals}

Here, we present closed-form expressions for the decomposition terms of the $p$-th power of the $p$-Wasserstein distance with $p \in \mathbb{N}$ for two normal distributions, $F = \norm(\mu_F,\sigma_F^2)$ and $G = \norm(\mu_G,\sigma_G^2)$.
As in Section \ref{Sec:Properties}, we use of the shorthand notations $\mudiff = |\mu_F - \mu_G|$ and $\sigmadiff = |\sigma_F - \sigma_G|$.

The following formulas are based on expressions $m_p(\mu,\sigma,a)$ for the $p$-th moments of a normal distribution $\norm(\mu, \sigma^2)$ that is \emph{truncated} from below at the value $a \in \mathbb{R}$.
\citetsupp{Orjebin2014} provides the recursive formula 
\[
m_p(\mu,\sigma,a) = (p - 1)\sigma^2m_{p-2}(\mu,\sigma,a) + \mu m_{p-1}(\mu,\sigma,a) + \sigma\frac{a^{p-1}\phi(\tfrac{a - \mu}{\sigma})}{1 - \Phi(\tfrac{a - \mu}{\sigma})}
\]
for $\mu\in\R, \sigma > 0$ and $p \in \mathbb{N}$ with $m_0(\mu,\sigma,a) = 1$ and $m_{-1}(\mu,\sigma,a) = 0$.
Then, the Wasserstein distance between two normal distributions is given by
\begin{align*}
	\WDp(F,G) = 
	\begin{cases}
		\mudiff^p  &\text{if} \qquad \sigma_F = \sigma_G,  \\
		\Phi(\mudiff/\sigmadiff)m_p(\mudiff,\sigmadiff,0) + \Phi(-\mudiff/\sigmadiff)m_p(-\mudiff,\sigmadiff,0) &\text{if} \qquad \sigma_F \not= \sigma_G. 
	\end{cases}
\end{align*}

If the variances are equal, the $p$-Wasserstein distance is assigned to a single shift component,
\begin{align*}
	\sigma_F = \sigma_G 
	\quad \Longrightarrow \quad 
	\WDp(F,G) = \mudiff^p = 
	\begin{cases}
		\shift_+^{\WDp}(F,G),& \text{if } \mu_F > \mu_G, \\
		\shift_-^{\WDp}(F,G),& \text{if } \mu_F < \mu_G,
	\end{cases}
\end{align*}
and the dispersion components are zero (as well as the opposing shift component).

In contrast, for differing variances, we get
\begin{align*}
	\sigma_F > \sigma_G 
	\quad \Longrightarrow \quad 
	\begin{cases}
		\disp_+^{\WDp}(F,G) &= m_p(\mudiff,\sigmadiff,\mudiff) + (1-\Phi(\mudiff/\sigmadiff))m_p(-\mudiff,\sigmadiff,0) - \Phi(\mudiff/\sigmadiff)m_p(\mudiff,\sigmadiff,0), \\
		\disp_-^{\WDp}(F,G) &= 0,
	\end{cases}
\end{align*}
and the above dispersion terms are swapped if $\sigma_F < \sigma_G$.

Finally, given differing variances and an ordering of the means, we get
\begin{align*}
	\mu_F \ge \mu_G 
	\quad \Longrightarrow \quad 
	\begin{cases}
		\shift_+^{\WDp}(F,G) = 2\Phi(\mudiff/\sigmadiff)m_p(\mudiff,\sigmadiff,0) - m_p(\mudiff,\sigmadiff,\mudiff) \\
		\shift_-^{\WDp}(F,G) = 0
	\end{cases}
\end{align*}   
and the above shift terms are swapped if $\mu_F < \mu_G$.

We omit the tedious derivations of the closed-form expressions for normal distributions given here and in the main text.

\section{Counterexamples}
\label{sec:Counterexamples}

This section contains various counterexamples that illustrate that certain restrictions in our theoretical results from Sections~\ref{Sec:Properties}--\ref{sec:Orders} are required.

\subsection{Counterexamples of decomposition comparisons}
\label{sec:ExamplesRelativeComparison}

\begin{figure}[tb]
	\begin{center}
		\subcaptionbox{Distributions in Example~\ref{Ex:pWD_loc-scale}}{
			\includegraphics[width = \defaultfigwidth\textwidth]{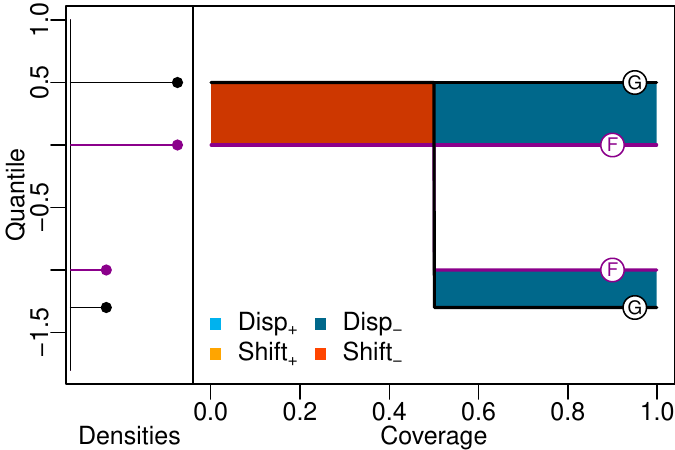}}
		\subcaptionbox{Distributions in Example~\ref{Ex:pWD_unimodal}}{\includegraphics[width = \defaultfigwidth\textwidth]{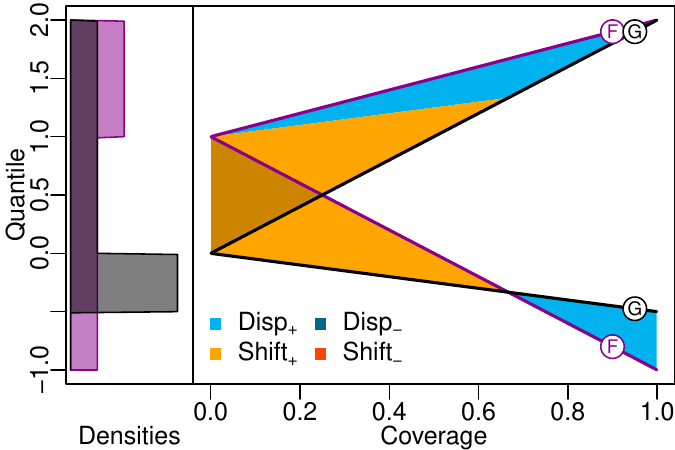}}\bigskip\\
		\subcaptionbox{Distributions in Example~\ref{Ex:CD-WD_sym}}{\includegraphics[width = \defaultfigwidth\textwidth]{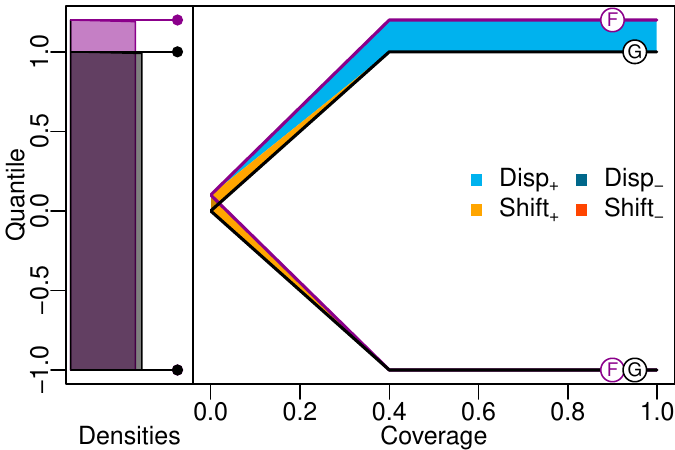}}
	\end{center}
	\caption{Illustrations of the AVM decompositions for distributions $F$ and $G$ given in Examples~\ref{Ex:pWD_loc-scale}--\ref{Ex:CD-WD_sym}. See Figure \ref{fig:Illustration_FoldedQuantiles} for a detailed description of the plots.}
	\label{fig:Counterexamples}
\end{figure}

We start to illustrate the necessity of the \emph{symmetry} condition in Theorem~\ref{Thm:DispInequality_pWD}.
The following counterexample shows that the corresponding inequality \eqref{Eq:DispInequality_pWD} is not guaranteed to hold for asymmetric distributions, even if we focus attention on location-scale families as formally defined in Proposition~\ref{Prop:CompLocScale}.

\begin{example}
	\label{Ex:pWD_loc-scale}
	Consider the distributions $F$ and $G$ that are illustrated in panel (a) of Figure~\ref{fig:Counterexamples}.
	The distribution $F$ places $\frac14$ probability mass at $-1$ and $\frac34$ probability mass at 0 and $G$ is given by the location-scale transformation $G(x) = F(\frac{x - 0.5}{1.8})$. 
	The first three $p$-th power $p$-Wasserstein distances between $F$ and $G$ and their decompositions are
	\begin{align*}
		\WD_1(F,G) &= \shift_+^{\WD_1}(F, G) + \shift_-^{\WD_1}(F, G) + \disp_+^{\WD_1}(F, G) + \disp_-^{\WD_1}(F, G) = 0 + 0.25 + 0 + 0.20, \\
		\WD_2(F,G) &= 0 + 0.125 + 0 + 0.085, \qquad \WD_3(F,G) = 0 + 0.0625 + 0 + 0.038.
	\end{align*}
	Hence, the dispersion component accounts for about 44.4\% % 0.4444444
	of $\WD_1(F,G)$, which reduces to about 40.5\% % 0.4047619
	of $\WD_2(F,G)$ and about 37.8\% % 0.3781095
	of $\WD_3(F,G)$.
	
	A corresponding example with continuous distributions can be obtained by slightly tilting the horizontal and vertical (jump) segments of the quantile functions in Figure~\ref{fig:Counterexamples} (a).
	However, the given example is easy to grasp as the components are simply given by the areas of the three rectangles, where it is important to note that the red rectangle is counted twice for the shift component (compare to Figure~\ref{fig:Illustration_FoldedQuantiles}).
	As the $p$-th power $p$-Wasserstein distance takes the $p$-th power of the height of the rectangles, the lower (smaller) rectangle contributing to the dispersion component is down-weighted in comparison to the other rectangles when increasing the power $p$, thereby reducing the relative weight of the dispersion component. Evidently, the relative weight of the dispersion component converges to $\frac13$ from above as $p \rightarrow \infty$.
\end{example}

The following example illustrates that Theorem~\ref{Thm:DispInequality_pWD} and its inequality \eqref{Eq:DispInequality_pWD} do not hold for arbitrary (weakly) unimodal distributions, formally defined in Conjecture~\ref{conj:AVCD_ineq_unimodal}.

\begin{example}
	\label{Ex:pWD_unimodal}
	Consider the distributions $F = 0.5 \times \unif[-1,1] + 0.5 \times \unif[1,2]$ and $G = 0.5 \times \unif[-0.5,0] + 0.5 \times \unif[0,2]$ illustrated in panel (b) of Figure~\ref{fig:Counterexamples}.
	The decompositions of the first three $p$-th power $p$-Wasserstein distances between $F$ and $G$ are
	\begin{align*}
		\WD_1(F,G) &= \shift_+^{\WD_1}(F, G) + \shift_-^{\WD_1}(F, G) + \disp_+^{\WD_1}(F, G) + \disp_-^{\WD_1}(F, G) = 0.3333 + 0 + 0.125 + 0, \\
		\WD_2(F,G) &= 0.2222 + 0 + 0.0694 + 0, \qquad \WD_3(F,G) = 0.1667 + 0 + 0.0469 + 0.
	\end{align*}
	Hence, the dispersion component accounts for about 27.3\% % 0.2727273
	of $\WD_1(F,G)$, which reduces to about 23.8\% % 0.2380959
	of $\WD_2(F,G)$ and about 22.0\% % 0.2195122
	of $\WD_3(F,G)$.
	A counterexample with strongly unimodal distributions (that are strictly in- and decreasing left and right of the unique mode) can again be constructed by slightly tilting the respective horizontal lines in the density functions shown in Figure~\ref{fig:Counterexamples} (b).
\end{example}

We turn now to generalizations of inequality \eqref{Eq:DispInequality_CD-WD_norm} between relative (normalized) dispersion components of the Cram\'er distance and the area validation metric.
Contrary to inequality \eqref{Eq:DispInequality_pWD}, inequality \eqref{Eq:DispInequality_CD-WD_norm} does not hold for arbitrary symmetric distributions, as illustrated by the following example.

\begin{example}
	\label{Ex:CD-WD_sym}
	Consider the distributions $F = \frac3{10}\times\dirac{-1} + \frac15\times \unif[-1,0.1] + \frac15\times \unif[0.1,1.2] + \frac3{10}\times\dirac{1.2}$ and $G = \frac3{10}\times\dirac{-1} + \frac15\times \unif[-1,0] + \frac15\times \unif[0,1] + \frac3{10}\times\dirac{1}$ illustrated in panel (c) of Figure~\ref{fig:Counterexamples}, where $\delta_y$ denotes the Dirac measure in $y\in\R$. 
	The decompositions of the Cram\'er distance and the area validation metric between $F$ and $G$ are given by 
	\[
	\CD(F,G) = 0.002 + 0 + 0.019 + 0, \qquad \AV(F,G) = 0.02 + 0 + 0.08 + 0.
	\]
	Here, the dispersion accounts for only about 9\% of the Cram\'er distance, whereas it accounts for 20\% of the area validation metric.
	
\end{example}

\subsection{Counterexamples for the order relations}
\label{sec:ExamplesOrderRelations}

The following example illustrates the necessity of the common support assumption in Proposition~\ref{Thm:WeakStochPreorder} by showing that the weak stochastic order is not guaranteed to be a transitive relation in general.

\begin{figure}[tb]
	\centering
	\subcaptionbox{AVM decomposition}{
		\includegraphics[page = 1,width = 0.32\textwidth]{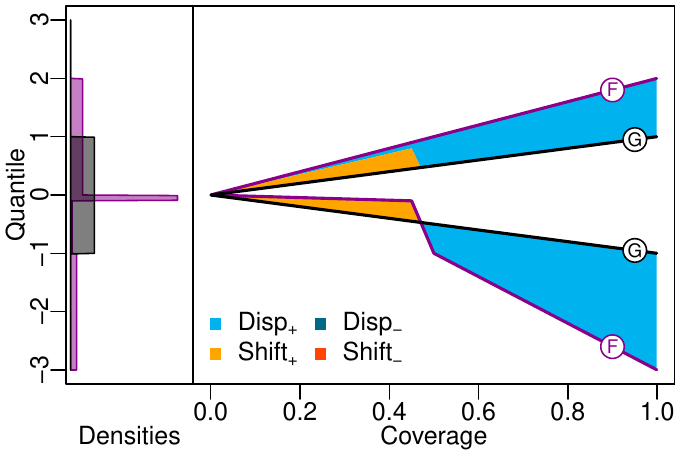}}
	\subcaptionbox{AVM decomposition}{
		\includegraphics[page = 2,width = 0.32\textwidth]{figures/FigS3.pdf}}
	\subcaptionbox{AVM decomposition}{
		\includegraphics[page = 3,width = 0.32\textwidth]{figures/FigS3.pdf}}
	\caption{Illustrations of AVM decompositions for the distributions given in Example \ref{Ex:IntransitiveCD}. Here, the same nonzero components arise in the $\CD$ decomposition. See Figure \ref{fig:Illustration_FoldedQuantiles} for a detailed description of the plots.}
	\label{Fig:IntransitiveCD}
\end{figure}

\begin{example}
	\label{Ex:IntransitiveCD}
	Consider the three distributions 
	\begin{align*}
		F &= \tfrac14\times\unif[-3,-1] + \tfrac1{40}\times\unif[-1,-0.1] + \tfrac9{40}\times\unif[-0.1,0] + \tfrac12\times\unif[0,3], \\
		G &= \unif[-2,2], \\
		H &= \tfrac12\times\unif[-3,0] + \tfrac9{40}\times\unif[0,0.1] + \tfrac1{40}\times\unif[0.1,1] + \tfrac12\times\unif[1,2],
	\end{align*}
	illustrated in Figure \ref{Fig:IntransitiveCD}, which do not have a common support, but continuous quantile functions.
	Their respective Cram\'er distance decompositions are approximately given by
	\begin{align*}
		&\CD(F,G) \approx 0.1336 + 0 + 0.8664 + 0, \\
		&\CD(G,H) \approx 0.1345 + 0 + 0 + 0.8655, \\
		&\CD(F,H) \approx 0.5615 + 0.4385 + 0 + 0.
	\end{align*}
	Thus, invoking part (a) of Theorem~\ref{Thm:StochOrderCDJoint} for all three comparisons, we have 
	$F \geq_\text{wS} G \geq_\text{wS} H$, but $F \not\geq_\text{wS} H$. 
	Hence, the weak stochastic order is not a transitive relation on arbitrary sets of distributions. 
\end{example}

The following example shows that in contrast to the \emph{weak} stochastic order (see Theorem~\ref{Thm:WeakStochPreorder}), a common support is not sufficient to establish transitivity for the \emph{relaxed} stochastic order.

\begin{figure}[tb]
	\centering
	\subcaptionbox{AVM decomposition}{
		\includegraphics[page = 1,width = 0.32\textwidth]{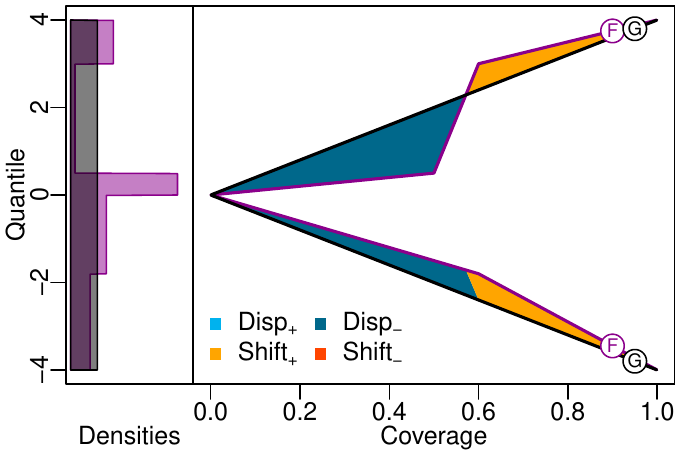}}
	\subcaptionbox{AVM decomposition}{
		\includegraphics[page = 2,width = 0.32\textwidth]{figures/FigS4.pdf}}
	\subcaptionbox{AVM decomposition}{
		\includegraphics[page = 3,width = 0.32\textwidth]{figures/FigS4.pdf}}
	\caption{Illustrations of AVM decompositions for the distributions given in Example \ref{Ex:Intransitive}. See Figure \ref{fig:Illustration_FoldedQuantiles} for a detailed description of the plots.}
	\label{Fig:Intransitive}
\end{figure}

\begin{example}
	\label{Ex:Intransitive}
	Consider the three distributions 
	\begin{align*}
		F &= \tfrac15\times\unif[-4,-1.8] + \tfrac3{10}\times\unif[-1.8,0] + \tfrac14\times\unif[0,0.5] + \tfrac1{20}\times\unif[0.5,3] + \tfrac15\times\unif[3,4], \\
		G &= \unif[-4,4], \\
		H &= \tfrac15\times\unif[-4,-3] + \tfrac1{20}\times\unif[-3,-0.5] + \tfrac14\times\unif[-0.5,0] + \tfrac3{10}\times\unif[0,1.8] + \tfrac15\times\unif[1.8,4],
	\end{align*}
	that are illustrated in Figure \ref{Fig:Intransitive}, which have common support and continuous quantile functions.
	The quantile spread plots show unique nonzero shift components of the AVM, when comparing $F$ and $G$ (left), and $G$ and $H$ (middle), which implies a relaxed stochastic ordering, $F \geq_\text{rS} G \geq_\text{rS} H$ by Theorem \ref{Thm:StochOrderAVMJoint}. However, a comparison of $F$ and $H$ (right) yields two nonzero shift components, hence $G \not\geq_\text{rS} H$ in relaxed stochastic order.
	Thus, the relaxed stochastic order is not a transitive relation on arbitrary sets of distributions with common support. 
	Notably, the CD decomposition produces two nonzero shift components, when comparing any of the three pairs, such that neither of the three comparisons yields a \emph{weak} stochastic ordering by Theorem~\ref{Thm:StochOrderCDJoint}.
\end{example}

\section{Computation and approximation}
\label{sec:Approximation}

When computing the distances and their decompositions numerically, two complications may arise. 
While in many cases standard numerical integration works fine when computing the required integrals, the nested double integral in the Cram\'er distance can result in numerical instabilities and long computation times, for example when dealing with discrete distributions.
Secondly, the distributions may only be partially specified via a limited number of known points on the distribution or quantile functions, as in the application of Section \ref{subsec:inflation}.

Here, we provide concrete formulas for discrete distributions in Section~\ref{Sec:Discrete} and outline two approximation strategies for scenarios with incomplete distributions:
In Section~\ref{subsec:inflation_approx}, we discuss an approximation based on linear interpolation of the quantile functions, whereas the approximation in Section~\ref{subsec:PiecewiseConstApprox} is based on a piecewise constant (nearest neighbor) interpolation.
We continue to derive upper and lower worst case bounds for the approximated distances and their decomposition components in Section~\ref{subsec:Bounds}.
Finally, Section~\ref{subsec:ApproxError} illustrates how these bounds can be leveraged for the construction of asymptotic approximation error rates.

As in our applications, we focus on the area validation metric and the Cram\'er distance in what follows, noting that analogous strategies and results for the higher-order Wasserstein distances are straight-forward generalizations of the case of the AVM.

\subsection{Discrete distributions with finite support}
\label{Sec:Discrete}

Here, we provide formulas for the AVM and the CD together with their decomposition terms for discrete distributions with finite support. In particular, we consider piecewise constant CDFs $F$ and $G$ that have a finite number of discontinuities (``jumps'') each, which may occur at different quantile levels.
Suppose the extended sequence of quantile levels $0 = \alpha_0 < \alpha_1 < \alpha_2 < \dots < \alpha_n < \alpha_{n+1} = 1$ contains all levels at which either quantile function exhibits a jump and, in addition, satisfies $\alpha_i = 1- \alpha_{n+1-i}$ (symmetry of levels about the median level) with $n$ being odd, i.e., $\alpha_{(n+1)/2} = \frac12$.
In detail, the set $\{\alpha_0, \dots,\alpha_{n+1}\}$ is constructed
as follows:
First, consider the sequence of all quantile levels at which either quantile function exhibits a jump. Second, for each jump level $\alpha$ add the level $1-\alpha$ to the sequence of levels if it is not yet contained in the sequence.
Third, add the median level $0.5$ (if not yet contained) as well as the boundary levels 0 and 1 to the sequence.
As such, all jump points of $F^{-1}$ and $G^{-1}$ are contained in the ordered sequence $\alpha_0, \dots,\alpha_{n+1}$, but neither quantile function necessarily jumps at all of these levels.

The distributions $F$ and $G$ are given by piecewise-constant quantile functions\footnote{Note that we do not restrict to a particular choice (such as the usual choice of a left-continuous quantile function) of the inherently set-valued quantile at the jump points $\alpha_i$, as it does not matter for the integrals in the distances and their decompositions.}
with
\begin{align*}
	F^{-1}(\alpha) = f_i
	\quad \text{and} \quad G^{-1}(\alpha) = g_i \qquad \text{ if } \alpha \in (\alpha_i,\alpha_{i+1})
\end{align*}
for $i = 0,1,\dots,n$ and some values
$f_0 \leq f_1 \leq f_2 \leq\dots \leq f_{n}$ and $g_0 \leq g_1 \leq g_2 \leq\dots \leq g_{n}$. 

With $\overline{a}_i = \alpha_{i+1} - \alpha_i$, the AVM reduces to
\[
\AV(F,G) = \sum_{i = 0}^n \overline{a}_i \vert f_i - g_i\vert,
\]
and the decomposition is given by 
\begin{align*}
	\disp_+^{\AV}(F,G) &= \sum_{i = 0}^{\lfloor \frac{n}{2}\rfloor} \overline{a}_i \left[(f_{n-i} - g_{n-i}) - (f_i - g_i)\right]_+, \\
	\shift_+^{\AV}(F,G) &= 2\sum_{i = 0}^{\lfloor \frac{n}{2}\rfloor} \overline{a}_i \left[\min\{f_{n-i} - g_{n-i}, f_i - g_i\}\right]_+,
\end{align*}
where the opposing components are obtained via the usual symmetry.

On the other hand, the CD reduces to
\[
\CD(F,G) = \sum_{i = j} \overline{a}_i\overline{a}_j \vert f_i - g_j \vert + 2\sum_{i < j} \overline{a}_i\overline{a}_j [f_i - g_j]_+ + 2\sum_{i > j} \overline{a}_i\overline{a}_j [g_j - f_i]_+
\]
with components given by
\begin{align*}
	\disp_+^{\CD}(F,G) &=  \sum_{j = 0}^{\lfloor \frac{n}{2}\rfloor} \Big(\overline{a}_j^2 \left[(f_{n-j} - g_{n-j}) - (f_j - g_j)\right]_+ + 2\overline{a}_j\sum_{i = j+1}^{\lfloor \frac{n}{2} \rfloor} \overline{a}_i \left[(f_{n-i} - g_{n-j}) - (f_i - g_j)\right]_+\Big), \\
	\shift_+^{\CD}(F,G) &=  \sum_{i=0}^{\lfloor \frac{n}{2}\rfloor} \sum_{j=0}^{\lfloor \frac{n}{2}\rfloor} 2\overline{a}_i\overline{a}_j \left(\left[\min\{f_{n-i} - g_{n-j}, f_i - g_j\}\right]_+ + \left[f_i - g_{n-j}\right]_+\right).
\end{align*}

\subsection{Piecewise linear approximation}
\label{subsec:inflation_approx}

Here, we outline an approximation based on a linear interpolation of the quantile functions that can be applied when the distributions $F$ and $G$ are only partially known
via a limited number of points on the distribution or quantile functions.
For instance in the application from Section \ref{subsec:inflation}, the assigned bin probabilities together with the bin boundaries (i.e., the values separating the bins) identify the distribution and quantile functions at these boundary values.
However, no further information on the quantile functions can be inferred from the responses.

In Figure~\ref{Fig:Inflation} of the main text, we illustrate the points where the distributions are known by the solid points in the quantile spread plots.
To compute the Cram\'er and Wasserstein distances and their associated decompositions in this case, the quantile functions need to be approximated. 
For this, we adopt a simple linear interpolation between the known quantiles of the distributions as illustrated in Figure~\ref{Fig:Inflation} with the straight lines connecting the points.
In terms of the underlying histogram forecasts, this corresponds to uniformly distributing the probability mass within each bin.

In the tails (i.e., the open bins on either side), the histogram survey methodology does not capture any information on the distributions apart from the total probability mass beyond the most extreme bin boundaries.
Hence, we adopt a conservative approach in order to avoid an overestimation of the distances and decomposition terms in the tails (where inherently little information is given):
In the upper tail, we use the larger of the two lower bounds of the open bins to limit the support of the distributions, thereby shrinking this bin to a point mass for one of the distributions (as shown in Figure~\ref{Fig:Inflation} in the density plots, which results in a horizontal segment in the quantile spread plots at large coverages).
Analogously, we use the smaller of the two upper bounds of the open bins capturing the lower tail probabilities to limit the support of the distributions from below, thereby shrinking the left-most bin to a point mass for one of the distributions.

\subsection{Piecewise constant approximation}
\label{subsec:PiecewiseConstApprox}
\sloppy

Here, we outline a different approximation method that relies on a piecewise constant (nearest neighbor) interpolation of the quantile functions. 
This approximation may be of use if both distributions are only known via the quantile values at a fixed (and typically identical) set of quantile levels. 
This is a common way of storing predictive distributions in a parsimonious format; see e.g., \citet{Cramer2022}.

As an alternative to the linear interpolation from the previous section, we propose to approximate the unknown quantile functions at each level via the closest known quantile value of the distributions here. The resulting approximate quantile functions give rise to discrete distributions, to which the formulas from Section \ref{Sec:Discrete} can be applied.
In some cases, applying this approach on a fine grid of levels may also help to avoid numerical instabilities or improve computational efficiency, especially for the CD.

Suppose the values of the quantile functions $F^{-1}$ and $G^{-1}$ are known at levels $0 \leq \beta_1 < \dots < \beta_K \leq 1$ and $0 \leq \gamma_1 < \ldots < \gamma_L \leq 1$, respectively. Let $\beta_0 = -\beta_1,\beta_{K+1} = 2-\beta_K, \gamma_0 = -\gamma_1, \gamma_{L+1} = 2-\gamma_L$.
We define two approximate discrete distributions $\widetilde{F}, \widetilde{G}$ with quantile functions given by
\begin{align*}
	\widetilde{F}^{-1}(\beta) &= F^{-1}(\beta_k) \quad \text{for} \quad \beta \in \left(\frac{\beta_k + \beta_{k-1}}{2}, \frac{\beta_k + \beta_{k+1}}{2}\right), \qquad \text{and} \\
	\widetilde{G}^{-1}(\gamma) &= G^{-1}(\gamma_\ell)  \quad \text{for} \quad  \gamma \in \left(\frac{\gamma_\ell + \gamma_{\ell-1}}{2}, \frac{\gamma_\ell + \gamma_{\ell+1}}{2}\right).
\end{align*}
Then, the distances and their decompositions can be approximated by applying the formulas provided in Section \ref{Sec:Discrete} to the approximate distributions $\widetilde{F}$ and $\widetilde{G}$. In particular, we propose the approximations
\[
\AV(F,G) \approx \AV(\widetilde{F},\widetilde{G}),\quad \shift_\pm^{\AV}(F,G) \approx  \shift_\pm^{\AV}(\widetilde{F},\widetilde{G}),\quad \disp_\pm^{\AV}(F,G) \approx  \disp_\pm^{\AV}(\widetilde{F},\widetilde{G})
\]
for the AVM and its decomposition, and analogously for the CD,
\[
\CD(F,G) \approx \CD(\widetilde{F},\widetilde{G}),\quad \shift_\pm^{\CD}(F,G) \approx \shift_\pm^{\CD}(\widetilde{F},\widetilde{G}),\quad \disp_\pm^{\CD}(F,G) \approx \disp_\pm^{\CD}(\widetilde{F},\widetilde{G}).
\]

\subsection{Bounds}
\label{subsec:Bounds}

In this subsection, we provide bounds for distances and components in settings with incomplete distributions. These bounds are used to study the error of approximations such as the ones proposed above in the next section. For ease of exposition, we only consider a simplified setting. However, generalizations to the more general settings discussed in the previous section are straightforward.
As with the approximation of the previous section, the bounds are obtained from certain discrete distributions that are derived from the given quantiles to which the formulas from Section \ref{Sec:Discrete} are applied. Notably, we only treat the case of distributions with bounded support as otherwise the unknown tail behavior could theoretically inflate the distances arbitrarily. To circumvent this issue when considering distributions with unbounded support, one may consider a censored distance, where the distributions are bounded by censoring upper and lower tails at some mutual cutoff values.

Let $F$ and $G$ be known only via the quantiles at levels $0 < \beta_1 < \beta_2 < \dots \beta_k < 1$ with
\begin{align*}
	F^{-1}(\beta_i) = F_i, \qquad
	G^{-1}(\beta_i) = G_i  
\end{align*}
for $i = 1,\dots, k$.
In what follows, we assume that the distributions have a common bounded support $[l,u]$, and we set $F_0 = G_0 = l$, $F_{k+1} = G_{k+1} = u$, $\beta_0 = 0$, and $\beta_{k+1} = 1$. Furthermore, we assume that the levels are symmetric about the median level, i.e., $\beta_i = 1- \beta_{k+1-i}$.

\begin{figure}[tb!]
	\centering
	\subcaptionbox{AVM upper bound}{
		\includegraphics[page = 1,width = 0.4\textwidth]{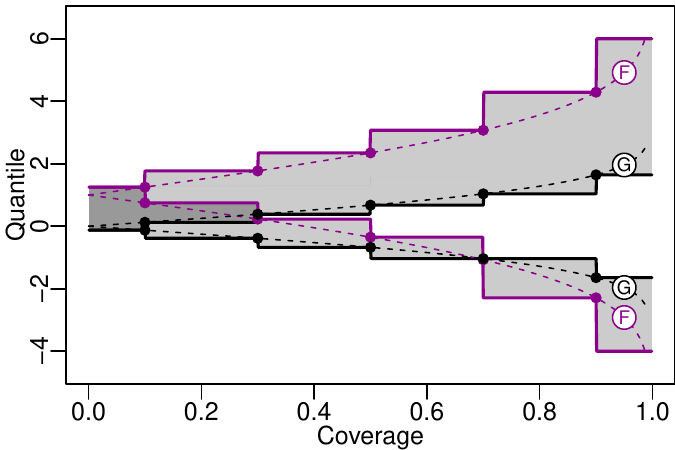}}
	\subcaptionbox{AVM lower bound}{
		\includegraphics[page = 2,width = 0.4\textwidth]{figures/FigS5.pdf}} \bigskip\\
	\subcaptionbox{AVM Shift upper bound}{
		\includegraphics[page = 3,width = 0.4\textwidth]{figures/FigS5.pdf}} 
	\subcaptionbox{AVM Dispersion upper bound}{
		\includegraphics[page = 4,width = 0.4\textwidth]{figures/FigS5.pdf}}
	\caption{Exemplary quantile spread plots illustrating upper and lower bounds for AVM and its components obtained from a a finite number of known quantiles (here, $n-1 = 10$) given by the solid dots from two unknown normal distributions (as in Figure \ref{fig:Illustration_FoldedQuantiles}). Panel (a) illustrates the upper bound \eqref{Eq:AVM_upper} for the overall AVM while panel (b) illustrates the lower bound \eqref{Eq:AVM_lower}. Panel (c) illustrates the upper bound $\eqref{Eq:AVM_Shift_upper}$ for the shift term (orange area). Note that the same distributions are used for the lower bound of the opposing shift term (here, the lower bound is zero). Analogously, panel (d) illustrates the upper bound $\eqref{Eq:AVM_Disp_upper}$ for the dispersion term (blue area). Again, the same distributions are used for the lower bound of the opposing dispersion term (here, the lower bound is again zero).
	}
	\label{Fig:Bounds}
\end{figure}

\newcommand{\up}{\max}
\newcommand{\lo}{\min}

The distance between the two distributions is maximized by moving the quantiles in between the known values as far away from one another as possible without violating the monotonicity of a quantile function. Thus, we obtain the distributions $F_{\up}, G_{\up}$ with quantile functions\footnote{Note that these quantile functions (and the following) are not necessarily left-continuous, as left-continuity may be violated when setting $F_{\up}^{-1}(\beta_i) = F_i$ for example. However, there are left-continuous quantile functions with the known quantiles that are arbitrarily close to these quantile functions.}
given by 
\begin{align*}
	F_{\up}^{-1}(\alpha) = \begin{cases}
		F_{i+1}& \text{if } F_{i+1} - G_{i} \geq G_{i+1} - F_{i}, \\
		F_{i}& \text{otherwise,}
	\end{cases}\qquad
	G_{\up}^{-1}(\alpha) = \begin{cases}
		G_{i}& \text{if } F_{i+1} - G_{i} \geq G_{i+1} - F_{i}, \\
		G_{i+1}& \text{otherwise,}
	\end{cases}
\end{align*}
for $\alpha \in (\beta_i,\beta_{i+1})$ that maximize the distance between distributions subject to having the known quantiles. 

Consequently, the AVM between the (unknown) distributions $F$ and $G$ is bounded by the distance between $F_{\up}$ and $G_{\up}$, as illustrated in Figure \ref{Fig:Bounds}(a). With $\overline{b}_i = \beta_{i+1} - \beta_i$, we obtain the tight upper bound
\begin{align}
	\AV(F,G) \leq \AV(F_{\up},G_{\up}) = \sum_{i = 0}^{k} \overline{b}_i \max\{F_{i+1} - G_{i}, G_{i+1} - F_{i}\}. \label{Eq:AVM_upper}
\end{align}
Analogously, we can move the qantiles as close as possible together to obtain a lower bound for the distance. We obtain the distributions $F_{\lo},G_{\lo}$ with quantiles at levels $\alpha \in (\beta_{i},\beta_{i+1})$ given by
\begin{align*}
	F_{\lo}^{-1}(\alpha) = \begin{cases}
		F_{i}& \text{if } F_{i} \geq G_{i}, \\
		\min\{F_{i+1},G_{i}\}& \text{otherwise},
	\end{cases}\qquad
	G_{\lo}^{-1}(\alpha) = \begin{cases}
		\min\{G_{i+1},F_{i}\}& \text{if } F_{i} \geq G_{i}, \\
		G_{i}& \text{otherwise},
	\end{cases}
\end{align*}
which minimize the distance subject to attaining the known quantiles.
The AVM between $F$ and $G$ is bounded by the distance between $F_{\lo}$ and $G_{\lo}$ from below, as illustrated in Figure \ref{Fig:Bounds}(b). Thus, we obtain the tight lower bound
\begin{align}
	\AV(F,G) \geq \AV(F_{\lo},G_{\lo}) = \sum_{i = 0}^{k} \overline{b}_i \left(F_{i} - \min\{G_{i+1},F_{i}\} + G_{i} - \min\{F_{i+1},G_{i}\}\right). \label{Eq:AVM_lower}
\end{align}
In a similar spirit, we may define distributions that maximize or minimize a particular component. The plus shift component, $\shift_+^{\D}$, may be bounded (from above) by shifting the quantiles of $F$ up as much as possible and the quantiles of $G$ down as much as possible, as illustrated in Figure \ref{Fig:Bounds}(c). On the other hand, it may be bounded from below by shifting the quantiles of $F$ down as much as possible and the quantiles of $G$ up as much as possible. Naturally, these distributions also yield bounds for the minus component. 
Thus, we define the distributions $F_\uparrow, F_\downarrow, G_\uparrow, G_\downarrow$ with quantiles at levels $\alpha \in (\beta_i,\beta_{i+1}), i = 0,\dots,k$ given by
\begin{align*}
	F_\uparrow^{-1}(\alpha) = F_{i+1}, \quad
	F_\downarrow^{-1}(\alpha) = F_i, \quad
	G_\uparrow^{-1}(\alpha) = G_{i+1}, \quad
	G_\downarrow^{-1}(\alpha) = G_i,
\end{align*}
respectively. With $\overline{a}_i = \overline{b}_i$ for $i \neq \frac k2$ and $\overline{a}_{k/2} = \frac{\overline{b}_{k/2}}2$ if $k$ is even, we obtain the bounds for the plus shift component as
\begin{align}
	\shift_+^{\AV}(F,G) &\geq \shift_+^{\AV}(F_\downarrow,G_\uparrow) = 2\sum_{i = 0}^{\lfloor \frac{k}{2}\rfloor} \overline{a}_i \left[\min\{F_{k-i} - G_{k-i + 1}, F_i - G_{i+1}\}\right]_+, \label{Eq:AVM_Shift_upper}\\
	\shift_+^{\AV}(F,G) &\leq \shift_+^{\AV}(F_\uparrow,G_\downarrow) = 2\sum_{i = 0}^{\lfloor \frac{k}{2}\rfloor} \overline{a}_i \left[\min\{F_{k-i + 1} - G_{k-i}, F_{i+1} - G_i\}\right]_+. \notag
\end{align}
The minus component is bounded via the usual symmetry.

\newcommand{\wide}{\longleftrightarrow}
\newcommand{\narrow}{\rightarrow\leftarrow}

The dispersion components are bounded by comparing a distribution with central intervals as wide as possible with a distribution with intervals as narrow as possible, as illustrated in Figure \ref{Fig:Bounds}(d). We define the distributions $F_{\wide}, F_{\narrow}, G_{\wide}, G_{\narrow}$ with quantiles at levels $\alpha \in (\beta_i,\beta_{i+1}), i = 0,\dots,k$ given by
\begin{align*}
	F_{\wide}^{-1}(\alpha) &= \begin{cases}
		F_{i+1},& \alpha > \frac12, \\
		F_i,& \text{otherwise},
	\end{cases} \quad
	F_{\narrow}^{-1}(\alpha) = \begin{cases}
		F_i,& \alpha > \frac12, \\
		\min\{F_{i+1},F_{n-i}\},& \text{otherwise},
	\end{cases} \\
	G_{\wide}^{-1}(\alpha) &= \begin{cases}
		G_{i+1},& \alpha > \frac12, \\
		G_i,& \text{otherwise},
	\end{cases} \quad
	G_{\narrow}^{-1}(\alpha) = \begin{cases}
		G_i,& \alpha > \frac12, \\
		\min\{G_{i+1},G_{n-i}\},& \text{otherwise}.
	\end{cases}
\end{align*}
Thus, we obtain the bounds for the dispersion component as 
\begin{align}
	\disp_+^{\AV}(F,G) &\geq \disp_+^{\AV}(F_{\narrow},G_{\wide}) = \sum_{i = 0}^{\lfloor \frac{k}{2}\rfloor} \overline{a}_i \left[(F_{k-i} - G_{k-i+1}) - (\min\{F_{i+1},F_{k-i}\} - G_i)\right]_+, \label{Eq:AVM_Disp_upper}\\
	\disp_+^{\AV}(F,G) &\leq \disp_+^{\AV}(F_{\wide},G_{\narrow}) = \sum_{i = 0}^{\lfloor \frac{k}{2}\rfloor} \overline{a}_i \left[(F_{k-i+1} - G_{k-i}) - (F_i - \min\{G_{i+1},G_{k-i}\})\right]_+. \notag
\end{align}

Lower and upper bounds for the other distances and their components are obtained analogously, e.g., for the CD, we obtain the bounds
\begin{align*}
	&\CD(F,G) \leq \CD(F_{\up},G_{\up}) \\ 
	& \hspace{5mm} = \sum_{i = j} \overline{b}_i\overline{b}_j \max\{F_{j+1} - G_i, G_{i+1} - F_j\} + \sum_{i > j} 2\overline{b}_i\overline{b_j}[F_{j} - G_{i+1}]_+ + \sum_{i < j} 2\overline{b}_i\overline{b_j}[G_{i} - F_{j+1}]_+, \\
	&\CD(F,G) \geq \CD(F_{\lo},G_{\lo}) \\ 
	& \hspace{5mm} = \sum_{i = j} \overline{b}_i\overline{b}_j \left(F_{j} - \min\{G_{i+1},F_{j}\} + G_{i} - \min\{F_{j+1},G_{i}\}\right) + \sum_{i > j} 2\overline{b}_i\overline{b_j}[F_{j+1} - G_i]_+  + \sum_{i < j} 2\overline{b}_i\overline{b_j}[G_{i+1} - F_j]_+. 
	\end{align*}

\subsection{Approximation error}
\label{subsec:ApproxError}

An approximation formula chooses distributions $F_\approx, G_\approx$ that respect the known quantiles, thus producing approximations in between the above bounds. Thus, the approximation error is bounded by the deviation between the lower and upper bounds for the distance or one of its components. For the AVM, we obtain (via a simple case distinction)
\begin{align*}
	&\vert\AV(F,G) - \AV(F_\approx,G_\approx) \vert \leq \AV(F_{\up},G_{\up}) - \AV(F_{\lo},G_{\lo}) \\
	&= \sum_{i = 0}^{k} \overline{b}_i \left(\max\{F_{i+1} - G_{i}, G_{i+1} - F_{i}\} -  F_{i} + \min\{G_{i+1},F_{i}\} - G_{i} + \min\{F_{i+1},G_{i}\}\right) \\
	&=  \sum_{i = 0}^{k} \overline{b}_i \left(\max\{F_{i+1} - F_{i} + \min\{F_{i},G_{i+1}\} - G_{i}, G_{i+1} - G_{i} + \min\{G_{i},F_{i+1}\} - F_{i}\}\right) \\
	&\leq \sum_{i = 0}^{k} \overline{b}_i \left(\max\{F_{i+1} - F_{i} + G_{i+1} - G_{i}, G_{i+1} - G_{i} + F_{i+1} - F_{i}\}\right) \\
	&= \sum_{i = 0}^{k} \overline{b}_i \left(F_{i+1} - F_{i} + G_{i+1} - G_{i}\right) \leq 2(u - l) \max_i \overline{b}_i = \mathcal{O}(k^{-1})
\end{align*}
if $\max_i \overline{b}_i = \mathcal{O}(k^{-1})$ (e.g., for equidistant quantile levels $a_i = \frac1k$).

Similar considerations can be applied to the components. Let $I= \{i \in \{0,1,\dots,\lfloor \tfrac{k}{2}\rfloor\}$ and define
\[I_1 = \{i \in I \mid F_{k-i} - G_{k-i+1} \leq F_{i} - G_{i+1}\},\quad I_2 = I \setminus I_1.\] 
For the plus shift component of the AVM, we obtain  
\begin{align*}
	&\vert \shift_+^{\AV}(F,G) - \shift_+^{\AV}(F_\approx,G_\approx)\vert \leq \shift_+^{\AV}(F_\uparrow,G_\downarrow) - \shift_+^{\AV}(F_\downarrow,G_\uparrow) \\
	&= 2\sum_{i = 0}^{\lfloor \frac{k}{2}\rfloor} \overline{a}_i \left( \left[\min\{F_{k-i+1} - G_{k-i}, F_{i+1} - G_i\}\right]_+ - \left[\min\{F_{k-i} - G_{k-i+1}, F_{i} - G_{i+1}\}\right]_+\right) \\
	&\leq 2\sum_{i \in I_1} \overline{a}_i \left( \left[F_{k-i+1} - G_{k-i}\right]_+ - \left[F_{k-i} - G_{k-i+1}\right]_+\right) + 2\sum_{i \in I_2} \overline{a}_i \left( \left[F_{i+1} - G_i\right]_+ - \left[F_{i} - G_{i+1}\right]_+\right) \\
	&\leq 2\sum_{i \in I_1} \overline{a}_i \left(F_{k-i+1} - F_{k-i} + G_{k-i+1} - G_{k-i}\right) + 2\sum_{i \in I_2} \overline{a}_i \left(F_{i+1} - F_{i} + G_{i+1} - G_i\right) \\
	&\leq 2\sum_{i = 0}^k \overline{b}_i \left( F_{i+1} - F_{i} + G_{i+1} - G_i\right) \leq 4(u-l)\max_i\overline{b}_i = \mathcal{O}(k^{-1})
\end{align*}
if $\max_i \overline{b}_i = \mathcal{O}(k^{-1})$.

For the dispersion component of the AVM, we obtain
\begin{align*}
	&\vert \disp_+^{\AV}(F,G) - \disp_+^{\AV}(F_\approx,G_\approx)\vert \leq \disp_+^{\AV}(F_{\wide},G_{\narrow}) - \disp_+^{\AV}(F_{\narrow},G_{\wide}) \\
	&= \sum_{i = 0}^{\lfloor \frac{k}{2}\rfloor} \overline{a}_i \left(\left[(F_{k-i+1} - G_{k-i}) - (F_i - \min\{G_{i+1},G_{k-i}\})\right]_+ - \left[(F_{k-i} - G_{k-i+1}) - (\min\{F_{i+1}, F_{k-i}\} - G_{i})\right]_+\right) \\
	&\leq \sum_{i = 0}^{\lfloor \frac{k}{2}\rfloor} \overline{a}_i \left((F_{k-i+1} - G_{k-i}) - (F_i - \min\{G_{i+1},G_{k-i}\}) - (F_{k-i} - G_{k-i+1}) - (\min\{F_{i+1}, F_{k-i}\} - G_{i})\right) \\
	&= \sum_{i = 0}^{\lfloor \frac{k}{2}\rfloor} \overline{a}_i \left(F_{k-i+1} - F_{k-i} + \min\{F_{i+1}, F_{k-i}\} - F_i + G_{k-i+1} - G_{k-i} + \min\{G_{i+1},G_{k-i}\} - G_{i}\right) \\
	&= \sum_{i = 0}^k \overline{b}_i (F_{i+1} - F_{i} + G_{i+1} - G_i) \leq 2(u-l)\max_i\overline{b}_i = \mathcal{O}(k^{-1})
\end{align*}
if $\max_i \overline{b}_i = \mathcal{O}(k^{-1})$.

Analogous considerations apply to the approximation error for the other distances and their components.
For instance for the CD, we obtain
\begin{align*}
	&\vert\CD(F,G) - \CD(F_\approx,G_\approx)\vert \leq \CD(F_{\up},G_{\up}) - \CD(F_{\lo},G_{\lo}) \\
	&\leq \sum_{i = j} \overline{b}_i\overline{b}_j (F_{j+1} - F_j + G_{i+1} - G_i) + \sum_{i \neq j} 2\overline{b}_i\overline{b_j} (F_{j+1} - F_j + G_{i+1} - G_i) \\
	&\leq 4k(u-l)(\max_i \overline{b_i})^2 = \mathcal{O}(k^{-1})
\end{align*}
as long as $\max_i \overline{b}_i = \mathcal{O}(k^{-1})$.

\section{Derivation of the decompositions}
\label{sec:ProofsDecomposition}

A simple case distinction yields the following lemma, which is used to derive the decompositions.

\begin{lemma}
	\label{Lem:Simple}
	For real numbers $A,B \in \mathbb{R}$ and $[A]_+ := \max(A,0)$, we have 
	\begin{align*}
		[A]_+ + [B]_+ = [A + B]_+ + \min \{A, -B\}]_+ + [\min \{-A,B\}]_+.
	\end{align*}
\end{lemma}

\begin{proof}[Proof of Lemma \ref{Lem:Simple}]
	The statement of Lemma \ref{Lem:Simple} follows from a simple case distinction:
	\begin{itemize}
		\item If $A,B > 0$, the statement is immediate as $[A]_+ + [B]_+ = A+B$ and $[A + B]_+ + [\min \{A, -B\}]_+ + [\min \{-A,B\}]_+ = A+B + 0 + 0$.
		\item If $A,B \le 0$, then $[A]_+ + [B]_+ = 0 + 0$ and $[A + B]_+ + [\min \{A, -B\}]_+ + [\min \{-A,B\}]_+ = 0 + 0 + 0$.
		\item If $A > 0$ and $B \le 0$, we get that $[A]_+ + [B]_+ = A + 0 = A$.
		If additionally $A \ge -B$, we get  $[A + B]_+ + [\min \{A, -B\}]_+ + [\min \{-A,B\}]_+ = (A+B) + (-B) + 0 = A$.
		Similarly, if $A < -B$, we get  $[A + B]_+ + [\min \{A, -B\}]_+ + [\min \{-A,B\}]_+ = 0 + A + 0 = A$.
		\item The result for $A \le 0$ and $B > 0$ follows from the previous case by interchanging $A$ and $B$.
	\end{itemize}
\end{proof}

\begin{proof}[Proof of Proposition \ref{prop:AVMExact}]
	The result follows directly from (the proof of) Proposition \ref{prop:WDExact} for $p=1$.
\end{proof}

\begin{proof}[Proof of Proposition \ref{prop:WDExact}]
	As in \eqref{Eq:WD_Shift}--\eqref{Eq:WD_Disp}, we denote by $z^{[p]} = \sgn(z) \cdot\vert z\vert^p$ the \emph{signed $p$-th power} of a number $z \in \R$. Let
	\begin{align*}
		A &= A(\alpha) = \Big(F^{-1}\left( \tfrac{1-\alpha}{2} \right) - G^{-1}\left( \tfrac{1-\alpha}{2} \right) \Big)^{[p]}, \\
		B &= B(\alpha) = \Big( F^{-1}\left( \tfrac{1+\alpha}{2} \right) - G^{-1}\left( \tfrac{1+\alpha}{2} \right) \Big)^{[p]},
	\end{align*}
	denote the signed powers of the difference between the lower and upper ends of the central intervals of $F$ and $G$, respectively.
	Then
	\begin{align*}
		\WD_p(F,G) &= \int_0^1 \big\vert F^{-1}(\tau) - G^{-1}(\tau) \big\vert^p \dd \tau \\
		&= \int_0^{0.5} \big\vert F^{-1}(\tau) - G^{-1}(\tau) \big\vert^p \dd \tau + \int_{0.5}^1 \big\vert F^{-1}(\tau) - G^{-1}(\tau) \big\vert^p \dd \tau \\
		&= \frac12\int_0^1 \big\vert F^{-1}(\tfrac{1-\alpha}{2}) - G^{-1}(\tfrac{1-\alpha}{2}) \big\vert^p \dd \alpha + \frac12\int_0^1 \big\vert F^{-1}(\tfrac{1+\alpha}{2}) - G^{-1}(\tfrac{1+\alpha}{2}) \big\vert^p \dd \alpha \\
		&= \frac12 \int_0^1 \big\vert A(\alpha) \big\vert + \big\vert B(\alpha) \big\vert \dd \alpha \\
		&= \frac12 \int_0^1 [A(\alpha)]_+ + [-A(\alpha)]_+ + [B(\alpha)]_+ + [-B(\alpha)]_+ \dd \alpha \\ 
		&= \frac12 \int_0^1 \big[B(\alpha) - A(\alpha)\big]_+ + \big[A(\alpha) - B(\alpha)\big]_+ + 2\big[\min\{A(\alpha),B(\alpha)\}\big]_+ + 2\big[\min\{-A(\alpha),-B(\alpha)\}\big]_+ \dd \alpha,
	\end{align*}
	where the third equality is obtained by a simple change of variables and the last equality is obtained by applying Lemma \ref{Lem:Simple} to $[A(\alpha)]_+ + [-B(\alpha)]_+$ and $[-A(\alpha)]_+ + [B(\alpha)]_+$.
	
	Thus, the components in \eqref{Eq:WD_Shift}--\eqref{Eq:WD_Disp} are obtained by assigning the summands in the above equation as follows:
	\begin{align*}
		\shift_+^{\WD_p}(F,G) &= \frac12 \int_0^1 2\big[\min\{A(\alpha), B(\alpha) \} \big]_+ \dd \alpha \\ 
		&= \int_0^1  \left[\min \left\{ \left(F^{-1}\left( \tfrac{1-\alpha}{2} \right) - G^{-1}\left( \tfrac{1-\alpha}{2} \right) \right)^{[p]}, \; \left( F^{-1}\left( \tfrac{1+\alpha}{2} \right) - G^{-1}\left( \tfrac{1+\alpha}{2} \right) \right)^{[p]} \right\}\right]_+  \dd\alpha, \\
		\shift_-^{\WD_p}(F,G) &= \frac12\int_0^1 2 \big[\min\{-A(\alpha),-B(\alpha) \}]_+ \dd \alpha = \shift_+^{\WD_p}(G,F),
	\end{align*}
	and
	\begin{align*}
		\disp_+^{\WD_p}(F,G) &= \frac12\int_0^1 \big[B(\alpha) - A(\alpha) \big]_+ \dd\alpha \\
		&= \frac12 \int_0^1 \left[ \left( F^{-1} \left( \tfrac{1+\alpha}{2} \right) - G^{-1}\left( \tfrac{1+\alpha}{2} \right)\right)^{[p]} - \left(F^{-1}\left( \tfrac{1-\alpha}{2} \right) - G^{-1}\left( \tfrac{1-\alpha}{2} \right) \right)^{[p]} \right]_+ \dd\alpha,  \\
		\disp_-^{\WD_p}(F,G) &= \frac12\int_0^1 \big[A(\alpha) - B(\alpha) \big]_+ \dd\alpha = \disp_+^{\WD_p}(G,F).
	\end{align*}
\end{proof}

\begin{proof}[Proof of Proposition \ref{prop:CDQuantileRep}]
	We show the desired equality in three steps, which we treat individually below:
	\begin{align}
		\int_{-\infty}^\infty (F(x) - G(x))^2\dd x
		&= \int_0^1\int_{-\infty}^\infty (F(x) - \one\{G^{-1}(\xi) \leq x\})^2 - (G(x) - \one\{G^{-1}(\xi) \leq x\})^2 \dd x\dd \xi \label{Eq:CD_quantile_step1} \\
		&= 2 \int_0^1 \int_0^1 (\one\{G^{-1}(\xi) \leq F^{-1}(\tau)\} - \tau)(F^{-1}(\tau) - G^{-1}(\xi)) \label{Eq:CD_quantile_step2} \\
		&~~~~~~~~~~~~~~~ - (\one\{G^{-1}(\xi) \leq G^{-1}(\tau)\} - \tau)(G^{-1}(\tau) - G^{-1}(\xi)) \dd\tau\dd\xi \notag \\
		&= 2 \int_0^1 \int_0^1\chi(\tau,\xi)\vert F^{-1}(\tau) - G^{-1}(\xi)\vert\dd\xi \dd \tau.
		\label{Eq:CD_quantile_step3} 
	\end{align}
	The first step \eqref{Eq:CD_quantile_step1} essentially rewrites the Cram\'er distance as the divergence function associated with the continuous ranked probability score \citepsupp[CRPS;][]{GR07}.
	By noting that $G(z)
	= \int_0^1 \one\{\xi \leq G(z)\}\dd\xi = \int_0^1\one\{G^{-1}(\xi) \leq z\}\dd\xi$, we obtain
	\begin{align*}
		&\int_0^1\int_{-\infty}^\infty (F(x) - \one\{G^{-1}(\xi) \leq x\})^2 - (G(x) - \one\{G^{-1}(\xi) \leq x\})^2 \dd x\dd \xi \\ 
		&= \int_{-\infty}^\infty\int_0^1 F(x)^2 - 2F(x)\one\{G^{-1}(\xi) \leq x\} + \one\{G^{-1}(\xi) \leq x\}^2 - G(x)^2 + 2G(x)\one\{G^{-1}(\xi) \leq x\} - \one\{G^{-1}(\xi) \leq x\}^2 \dd \xi\dd x \\
		&=\int_{-\infty}^\infty F(x)^2 - 2F(x)\int_0^1\one\{G^{-1}(\xi) \leq x\}\dd\xi + 2G(x)\int_0^1\one\{G^{-1}(\xi) \leq x\}\dd\xi - G(x)^2\dd x \\
		&= \int_{-\infty}^\infty F(x)^2 - 2F(x)G(x) + 2G(x)G(x) - G(x)^2\dd x \\
		&=\int_{-\infty}^\infty (F(x) - G(x))^2\dd x,
	\end{align*}
	i.e., the first equality in \eqref{Eq:CD_quantile_step1}.
	
	The second step \eqref{Eq:CD_quantile_step2} is essentially equivalent to rewriting the CRPS in terms of quantile scores \citepsupp{GR11} and proceeds as suggested by \citetsupp[][Eq.\ (6.4) and (6.5)]{Jordan2016}:
	To derive the second equality, it suffices to rewrite the inner integral in \eqref{Eq:CD_quantile_step2}. We write $y = G^{-1}(\xi)$ for $\xi \in (0,1)$ for ease of exposition. Note that the integral in \eqref{Eq:CD_quantile_step2} is made up of two similar terms, where the second term is rewritten analogously by replacing $F^{-1}(\tau)$ with $G^{-1}(\tau)$ in the following. With these remarks, we obtain the second equality from
	\begin{align*}
		&\int_0^1 2(\one\{y \leq F^{-1}(\tau)\} - \tau)(F^{-1}(\tau) - y) \dd\tau \\    
		&= \int_0^1 \int_y^{F^{-1}(\tau)} 2(\one\{y \leq F^{-1}(\tau)\} - \tau)\dd x \dd\tau\\
		&= \int_0^1\int_{-\infty}^\infty  2(\one\{x < F^{-1}(\tau)\} - \one\{x < y\})(\one\{y \leq F^{-1}(\tau) \} -\tau)\dd x \dd \tau \\
		&= \int_{-\infty}^\infty \int_0^1 2(\one\{x < F^{-1}(\tau)\} - \one\{x < y\})(\one\{y \leq F^{-1}(\tau) \} -\tau)\dd\tau \dd x \\
		&= \int_{-\infty}^y \int_0^1 \underbrace{2(-\one\{x \geq F^{-1}(\tau)\})(\one\{y \leq F^{-1}(\tau) \}-\tau)}_{\neq 0, \text{ only if } F^{-1}(\tau) \leq x \leq y \text{ for } x\in(-\infty,y]}\dd\tau \dd x + \int_y^\infty \int_0^1 \underbrace{2(\one\{x < F^{-1}(\tau)\})(\one\{y \leq F^{-1}(\tau) \}-\tau)}_{\neq 0, \text{ only if } y \leq x <F^{-1}(\tau)\text{ for } x\in[y,\infty)}\dd\tau \dd x \\
		&= \int_{-\infty}^y \int_0^1 2(-\one\{F(x) \geq \tau\})(0-\tau)\dd\tau \dd x + \int_y^\infty \int_0^1 2(\one\{F(x) < \tau\})(1-\tau)\dd\tau \dd x \\
		&= \int_{-\infty}^y \int_0^{F(x)}2\tau\dd\tau \dd x + \int_y^\infty \int_{F(x)}^1 2(1-\tau)\dd\tau \dd x \\
		&= \int_{-\infty}^y F(x)^2 \dd x + \int_y^\infty (1-F(x))^2 \dd x \\
		&= \int_{-\infty}^\infty (F(x) - \one\{y \leq x\})^2 \dd x.
	\end{align*}
	The third step \eqref{Eq:CD_quantile_step3} further simplifies the quantile-based representation of the CD to derive the concise formula presented in the proposition:
	\begin{align*}
		2 \int_0^1 &\int_0^1 (\one\{G^{-1}(\xi) \leq F^{-1}(\tau)\} - \tau)(F^{-1}(\tau) - G^{-1}(\xi)) - (\one\{G^{-1}(\xi) \leq G^{-1}(\tau)\} - \tau)(G^{-1}(\tau) - G^{-1}(\xi)) \dd\tau\dd\xi \\
		= 2 \int_0^1 &\int_0^\tau (\underbrace{\one\{G^{-1}(\xi) \leq F^{-1}(\tau)\}}_{= 1 - \chi(\tau,\xi), \text{ as } \xi \leq \tau} - \tau)(F^{-1}(\tau) - G^{-1}(\xi)) - (1-\tau)(G^{-1}(\tau) - G^{-1}(\xi)) \dd\xi \\
		+ &\int_\tau^1 (\underbrace{\one\{G^{-1}(\xi) \leq F^{-1}(\tau)\}}_{= \chi(\tau,\xi), \text{ as } \xi \geq \tau} - \tau)(F^{-1}(\tau) - G^{-1}(\xi)) + \tau (G^{-1}(\tau) - G^{-1}(\xi)) \dd\xi\dd \tau \\
		= 2 \int_0^1 &\underbrace{\int_0^\tau (1 - \tau)(F^{-1}(\tau) - G^{-1}(\tau))\dd\xi}_{= \tau(1 - \tau)(F^{-1}(\tau) - G^{-1}(\tau))} - \int_0^\tau\underbrace{\chi(\tau,\xi)(F^{-1}(\tau) - G^{-1}(\xi))}_{\leq 0 \text{ as } \xi \leq \tau}\dd\xi \\
		+ &\underbrace{\int_\tau^1 (-\tau)(F^{-1}(\tau) - G^{-1}(\tau))\dd\xi}_{=-(1-\tau)\tau(F^{-1}(\tau) - G^{-1}(\tau))} + \int_\tau^1\underbrace{\chi(\tau,\xi)(F^{-1}(\tau) - G^{-1}(\xi)) }_{\geq 0 \text{ as } \xi \geq \tau}\dd\xi \dd \tau \\
		= 2 \int_0^1 &\int_0^1\chi(\tau,\xi)\vert F^{-1}(\tau) - G^{-1}(\xi)\vert\dd\xi \dd \tau.
	\end{align*}
	
	As all integrands are non-negative, changing the order of integration throughout the proof is not an issue by Fubini-Tonelli. No assumptions on the distributions are needed as long as integrals are allowed to be infinite. In the literature on proper scoring rules (in particular, the CRPS), existence of first moments is typically assumed to obtain a meaningful scoring rule that is strictly proper. For divergences, comparisons of two distributions without first moments might in some cases still yield a finite distance. Therefore, we dispense with such an assumption and provide a fully general proof.
\end{proof}

\begin{proof}[Proof of Proposition \ref{prop:CDExact}]
	We start to define the following four differences of central interval endpoints
	\begin{align*}
		A &= A(\alpha,\beta) = F^{-1} \left( \tfrac{1-\alpha}{2} \right) - G^{-1}\left( \tfrac{1-\beta}{2} \right), \qquad \quad
		B = B(\alpha,\beta) = F^{-1}\left( \tfrac{1+\alpha}{2} \right) - G^{-1}\left( \tfrac{1+\beta}{2} \right), \\
		C &= C(\alpha,\beta) = F^{-1}\left( \tfrac{1-\alpha}{2} \right) - G^{-1}\left( \tfrac{1+\beta}{2} \right), \qquad \quad
		D = D(\alpha,\beta) = F^{-1}\left( \tfrac{1+\alpha}{2} \right) - G^{-1}\left( \tfrac{1-\beta}{2} \right).
	\end{align*}
	We further define the function
	\begin{align*}
		f(\tau,\xi) 
		&:= \chi(\tau,\xi)  \, \big\vert F^{-1}(\tau) - G^{-1}(\xi) \big\vert  \\
		&= \one \left\{ \sgn(\tau - \xi) \neq \sgn \big( F^{-1}(\tau) - G^{-1}(\xi) \big) \right\} \big\vert F^{-1}(\tau) - G^{-1}(\xi) \big\vert.
	\end{align*}
	for all $\tau, \xi \in [0,1]$, which arises as the integrand in \eqref{Eq:CD_quantile}.
	By the following case distinctions, we note that
	\begin{align*}
		f \left(\tfrac{1-\alpha}{2},\tfrac{1-\beta}{2} \right) 
		&= 
		\begin{cases}
			0,& \text{if } \alpha \leq \beta \text{ and } A \geq 0 \\
			-A,& \text{if } \alpha \leq \beta \text{ and } A \leq 0 \\
			A,& \text{if } \alpha \geq \beta \text{ and } A \geq 0 \\
			0,& \text{if } \alpha \geq \beta \text{ and } A \leq 0
		\end{cases} \\
		&= \one\{\alpha \geq \beta\} [A]_+ + \one\{\alpha \leq \beta\} [-A]_+.
	\end{align*}
	Similar considerations yield that
	\begin{align*}
		f \left(\tfrac{1-\alpha}{2},\tfrac{1+\beta}{2} \right) = [C]_+, \quad 
		f \left(\tfrac{1+\alpha}{2},\tfrac{1-\beta}{2} \right) = [-D]_+, \quad
		f \left(\tfrac{1+\alpha}{2},\tfrac{1+\beta}{2} \right) = \one\{\alpha \geq \beta\} [-B]_+ + \one\{\alpha \leq \beta\} [B]_+.
	\end{align*}    
	
	Then, by a transformation of variables in the third equality, plugging in the above expressions in the fourth equality and by applying Lemma \ref{Lem:Simple} in the fifth equality below, we get that 
	\begin{align*}
		\CD(F,G) 
		&= 2\int_0^1\int_0^1 f(\tau,\xi) \dd\tau\dd\xi \\
		&= 2 \left[\int_0^{0.5}\int_0^{0.5} f(\tau,\xi) \dd\tau\dd\xi 
		+ \int_0^{0.5}\int_{0.5}^1 f(\tau,\xi) \dd\tau\dd\xi 
		+ \int_{0.5}^1\int_0^{0.5} f(\tau,\xi) \dd\tau\dd\xi 
		+ \int_{0.5}^1\int_{0.5}^1 f(\tau,\xi) \dd\tau\dd\xi\right] \\
		&= 2 \left[\frac14\int_0^1\int_0^1 f\big(\tfrac{1-\alpha}{2},\tfrac{1-\beta}{2}\big) \dd\alpha\dd\beta 
		+ \frac14\int_0^1\int_0^1 f\big(\tfrac{1-\alpha}{2},\tfrac{1+\beta}{2}\big) \dd\alpha\dd\beta \right. \\
		&\qquad\;\left. + \frac14\int_0^1 \int_0^1 f\big(\tfrac{1+\alpha}{2},\tfrac{1-\beta}{2}\big) \dd\alpha\dd\beta 
		+ \frac14\int_0^1\int_0^1 f\big(\tfrac{1+\alpha}{2},\tfrac{1+\beta}{2}\big) \dd\alpha\dd\beta\right] \\
		&= \frac12 \int_0^1\int_0^1 \one\{\alpha \geq \beta\} [A]_+ + \one\{\alpha \leq \beta\} [-A]_+ + [C]_+ + [-D]_+ + \one\{\alpha \geq \beta\} [-B]_+ + \one\{\alpha \leq \beta\} [B]_+ \dd\alpha \dd\beta \\
		&= \frac12 \int_0^1\int_0^1 \one\{\alpha \geq \beta\} \Big( [A]_+ + [-B]_+ \Big) + \one\{\alpha \leq \beta\} \Big( [-A]_+ + [B_+] \Big) + [C]_+ + [-D]_+ \dd\alpha \dd\beta \\
		&= \frac12 \int_0^1\int_0^1 \one\{\alpha \geq \beta\} \Big( [A - B]_+ + [\min \{A, B\}]_+ + [\min \{-A,-B\}]_+  \Big)  \\
		&\hspace{1.8cm}+ \one\{\alpha \leq \beta\} \Big( [-A + B]_+ + [\min \{-A, -B\}]_+ + [\min \{A,B\}]_+ \Big) + [C]_+ + [-D]_+ \dd\alpha \dd\beta \\
		&= \frac12 \int_0^1 \int_0^1 \one\{\alpha \geq \beta\} [A - B]_+ +  \one\{\alpha \leq \beta\} [B - A]_+  \\
		&\hspace{1.8cm}+ [\min\{A,B\}]_+ + [\min\{-A,-B\}]_+ + [C]_+ + [-D]_+ \dd\alpha \dd\beta.
	\end{align*}

	The components of the decomposition in \eqref{Eq:CD_Shift}--\eqref{Eq:CD_Disp} are then obtained by setting
	\begin{align*}
		\disp_+^{\CD}(F,G) &= \frac12 \int_0^1 \int_0^\beta [B - A]_+ \dd\alpha \dd\beta \\
		\shift_+^{\CD}(F,G) &= \frac12 \int_0^1 \int_0^1 [\min\{A,B\}]_+ + [C]_+ \dd\alpha \dd\beta, \\
		\disp_-^{\CD}(F,G) &= \frac12 \int_0^1 \int_0^\alpha [A - B]_+ \dd\beta \dd\alpha, \\
		\shift_-^{\CD}(F,G) &= \frac12 \int_0^1 \int_0^1 [\min\{-A,-B\}]_+ + [-D]_+ \dd\alpha \dd\beta.
	\end{align*}
	Notice that for the two minus components with subscript `$-$', changing the roles of $F$ and $G$ merely changes the sign of $A$ and $B$.
\end{proof}

\section{Derivation of theoretical properties}
\label{sec:ProofsProperties}

We start by proving Propositions \ref{prop:PositiveDispersion} and \ref{prop:PositiveShift} as these help to simplify the proofs of some of the basic properties given in Section \ref{subsec:BasicProperties}.
Notice that the proofs of Propositions \ref{prop:PositiveDispersion} and \ref{prop:PositiveShift} do not require  any of the results from Section \ref{subsec:BasicProperties} apart from the obvious symmetry given by Proposition \ref{prop:Symmetry}.

\subsection{Proofs of propositions on equivalence of nonzero components}

\begin{proof}[Proof of Proposition~\ref{prop:PositiveDispersion}]
	For the Wasserstein distances, the equivalence between nonzero dispersion components is clear, as the signed $p$-th power preserves nonzero values in the integrands. The following similar (but technical) argument for the plus component shows that the CD dispersion components are nonzero whenever the respective AVM dispersions are nonzero by the symmetry from Proposition \ref{prop:Symmetry}.
	
	Let $B = \{\alpha\mid \disp_{\alpha,+}^{\AV}(F,G) > 0\}$ be the set of all values of the integration variable such that the integrand given in \eqref{Eq:AVM_Disp} is nonzero. The set $B$ has positive Lebesgue measure if $\disp_+^{\AV}(F,G) > 0$. On the other hand, for $\beta \in B$, the inner integral in the CD dispersion,
	\[
	\int_0^\beta \left[\left(F^{-1}\left(\tfrac{1+\alpha}{2}\right) - F^{-1}\left(\tfrac{1-\alpha}{2}\right)\right) - \left(G^{-1}\left(\tfrac{1+\beta}{2}\right) - G^{-1}\left(\tfrac{1-\beta}{2}\right)\right)\right]_+ \dd \alpha,
	\]
	is zero only if $\alpha \mapsto F^{-1}\left(\frac{1+\alpha}{2}\right) - F^{-1}\left(\frac{1-\alpha}{2}\right)$ is discontinuous at $\beta \in B$ (as otherwise, we can find a value $\beta' < \beta$ such that the integrand is strictly positive for all $\alpha \in [\beta',\beta]$, which yields a nonzero inner integral). 
	By left-continuity, the quantile function $F^{-1}$ is discontinuous at countably many values at most. Hence, the inner integral is nonzero for almost all $\beta \in B$, and integration across $\beta \in B$ yields a nonzero dispersion component.
	
	Conversely, if the inner integral is nonzero in the CD dispersion, the integrand in the AVM dispersion component is also nonzero (because the integrand in the displayed term is increasing in $\alpha$). Hence, the reverse implication also holds.
\end{proof}

\begin{proof}[Proof of Proposition~\ref{prop:PositiveShift}]
	Here, we proceed similarly as in the proof of Proposition \ref{prop:PositiveDispersion}.
	For the Wasserstein distances, the equivalence between nonzero shift components is clear, as the signed $p$-th power preserves nonzero values in the integrands. A similar (but technical) argument for the plus component shows that the CD shift components are nonzero whenever the respective AVM shifts are nonzero by the symmetry from Proposition \ref{prop:Symmetry}. 
	
	Let $B = \{\alpha\mid \shift_{\alpha,+}^{\AV}(F,G) > 0\}$ be the set of all values of the integration variable such that the integrand given in \eqref{Eq:AVM_Shift} is nonzero. The set $B$ has positive Lebesgue measure if $\shift_+^{\AV}(F,G) > 0$. On the other hand, for $\beta \in B$, the inner integral in the CD shift, 
	\[
	\int_{0}^1 \left[\min\left\{ F^{-1}\left(\tfrac{1+\alpha}{2}\right) - G^{-1}\left(\tfrac{1+\beta}{2} \right), \; F^{-1}\left(\tfrac{1-\alpha}{2}\right) - G^{-1}\left(\tfrac{1-\beta}{2}\right) \right\} \right]_+ + \left[F^{-1}\left(\tfrac{1-\alpha}{2}\right) - G^{-1}\left(\tfrac{1+\beta}{2}\right)\right]_+ \dd \alpha.
	\]
	is zero only if $F^{-1}$ is discontinuous at $\tfrac{1+\beta}{2}$ or $\tfrac{1-\beta}{2}$, which is only true for at most countably many values $\beta$ by left-continuity of the quantile function $F^{-1}$. As the inner integral is nonzero for almost all $\beta \in B$, integration across $\beta \in B$ yields a nonzero shift component.
	
	In contrast to the proof of Proposition \ref{prop:PositiveDispersion}, the converse is not true, because the integrand is not monotonic as illustrated by Example \ref{Ex:DiffShift}.
\end{proof}

\subsection{Proofs of basic properties}

\begin{proof}[Proof of Proposition \ref{prop:DispInvarianceforShifts}]
	Note that $F_s^{-1}(z) = F^{-1}(z) + s$, and hence the shift $s$ cancels out in the dispersion terms of the $\AV$ and $\CD$.
\end{proof}

\begin{proof}[Proof of Proposition \ref{prop:ShiftSymDist}]
	\begin{enumerate}
		\item 
		We start by showing the claim for the $\CD$.
		Suppose $m_F - m_G \leq 0$. For almost all pairs $(\alpha,\beta) \in (0,1)^2$, either $F^{-1}(\frac{1-\alpha}{2}) - G^{-1}(\frac{1-\beta}{2}) \leq 0$ or $F^{-1}(\frac{1+\alpha}2) -G^{-1}(\frac{1+\beta}2) = 2(m_F - m_G) - (F^{-1}(\frac{1-\alpha}{2}) - G^{-1}(\frac{1-\beta}{2})) \leq 0$ by symmetry.
		Therefore, the minimum across these two terms is for almost all  $(\alpha,\beta) \in (0,1)^2$ not positive. Furthermore, we have $F^{-1}(\tfrac{1-\alpha}{2}) \leq m_F \leq m_G \leq G^{-1}(\tfrac{1+\beta}{2})$. Hence, the shift component
		\begin{align*}
			\shift_+^{\CD}(F,G) = \frac{1}{2} \int_{0}^{1}\int_{0}^1 &\Big[\underbrace{\min\left\{F^{-1}\left(\tfrac{1+\alpha}{2}\right) - G^{-1}\left(\tfrac{1+\beta}{2}\right), F^{-1}\left(\tfrac{1-\alpha}{2}\right) - G^{-1}\left(\tfrac{1-\beta}{2}\right)\right\}}_{\leq 0 \;\;  \text{for almost all } (\alpha,\beta) \in (0,1)^2}\Big]_+ \\ &+ \Big[\underbrace{F^{-1}\left(\tfrac{1-\alpha}{2}\right) - G^{-1}\left(\tfrac{1+\beta}{2}\right)}_{\le 0}\Big]_+\dd\alpha\dd\beta = 0.
		\end{align*}
		From Proposition \ref{prop:PositiveShift} (shown above), it follows that $\shift_+^{\WDp}(F,G) = 0$ is zero as well for all $p \in [1,\infty)$.
		
		\item 
		By contraposition to (a), a positive shift ($\shift_+^{\D}(F,G) > 0$) implies a corresponding ordering of medians ($m_F > m_G$).
		
		To finish the proof, we show that a strict ordering of \emph{unique} medians, $m_F > m_G$, implies a positive shift component of the area validation metric. 
		(For the other distances, the shift component is then positive by Proposition \ref{prop:PositiveShift}.)
		By continuity of the quantile functions at $\frac12$ (otherwise the medians would not be unique), there exists a neighborhood $\left[\frac12-\delta,\frac12 + \delta\right] \subset (0,1)$ for some (small enough) $\delta \in (0,1)$ such that $F^{-1}(\tau) > G^{-1}(\tau)$ holds for all $\tau \in \left[\frac12-\delta,\frac12 + \delta\right]$. Therefore, we obtain
		\[
		\shift_+^{\AV}(F,G) \geq \int_0^{2\delta} \Big[\min\Big\{\underset{{>0,\text{ as }\tfrac{1+\alpha}2 \leq \frac12 + \delta}}{\underbrace{F^{-1}(\tfrac{1+\alpha}2) - G^{-1}(\tfrac{1+\alpha}2)}}, \; \underset{{>0,\text{ as }\tfrac{1-\alpha}2 \geq \frac12 - \delta}}{\underbrace{F^{-1}(\tfrac{1-\alpha}2) - G^{-1}(\tfrac{1-\alpha}2)}} \Big\} \Big]_+ \dd\alpha > 0.
		\]
	\end{enumerate}
\end{proof}

\begin{proof}[Proof of Proposition~\ref{Prop:CompLocScale}]
	By Propositions \ref{prop:PositiveDispersion} and \ref{prop:PositiveShift} (proved above), it suffices to show the equivalences in (a) and (b) for the AVM.
	\begin{enumerate}
		\item The AVM dispersion component is given by
		\begin{align*}
			\disp_+^{\AV}(F,G) &= \frac12 \int_0^1\left[\left(F^{-1}\left(\tfrac{1+\alpha}2\right) - F^{-1}\left(\tfrac{1-\alpha}2\right)\right) - \left(G^{-1}\left(\tfrac{1-\alpha}2\right) - G^{-1}\left(\tfrac{1+\alpha}2\right)\right)\right]_+ \dd\alpha \\
			&= \frac12 \int_0^1\left[ (s_F - s_G)\left(H^{-1}\left(\tfrac{1+\alpha}2\right) - H^{-1}\left(\tfrac{1-\alpha}2\right)\right)\right]_+\dd\alpha \\
			&= \frac12 [s_F - s_G]_+ \int_0^1 H^{-1}\left(\tfrac{1+\alpha}2\right) - H^{-1}\left(\tfrac{1-\alpha}2\right) \dd\alpha,
		\end{align*}
		where the latter integral is nonzero as $H$ is non-degenerate,
		and hence the equivalence holds.
		
		\item Without loss of generality, we assume that the central median $m_H$ of $H$ is 0 (as we can always standardize an arbitrary reference distribution $H$ by replacing it with $\overline{H}$ given by $\overline{H}^{-1} = H^{-1} - m_H$). Then, the location parameters $\ell_F$ and $\ell_G$ are the central medians, $m_F$ and $m_G$, of $F$ and $G$, respectively. 
		
		Now suppose that $m_F - m_G \leq 0$.
		If $s_F \geq s_G$, then $F^{-1}\left(\frac{1-\alpha}{2}\right) - G^{-1}\left(\frac{1-\alpha}{2}\right) = m_F - m_G + (s_F - s_G) H^{-1}\left(\frac{1-\alpha}2\right) \leq 0$ holds for all coverages $\alpha \in (0,1)$ as $0 = m_H \geq H^{-1}\left(\frac{1-\alpha}2\right)$. Otherwise, if $s_F \leq s_G$, then $F^{-1}\left(\frac{1+\alpha}2\right) -G^{-1}\left(\frac{1+\alpha}2\right) = m_F - m_G + (s_F - s_G) H^{-1}\left(\frac{1+\alpha}2\right) \leq 0$ holds for all $\alpha \in (0,1)$.
		Therefore, the minimum across these two terms is non-positive $\alpha \in (0,1)$. 
		Hence, the shift component
		\[
		\shift_+^{\AV}(F,G) = \frac{1}{2} \int_{0}^{1} \Big[\underbrace{\min\Big\{F^{-1}\left(\tfrac{1+\alpha}{2}\right) - G^{-1}\left(\tfrac{1+\alpha}{2}\right), F^{-1}\left(\tfrac{1-\alpha}{2}\right) - G^{-1}\left(\tfrac{1-\alpha}{2}\right)\Big\}}_{\leq 0}\Big]_+ \dd\alpha = 0
		\]
		is zero.
		Proposition~\ref{prop:PositiveShift} implies that $\shift_+^{\WD_p}(F,G) = 0$ for all $p \in [1,\infty)$.
		
		The proof is finished by following the proof of Proposition~\ref{prop:ShiftSymDist} (b) word by word.
	\end{enumerate}
\end{proof}

\begin{proof}[Proof of Theorem~\ref{thm:Uniqueness}]
	\begin{enumerate}
		\item
		Let $F$ and $G$ be distributions from the same location-scale family with location parameters $\ell_F$ and $\ell_G$, and scale parameters $s_F$ and $s_G$, respectively. As in part (b) of Proposition~\ref{Prop:CompLocScale}, we assume w.l.o.g.\ that the location parameters match the central medians of the distributions, i.e., $\ell_F = m_F$ and $\ell_G = m_G$. Let $F_{m_G - m_F}$ be the shifted version of $F$ that has location parameter $l_G$ and scale $s_F$. By condition \eqref{Eq:ShiftLocScale} and the symmetry in \eqref{Eq:Symmetry}, we obtain
		\[\shift_\pm^{\AV}(F_{m_G - m_F},G) = 0,\]
		and condition \eqref{Eq:DispLocScale} (plus symmetry in \eqref{Eq:Symmetry}) yields
		\[\AV(F_{m_G - m_F},G) = \begin{cases}
			\disp_+^{\AV}(F_{m_G - m_F},G),& \text{if } s_F \geq s_G, \\
			\disp_-^{\AV}(F_{m_G - m_F},G),& \text{if } s_F < s_G. 
		\end{cases}\]
		By invariance of dispersion components to shifts in  \eqref{Eq:DispInvarianceforShifts}, we obtain the unique dispersion terms
		\begin{align*}
			\disp_\pm^{\AV}(F,G) &= \disp_\pm^{\AV}(F_{m_G - m_F},G) \\
			&= \begin{cases}
				\AV(F_{m_G - m_F},G),& \text{if } \{\pm = + \text{ and } s_F \geq s_G\} \text{ or } \{\pm = - \text{ and } s_F < s_G\}, \\
				0,& \text{if } \{\pm = - \text{ and } s_F \geq s_G\} \text{ or } \{\pm = + \text{ and } s_F < s_G\}.
			\end{cases}
		\end{align*}
		Invoking \eqref{Eq:ShiftLocScale} (plus the symmetry in \eqref{Eq:Symmetry}) once more, we obtain the unique shift terms
		\begin{align*}
			&\shift_\pm^{\AV}(F,G)  \\
			&\quad= \begin{cases}
				\AV(F,G) - \AV(F_{m_G - m_F},G),& \text{if } \{\pm = + \text{ and } m_F \geq m_G\} \text{ or } \{\pm = - \text{ and } m_F < m_G\}, \\
				0,& \text{if } \{\pm = - \text{ and } m_F \geq m_G\} \text{ or } \{\pm = + \text{ and } m_F < m_G\}.
			\end{cases}
		\end{align*}
		The previous formulas for $\disp_\pm^{\AV}(F,G)$ and $\shift_\pm^{\AV}(F,G)$ provide a unique form for the decomposition terms after invoking the respective properties.
		As our decomposition given in Proposition \ref{prop:AVMExact} is one candidate for a decomposition that satisfies these properties, and there can only be one, we can conclude that the previously derived terms must equal our decomposition, which concludes this proof.
		
		\item The proof proceeds with analogous arguments as for the AVM in (a) with \eqref{Eq:ShiftLocScale} replaced by \eqref{Eq:ShiftSymDist}.
	\end{enumerate}
\end{proof}

\subsection{Proofs of theorems on comparisons across distances}

To prove Theorem \ref{Thm:DispInequality_pWD}, we need the following lemma.
\begin{lemma}
	Let $q > p \geq 1$. Then, the inequality
	\begin{equation}
		a^p + (a-2)^p \leq a^q + (a-2)^q
		\label{Eq:AuxIneq}
	\end{equation}
	holds for $a \geq 2$.
	\label{Lem:AuxIneq}
\end{lemma}

\begin{proof}
	The function $g(p) = a^p + (a-2)^p$ has first and second derivative
	\[g'(p) = a^p\ln(a) + (a-2)^p \ln(a-2),\quad\text{and}\quad g''(p) = a^p\ln(a)^2 + (a-2)^p\ln(a-2)^2,\]
	respectively. For the inequality to hold, it suffices to show that the first derivative $g'$ is nonnegative (for $p\geq 1$) and thus the function $g$ is increasing. As the second derivative is clearly positive, the function is (strictly) convex.
	Thus, it suffices to show that $g'(1) = a\ln(a) + (a-2)\ln(a-2) = h(a) \geq 0$ for $a \geq 2$. The function $h$ has first and second derivative
	\[h'(a) = \ln(a) + \ln(a-2) + 2\quad \text{and}\quad h''(a) = \frac1a + \frac1{a-2},\]
	respectively. It is easy to see that $h$ is convex on $[2,\infty)$ with a minimum at $a_0 = 1 + \sqrt{1 + e^{-2}}$. Therefore, we have $g'(1) = h(a) \geq h(a_0) > 0$ as $h(a_0) \approx 1.32$.
\end{proof}

\begin{proof}[Proof of Theorem \ref{Thm:DispInequality_pWD}]
	Let $\widetilde{m} = m_F - m_G$ be the difference between the central medians of $F$ and $G$ (as defined prior to Proposition \ref{prop:ShiftSymDist}).
	If $\widetilde{m} = 0$, then both shift components are zero by Proposition \ref{prop:ShiftSymDist} and inequality \eqref{Eq:DispInequality_pWD} is trivial.
	We assume w.l.o.g.\ that $\widetilde{m} > 0$ (otherwise, we exchange $F$ and $G$). Then, the minus shift components are zero by Proposition \ref{prop:ShiftSymDist}. Furthermore, we define $f(\tau) = F^{-1}(\tau) - G^{-1}(\tau)$ for $\tau \in (0,1)$. By symmetry of $F$ and $G$, the function $f$ is (almost surely) symmetric as well, that is,
	\begin{equation}
		f(\tau) = 2\widetilde{m} - f(1-\tau)
		\label{Eq:Symmetryf}
	\end{equation}
	holds for almost all $\tau \in (0,1)$. Let $A = \{\tau\mid f(1-\tau) > f(\tau) \geq 0\} \cup \{\tau < \tfrac12 \mid f(1-\tau) = f(\tau) \geq 0\}$. Then, the plus shift component of the $p$-Wasserstein distance is twice the integral of $f$ over $A$ as a change of variables yields $\shift_+^{\WDp}(F,G) = \int_0^1 \left[\min\left\{f(\tau)^{[p]},f(1-\tau)^{[p]}\right\}\right]_+\dd\tau = 2\int_A \vert f(\tau)\vert^p\dd\tau$.
	Let $\overline{A} = \{1-\tau\mid \tau \in A\}$, $B = \{\tau\mid f(\tau) \geq 0 > f(1-\tau)\}$, and $\overline{B} = \{1-\tau\mid \tau \in B\}$. By the (almost sure) symmetry of $f$, the values $\tau$ such that $0 > f(\tau)$ and $0 > f(1-\tau) = 2\widetilde{m} - f(\tau)$ form a null set in $(0,1)$, and so does the complement of the disjoint union $A\cup \overline{A}\cup B \cup \overline{B}$ in $(0,1)$. Note that the symmetry \eqref{Eq:Symmetryf} of $f$ yields the inequalities
	\begin{alignat}{3}
		\widetilde{m} &\geq f(\tau) &&\quad\text{for almost all }\tau \in A  \label{Eq:SimpleA1} \\
		f(1-\tau) &\geq \widetilde{m} &&\quad\text{for almost all }\tau \in A \label{Eq:SimpleA2} \\
		f(\tau) &\geq 2\widetilde{m} &&\quad\text{for almost all }\tau \in B \label{Eq:SimpleB1}
	\end{alignat}
	We use these a.s.\ (almost sure with respect to the Lebesgue measure) inequalities to derive the inequality
	\begin{alignat*}{3}
		&\phantom{=~} \int_A \vert \underset{\mathrlap{\leq \widetilde{m}\text{ a.s.\ by }\eqref{Eq:SimpleA1}}}{\underbrace{f(\tau)}}\vert^q\dd\tau &\cdot~&\bigg(\int_A \vert f(1-\tau)\vert^{p}\dd\tau + \int_B \vert f(\tau)\vert^{p} + \underset{\mathrlap{= f(\tau) - 2\widetilde{m} \text{ a.s.\ by the symmetry \eqref{Eq:Symmetryf}}}}{\underbrace{\vert f(1-\tau)\vert}}^{p}\dd\tau \bigg) \\
		&\leq \widetilde{m}^{q-p} \int_A \vert f(\tau)\vert^{p}\dd\tau &~\cdot~\frac1{\widetilde{m}^{q-p}}&\bigg(\int_A\underset{\mathrlap{\hspace{-0.5cm}\leq f(1-\tau)^{q-p}\text{ a.s.\ by }\eqref{Eq:SimpleA2}}}{\underbrace{\widetilde{m}^{q-p}}}\vert f(1-\tau)\vert^{p}\dd\tau 
		+ \int_{B} \widetilde{m}^{q}\underset{\mathrlap{\hspace{-2.6cm}\leq \big\vert \frac{f(\tau)}{\widetilde{m}}\big\vert^{q} + \big\vert \frac{f(\tau)}{\widetilde{m}} - 2\big\vert^{q} \text{ a.s.\ by Lemma \ref{Lem:AuxIneq} and }\eqref{Eq:SimpleB1}}}{\underbrace{\Big(\Big\vert \frac{f(\tau)}{\widetilde{m}}\Big\vert^{p} + \Big\vert \frac{f(\tau)}{\widetilde{m}} - 2\Big\vert^{p}\Big)}} \dd\tau \bigg) \\
		&\leq \int_A \vert f(\tau)\vert^{p}\dd\tau &\cdot~&\bigg(\int_A \vert f(1-\tau)\vert^{q}\dd\tau + \int_{B} \vert f(\tau)\vert^{q} + \vert f(1-\tau)\vert^{q}\dd\tau \bigg),
	\end{alignat*}
	which is equivalent (as $\tau \in A \Longleftrightarrow (1-\tau) \in \overline{A}$, and analogously for $B$ and $\overline{B}$) to the inequality
	\[
	\int_A \vert f(\tau)\vert^q\dd\tau ~\cdot~\left(\int_{\overline{A}} \vert f(\tau)\vert^{p}\dd\tau + \int_{B\cup \overline{B}} \vert f(\tau)\vert^{p}\dd\tau\right)
	\leq \int_A \vert f(\tau)\vert^{p}\dd\tau ~\cdot~\left(\int_{\overline{A}} \vert f(\tau)\vert^{q}\dd\tau + \int_{B\cup \overline{B}} \vert f(\tau)\vert^{q}\dd\tau\right).
	\]
	By adding the term $\int_A \vert f(\tau)\vert^q\dd\tau \cdot\int_{A} \vert f(\tau)\vert^{p}\dd\tau$ to both sides, we end up with the inequality
	\[
	\int_A \vert f(\tau)\vert^q\dd\tau \int_0^1 \vert f(\tau)\vert^{p}\dd\tau \leq \int_A \vert f(\tau)\vert^{p}\dd\tau \int_0^1 \vert f(\tau)\vert^q\dd\tau,
	\]
	which is equivalent to
	\[
	\frac{\int_A \vert f(\tau)\vert^q\dd\tau}{\int_0^1 \vert f(\tau)\vert^q\dd\tau} \leq \frac{\int_A \vert f(\tau)\vert^{p}\dd\tau}{\int_0^1 \vert f(\tau)\vert^{p}\dd\tau}
	\]
	or
	\[
	\frac{\shift_+^{\WD_q}(F,G)}{\WD_q(F,G)} \leq \frac{\shift_+^{\WD_{p}}(F,G)}{\WD_{p}(F,G)}.
	\]
	As the relative (normalized) shift component decreases as $p$ increases, the sum of the relative (normalized) dispersion components increases with $p$, as formalized by inequality \eqref{Eq:DispInequality_pWD}.
\end{proof}

\begin{proof}[Proof of Theorem \ref{Thm:DispInequalities}]
	Let $\sigma_F \neq \sigma_G$ (otherwise all dispersion components are 0, and hence the inequalities are satisfied).
	With notation as used in the closed-form expressions given at the end of Section \ref{subsec:BasicProperties}, inequality \eqref{Eq:DispInequality_CD-WD_norm} is equivalent to
	\begin{align*}
		0 &\leq \frac{\CD(F,G)}{\disp_+^{\CD}(F,G) + \disp_-^{\CD}(F,G)} - \frac{\AV(F,G)}{\disp_+^{\AV}(F,G) + \disp_-^{\AV}(F,G)} \\
		&= \frac{2\sigmaavg\phi(\mudiff/\sigmaavg) + \mudiff(2\Phi(\mudiff/\sigmaavg) - 1) - \sqrt{2}\phi(0)(\sigma_F + \sigma_G)}{2\sigmaavg\phi(0) - \sqrt{2}\phi(0)(\sigma_F + \sigma_G)} - \frac{(2\Phi(\mudiff/\sigmadiff) - 1)\mudiff + 2\sigmadiff\phi(\mudiff/\sigmadiff)}{2\sigmadiff\phi(0)} =: f(\mudiff).
	\end{align*}
	Note that $f(0) = 0$, and hence it suffices to show that the first derivative is non-negative, i.e., $f'(\mudiff) \geq 0$ for $\mudiff > 0$. The first derivative is given by (note that $\phi'(x) = -x\phi(x)$)
	\begin{align*}
		f'(\mudiff) &= \frac{2\sigmaavg\phi'(\mudiff/\sigmaavg)/\sigmaavg + \mudiff(2\phi(\mudiff/\sigmaavg)/\sigmaavg) + (2\Phi(\mudiff/\sigmaavg) - 1)}{2\sigmaavg\phi(0) - \sqrt{2}\phi(0)(\sigma_F + \sigma_G)} - \frac{(2\Phi(\mudiff/\sigmadiff) - 1) + (2\phi(\mudiff/\sigmadiff)/\sigmadiff)\mudiff + 2\sigmadiff\phi'(\mudiff/\sigmadiff)/\sigmadiff}{2\sigmadiff\phi(0)} \\
		&= \frac{(2\Phi(\mudiff/\sigmaavg) - 1)}{2\sigmaavg\phi(0) - \sqrt{2}\phi(0)(\sigma_F + \sigma_G)} - \frac{(2\Phi(\mudiff/\sigmadiff) - 1)}{2\sigmadiff\phi(0)}
	\end{align*}
	Note that $f'(0) = 0$ and hence it suffices to show that the second derivative $f''(\mudiff)$ is non-negative for $\mudiff > 0$. The second derivative is given by
	\begin{align}
		f''(\mudiff) &=  \frac{2\phi(\mudiff/\sigmaavg)/\sigmaavg}{2\sigmaavg\phi(0) - \sqrt{2}\phi(0)(\sigma_F + \sigma_G)} - \frac{2\phi(\mudiff/\sigmadiff)/\sigmadiff}{2\sigmadiff\phi(0)} \notag \\
		&= \frac{\phi(\mudiff/\sigmaavg)\sigmadiff/\sigmaavg - \phi(\mudiff/\sigmadiff)(\sigmaavg - \frac1{\sqrt{2}}(\sigma_F + \sigma_G))/\sigmadiff}{\phi(0)\sigmadiff(\sigmaavg - \frac1{\sqrt{2}}(\sigma_F + \sigma_G))}
		\label{Eq:SecDerivative}
	\end{align}
	W.l.o.g.\ let $\sigma_F < \sigma_G = \sigma_F + \sigmadiff$. Note that $\sigmaavg - \frac1{\sqrt{2}}(\sigma_F + \sigma_G) = \sqrt{\sigma_F^2 + (\sigma_F + \sigmadiff)^2} - \frac1{\sqrt{2}}(2\sigma_F + \sigmadiff)$ is zero if $\sigmadiff = 0$ and has positive derivative 
	\begin{align*}
		\frac{\partial}{\partial\sigmadiff}(\sqrt{\sigma_F^2 + (\sigma_F + \sigmadiff)^2} - \frac1{\sqrt{2}}(2\sigma_F + \sigmadiff))
		&= \frac{\sigma_F + \sigmadiff}{\sqrt{\sigma_F^2 + (\sigma_F + \sigmadiff)^2}} - \frac1{\sqrt{2}} = \frac{1}{\sqrt{\sigma_F^2/(\sigma_F + \sigmadiff)^2 + 1}} - \frac1{\sqrt{2}}
	\end{align*}
	if $\sigmadiff > 0$.
	Hence, the denominator in expression \eqref{Eq:SecDerivative} is positive. Therefore, it suffices to show that the numerator in expression \eqref{Eq:SecDerivative} is also non-negative. As $\sigmaavg = \sqrt{\sigma_F^2 + \sigma_G^2} = \sqrt{(\sigma_F - \sigma_G)^2 + 2\sigma_F\sigma_G} > \sqrt{(\sigma_F - \sigma_G)^2}  = \sigmadiff$ and hence $\phi(\mudiff/\sigmaavg) > \phi(\mudiff/\sigmadiff)$, it suffices to show that $\sigmadiff/\sigmaavg \geq (\sigmaavg - \frac1{\sqrt{2}}(\sigma_F + \sigma_G))/\sigmadiff$ or, equivalently,
	\begin{align*}
		1&\geq\sigmaavg(\sigmaavg - \frac1{\sqrt{2}}(\sigma_F + \sigma_G))/\sigmadiff^2 \\
		&= \sqrt{(\sigma_F/\sigmadiff)^2 + (\sigma_F/\sigmadiff + 1)^2} \left(\sqrt{(\sigma_F/\sigmadiff)^2 + (\sigma_F/\sigmadiff + 1)^2} - \left(\sigma_F/\sigmadiff + (\sigma_F/\sigmadiff + 1)\right)/\sqrt{2}\right) =: g(\sigma_F/\sigmadiff).
	\end{align*}
	We can rewrite $g$ as
	\begin{align*}
		g(x) &= \sqrt{x^2 + (x + 1)^2}\left(\sqrt{x^2 + (x + 1)^2} - (x + (x + 1))/\sqrt{2}\right) \\
		&= \sqrt{(2x^2 + 2x) + 1}\left(\sqrt{(2x^2 + 2x) + 1} - \sqrt{(2x^2 + 2x) + 1/2)}\right) =: h(2x^2 + 2x).\\
	\end{align*}
	As $g(0) = h(0) = 1-1/\sqrt{2} < 1$ and $h$ has negative derivative 
	\begin{align*}
		h'(y) &= \frac{1/2}{\sqrt{y+1}}\left(\sqrt{y+1} - \sqrt{y+1/2}\right) + \sqrt{y+1}\left(\frac{1/2}{\sqrt{y+1}} - \frac{1/2}{\sqrt{y+1/2}}\right) \\
		&= 1 - \frac12\left(\frac{\sqrt{y+1/2}}{\sqrt{y+1}} + \frac{\sqrt{y+1}}{\sqrt{y+1/2}}\right) = 1 - \sqrt{\frac{y^2 + \frac32y + \frac34}{y^2 + \frac32y + \frac12}} < 0
	\end{align*}
	for $y \geq 0$ (note that the numerator of the fraction under the root in the last term is clearly larger than the denominator), $g(x) = h(2x^2 + 2x)$ is also decreasing in $x$ and, hence, the inequality $1 > g(0) > g(\sigma_F/\sigmadiff)$ holds, which finishes the proof of the theorem.
\end{proof}

\section{Derivation of connections to stochastic order relations}
\label{sec:ConnectionOrders}

\subsection{Proofs of connections to dispersive orders}

\begin{proof}[Proof of Theorem \ref{Thm:DispOrderJoint}]
	\begin{enumerate}
		\item 
		We prove the equivalence for the AVM. Results for the other distances are an immediate consequence as nonzero dispersion components coincide by Proposition \ref{prop:PositiveDispersion}.
		
		Clearly, $F \geq_\text{wD} G$ implies $\disp_-^{\AV}(F,G) = 0$ as the defining condition implies that the integrand \eqref{Eq:AVM_Disp} is zero for all coverages $\alpha$ in $(0,1)$.
		
		We prove the reverse implication by contraposition. If $F \not\geq_\text{wD} G$, then there exists a $\tau \in (\frac12,1)$ such that $F^{-1}(\tau) - F^{-1}(1-\tau) < G^{-1}(\tau) - G^{-1}(1-\tau)$. By continuity of the quantile functions there is a neighborhood $B\subset (\frac12,1)$ of $\tau$ such that the inequality $F^{-1}(\xi) - F^{-1}(1-\xi) < G^{-1}(\xi) - G^{-1}(1-\xi)$ holds for all $\xi \in B$. Substitution with $\tau = \frac{1+\alpha}2$ in the integral for the minus dispersion component (as in \eqref{eq:avm_disp_plus} with $F$ and $G$ switched) and rearranging terms yields
		\[
		\disp_-^{\AV}(F,G) \geq \int_B \left[ \left(G^{-1}\left(\tau\right) - G^{-1}\left(1-\tau\right)\right) - \left(F^{-1}\left(\tau\right) - F^{-1}\left(1-\tau\right)\right) \right]_+ \dd\tau > 0.
		\]
		
		\item The equivalence is an immediate consequence of (a) as the strict ordering $F >_\mathrm{wD} G$ is given by the two conditions $F \geq_\mathrm{wD} G$ and  $F \not\leq_\mathrm{wD} G$, which are equivalent to $\disp_-^{\D}(F,G) = 0$ and $\disp_+^{\D}(F,G) > 0$, respectively (where the latter invokes the symmetry from Proposition \ref{prop:Symmetry}).
	\end{enumerate}
\end{proof}

\begin{proof}[Proof of Proposition \ref{Prop:DispOrder}]
	\begin{enumerate}
		\item Clearly, a dispersive ordering implies \eqref{Eq:wD} for all $\tau \in (0.5,1)$.
		
		\item 
		If $F >_\text{D} G$, then there exist $0 < \xi < \tau < 1$ such that the defining inequality
		\[
		F^{-1}(\tau) - F^{-1}(\xi) > G^{-1}(\tau) - G^{-1}(\xi)
		\]
		is strict. Since the inequality
		\[
		F^{-1}(1-\xi) - F^{-1}(1-\tau) \geq G^{-1}(1-\xi) - G^{-1}(1-\tau)
		\]
		also holds (as $1-\tau < 1-\xi$), we get the strict inequality
		\begin{equation}
			\underbrace{F^{-1}(\tau) - F^{-1}(1-\tau)}_{=: f(\tau)} - \underbrace{(F^{-1}(\xi) - F^{-1}(1-\xi))}_{=: f(\xi)} > \underbrace{G^{-1}(\tau)  - G^{-1}(1-\tau)}_{=: g(\tau)} - \underbrace{(G^{-1}(\xi) - G^{-1}(1-\xi))}_{=: g(\xi)}.
			\label{Eq:ProofStrictDisp}
		\end{equation}
		We now consider a case distinction: 
		\begin{itemize}
			\item
			In the case of $\tau > \xi \geq \frac12$, inequality \eqref{Eq:ProofStrictDisp} yields
			\[f(\tau) - \underset{\mathrlap{\geq 0 \text{ by $F >_\mathrm{D} G$}}}{\underbrace{(f(\xi) - g(\xi))}} > g(\tau) \quad\Longrightarrow\quad f(\tau) > g(\tau).\]
			
			\item 
			In the case of $\tau \geq \frac12 > \xi$, either $f(\tau) > g(\tau)$ or $f(\tau) = g(\tau)$ (by the dispersive ordering $F >_\mathrm{D} G$). In the latter case, inequality \eqref{Eq:ProofStrictDisp} yields
			\[- f(\xi) > \underset{\mathrlap{= 0}}{\underbrace{g(\tau) - f(\tau)}}  - g(\xi)  \quad\Longrightarrow\quad f(1-\xi) > g(1-\xi).\]
			
			\item 
			Finally, in the case of $\frac12 > \tau > \xi$, inequality \eqref{Eq:ProofStrictDisp} yields
			\[\underset{\mathrlap{\leq 0 \text{ by $F >_\mathrm{D} G$}}}{\underbrace{f(\tau) - g(\tau)}} - f(\xi) >  - g(\xi)  \quad\Longrightarrow\quad f(1-\xi) > g(1-\xi).\]
		\end{itemize}
		Hence, in each of the three cases, there exists a $\tau \in (0.5,1)$ such that the inequality in \eqref{Eq:wD} is strict, which results in a strict weak dispersive ordering $F >_\mathrm{wD} G$.
	\end{enumerate}
\end{proof}

\subsection{Proofs of connections to stochastic orders}

\begin{proof}[Proof of Theorem~\ref{Thm:StochOrderCDJoint}]
	Arguments similar to those in the proof of Theorem \ref{Thm:DispOrderJoint} yield the desired result, as we show in the following:
	
	\begin{enumerate}[label=(\alph*)]
		\item 
		It is easy to see that $F \geq_\text{wS} G$ implies $\shift_-^{\CD}(F,G) = 0$ as the defining condition implies that the integrand in \eqref{Eq:CD_Shift} is zero for all coverage levels $\alpha$ and $\beta$ in $(0,1)$.
		
		We prove the reverse implication by contraposition. If $F \not\geq_\text{wS} G$, then there exists a pair $(\tau, \xi) \in (\frac12,1)^2$ such that $\max \big\{F^{-1}(\tau) - G^{-1}(\xi), \, F^{-1}(1-\tau) - G^{-1}(1-\xi) \big\} \; < \; 0$. By continuity of the quantile functions there is a neighborhood $B\subset (\frac12,1)^2$ of $(\tau,\xi)$ such that the above inequality holds for all pairs $(\tau,\xi) \in B$. Substitution with $\tau = \frac{1+\beta}2$ and $\xi = \frac{1+\alpha}{2}$ in the integral for the shift component of the $\CD$ (as in \eqref{Eq:CD_Shift} with $F$ and $G$ switched) yields
		\begin{align*}
			\shift_-^{\CD}(F,G) 
			&\geq 2 \int_B \Big[\underbrace{\min\left\{ G^{-1}\left(\xi\right) - F^{-1}\left(\tau\right), \; G^{-1}\left(1-\xi\right) - F^{-1}\left(1-\tau\right) \right\}}_{>0}\Big]_+ \\
			&\hspace{1cm}
			+ \left[G^{-1}\left(\xi\right) - F^{-1}\left(1-\tau\right)\right]_+ \dd \xi \dd \tau  > 0.
		\end{align*}
		Notice that the term in the lower line is non-negative such that omitting it still yields a strictly positive shift component.
		
		\item 
		In analogy to the proof of Theorem \ref{Thm:DispOrderJoint} (b), the equivalence is an immediate consequence of part (a).
	\end{enumerate}
\end{proof}

\begin{proof}[Proof of Proposition \ref{Prop:StochOrderCD}]
	\begin{enumerate}
		\item It is easy to see that stochastic ordering $F \ge_\mathrm{S} G$ implies weak stochastic ordering $F \ge_\mathrm{wS} G$: If $\tau \geq \xi$, we obtain $F^{-1}(\tau) -G^{-1}(\xi) \geq F^{-1}(\tau) -G^{-1}(\tau) \geq 0$ by the monotonicity of $G^{-1}$ and stochastic ordering. Otherwise, if $\tau < \xi$, we obtain $F^{-1}(1-\tau) -G^{-1}(1-\xi) \geq F^{-1}(1-\tau) -G^{-1}(1-\tau) \geq 0$ from the stochastic ordering constraint. Hence, condition \eqref{Eq:wS} is satisfied.
		
		\item Clearly, strict stochastic ordering implies strict weak stochastic ordering as there exists a $\tau$ such that one of the inequalities in (a) is strict.
	\end{enumerate}    
\end{proof}

\begin{proof}[Proof of Proposition~\ref{Thm:WeakStochPreorder}]
	Clearly, the weak stochastic order is reflexive.
	To show the transitivity, consider three distributions $F \geq_\text{wS} G \geq_\text{wS} H$, and arbitrary $0.5 < \tau,\xi < 1$, which are fixed throughout the proof. 
	The distribution $F$ is larger than $H$ in weak stochastic order if the following inequality is satisfied
	\begin{align*}
		0 &\leq \max\{F^{-1}(\tau) - H^{-1}(\xi), F^{-1}(1-\tau) - H^{-1}(1-\xi)\} \\
		&= \max\{\underbrace{F^{-1}(\tau) - G^{-1}(\gamma)}_{A(\gamma)} + \underbrace{G^{-1}(\gamma) - H^{-1}(\xi)}_{C(\gamma)}, \underbrace{F^{-1}(1-\tau) - G^{-1}(1-\gamma)}_{B(\gamma)} + \underbrace{G^{-1}(1-\gamma) - H^{-1}(1-\xi)}_{D(\gamma)}\},
	\end{align*}
	where $0.5 < \gamma < 1$. For any $\gamma$, at least one of the terms $C(\gamma)$ or $D(\gamma)$ is non-negative by definition of the weak stochastic ordering $G \geq_\text{wS} H$. 
	Therefore, it suffices to show that there exists a $\gamma_0 \in [0.5,1]$ such that both $A(\gamma_0)$ and $B(\gamma_0)$ are non-negative to prove the above inequality. 
	To this end, let $\gamma_0 = \sup\{\gamma \mid A(\gamma) \geq 0\}$. Note that $\gamma_0 \geq 0.5$ as $F^{-1}(\tau) \geq G^{-1}(0.5)$ by weak stochastic ordering (which, by setting $\xi =0.5$ in its definition, implies that either $F^{-1}(\tau) \geq G^{-1}(0.5)$ or $F^{-1}(1-\tau) \geq G^{-1}(0.5)$ and $F^{-1}(\tau) \ge F^{-1}(1-\tau)$ holds as $\tau > 0.5$.)
	
	We now distinguish two cases:
	First, in the case of $\gamma_0 = 1$, the equality in the centered equation above extends to $\gamma = 1$, and we obtain $A(\gamma_0) = A(1) = 0$ and $B(\gamma_0) \geq 0$ as $F^{-1}(1-\tau) \geq G^{-1}(1-\gamma_0) = G^{-1}(0)$ because of the common support.
	Second, if $\gamma_0 < 1$, assume $B(\gamma_0) < 0$. 
	Then there exists $\varepsilon > 0$ such that $B(\gamma_0 + \varepsilon) < 0$ by continuity of the quantile functions, and $A(\gamma_0 + \varepsilon) < 0$ by the definition of $\gamma_0$.
	This however contradicts $F \geq_\text{wS} G$, such that we can conclude that $B(\gamma_0) \geq 0$. As $A(\gamma_0) = 0$ holds by continuity, this concludes the proof.
\end{proof}

\begin{proof}[Proof of Theorem~\ref{Thm:StochOrderAVMJoint}]
	Arguments similar to those in the proofs of Theorems \ref{Thm:DispOrderJoint} and \ref{Thm:StochOrderCDJoint} yield the desired results.
\end{proof}

\begin{proof}[Proof of Proposition \ref{Prop:StochOrderAVM}]
	\begin{enumerate}
		\item Clearly, condition \eqref{Eq:wS} implies condition \eqref{Eq:rS}.
		\item If $\tau$ satisfies condition \eqref{Eq:sS}, then $\max\left\{G^{-1}(\tau) - F^{-1}(\tau), G^{-1}(1-\tau) - F^{-1}(1-\tau)\right\} < 0$. Hence, $F \not\leq_\mathrm{rS} G$.
	\end{enumerate}
\end{proof}

\begin{proof}[Proof of Proposition~\ref{Thm:RelaxedSOPreorder}]
	In the case of two symmetric distributions, $F$ and $G$, with continuous quantile functions, a combination of Theorem \ref{Thm:StochOrderAVMJoint} and Proposition \ref{prop:ShiftSymDist} yields that $F \geq_\text{rS} G$ is equivalent to an ordering of the medians, $F^{-1}(\frac12) \geq G^{-1}(\frac12)$, which clearly correspond to a reflexive and transitive relation.
\end{proof}

\section{List of Climate Models}

See Table \ref{tab:ListCMIP5models}.

\begin{table*}[!ht]\small
	\caption{List of CMIP5 models used in our meteorological application together with the corresponding institute(s).} 
	\label{tab:ListCMIP5models}
	%\footnotesize
	\begin{tabular}{ll}
		\toprule 
		Climate model &     Institute(s) \\
		\midrule 
		ACCESS1-0 &         Commonwealth Scientific and Industrial Research Organisation (CSIRO) and \\ 
		&                   Bureau of Meteorology, Australia \\
		BCC-CSM1-1 &        Beijing Climate Center / China Meteorological Administration, China \\
		BNU-ESM &           College of Global Change and Earth System Science, \\
		&                   Beijing Normal University, China \\
		CanESM2  &          Canadian Centre for Climate Modelling and \\
		&                   Analysis, Canada 2 National Center for Atmospheric Research (NCAR), USA \\
		CCSM4 &             National Center for Atmospheric Research (NCAR), USA \\
		CESM1-BGC &         National Center for Atmospheric Research (NCAR), USA \\
		CMCC-CM &           Centro Euro-Mediterraneo per i Cambiamenti Climatici, Italy \\
		CNRM-CM5 &          Centre National de Recherches Météorologiques /  Centre Européen de \\
		&                   Recherche et de Formation Avancée en Calcul Scientifique, France \\
		CSIRO-Mk3-6-0 &     CSIRO in collaboration with the Queensland Climate Change \\
		&                   Centre of Excellence, Australia \\
		EC-EARTH  &         EC-Earth consortium, Sweden \\
		GFDL-CM3  &         Geophysical Fluid Dynamics Laboratory, USA \\
		GFDL-ESM2G &        Geophysical Fluid Dynamics Laboratory, USA \\
		GFDL-ESM2M &        Geophysical Fluid Dynamics Laboratory, USA \\
		% !!!!! not used !!!  GISS-E2-R &        !!!!! NASA Goddard Institute for Space Studies, USA \\
		HadCM3 &            Met Office Hadley Centre, UK \\
		HadGEM2-CC &        Met Office Hadley Centre, UK \\
		HadGEM2-ES &        Met Office Hadley Centre, UK \\
		INMCM4 &            Institute for Numerical Mathematics, Russia \\
		IPSL-CM5A-LR &      Institut Pierre Simon Laplace, France \\
		IPSL-CM5A-MR &      Institut Pierre Simon Laplace, France \\
		IPSL-CM5B-LR &      Institut Pierre Simon Laplace, France \\
		MIROC4h &           Atmosphere and Ocean Research Institute (AORI), National Institute for \\
		&                   Environmental Studies (NIES) and Japan Agency for Marine-Earth Science \\
		&                   and Technology (JAMSTEC), Japan \\
		MIROC5 &            AORI, NIES and JAMSTEC, Japan \\
		MIROC-ESM  &        JAMSTEC, AORI and NIES, Japan \\
		MIROC-ESM-CHEM &    JAMSTEC, AORI and NIES, Japan \\
		MPI-ESM-LR &        Max Planck Institute for Meteorology, Germany \\
		MPI-ESM-MR &        Max Planck Institute for Meteorology, Germany \\
		MPI-ESM-P &         Max Planck Institute for Meteorology, Germany \\
		MRI-CGCM3 &         Meteorological Research Institute, Japan \\
		NorESM1-M &         Norwegian Climate Centre, Norway \\
		\bottomrule
	\end{tabular}
\end{table*}

\bibliographystylesupp{apalike}
\bibliographysupp{manuscript}

\end{document}